\documentclass{amsart}
\usepackage[utf8]{inputenc}
\usepackage{courier}
\usepackage{tikz-cd}
 \usepackage{mathtools}
 \usepackage{mathrsfs}
 \usepackage{fancyhdr}
\usepackage{amsthm}
\usepackage{etoolbox}
\usepackage{ dsfont }
\usepackage{listings}
\usepackage{color} 
\definecolor{mygreen}{RGB}{28,172,0} 
\definecolor{mylilas}{RGB}{170,55,241}

\theoremstyle{section}
\newtheorem{example}{Example}

\newtheorem{theorem}{Theorem}
\newtheorem{assumption}{Assumption}
\newtheorem{corollary}{Corollary}[theorem]
\newtheorem{lemma}[theorem]{Lemma}
\newtheorem{proposition}{Proposition}

\newtheorem{remark}{Remark}
\usepackage[british]{babel}
\usepackage[a4paper,top=3cm,bottom=3cm,left=3cm,right=3cm,headsep=11pt]{geometry}
\usepackage{amsmath,amssymb,amsthm,mathtools} 
\usepackage{color,graphicx} 
\DeclareGraphicsExtensions{.pdf,.png,.jpg,.svg} 
\usepackage{hyperref} 
\usepackage{enumerate} 
\allowdisplaybreaks

\newtheorem{innercustomthm}{Corollary}
\newenvironment{customthm}[1]
  {\renewcommand\theinnercustomthm{#1}\innercustomthm}
  {\endinnercustomthm}

\DeclarePairedDelimiter\abs{\lvert}{\rvert}%
\DeclarePairedDelimiter\norm{\lVert}{\rVert}%

\makeatletter
\let\oldabs\abs
\def\abs{\@ifstar{\oldabs}{\oldabs*}}
\let\oldnorm\norm
\def\norm{\@ifstar{\oldnorm}{\oldnorm*}}
\makeatother

\newcommand{\mcI}{\mathcal{I}}

\newcommand{\mcR}{\mathcal{R}}

\newcommand{\mcM}{\mathcal{M}}

\newcommand{\HH}{\mathbb{H}}
\newcommand{\B}{\mathbb{B}}
\newcommand{\R}{\mathbb{R}}
\newcommand{\C}{\mathbb{C}}

\newcommand{\N}{\mathbb{N}}
\newcommand{\Z}{\mathbb{Z}}
\newcommand{\K}{\mathbb{K}}
\renewcommand{\P}{\mathbb{P}}

\newcommand{\rarrow}{\rightarrow}

\newcommand{\eps}{\varepsilon}

\newcommand{\ind}{\mathds{1}}

\begin{document}
\title[Dimer-dimer correlations at the rough-smooth boundary]{Dimer-dimer correlations at the rough-smooth boundary}
\author{Kurt Johansson\textsuperscript{*} and Scott Mason\textsuperscript{$\dagger$}}

\thanks{\textsuperscript{*}Department of Mathematics, KTH Royal Institute of Technology, kurtj@kth.se}
\thanks{\textsuperscript{$\dagger$}Department of Mathematics, KTH Royal Institute of Technology, scottm@kth.se}
\thanks{Supported by the grant KAW 2015.0270 from the Knut and Alice Wallenberg Foundation.}
\maketitle
\begin{abstract}
Three phases of macroscopic domains have been seen for large but finite periodic dimer models; these are known as the frozen, rough and smooth phases.
The transition region between the frozen and rough region has received a lot of attention for the last twenty years and recently work has been underway to understand the rough-smooth transition region in the case of the two-periodic Aztec diamond. We compute uniform asymptotics for dimer-dimer correlations of the two-periodic Aztec diamond when the dimers lie in the rough-smooth transition region. These asymptotics rely on a formula found in \cite{C/J} for the inverse Kasteleyn matrix, they also apply to the infinite square grid dimer model with a variable weighting which interpolates between the rough and smooth phase \cite{K.O.S}. In particular, we find that dimers initially decay exponentially when the magnetic coordinates are very close to the bounded complementary component of the associated amoebae, they then transition to a power law decay once far enough apart.\end{abstract}

\section{Introduction}
 A dimer model of a bipartite planar graph $G$ consists of a probability measure $\P$ on the set of dimer configurations of $G$. We recall that a \emph{dimer configuration} of a graph $G$ is a set of edges in $G$ with the property that each vertex belongs to exactly one edge. We call an edge in a dimer configuration a \emph{dimer}. Dimer configurations have a simple bijection to an associated set of tilings.  The bijection is in terms of the dual graph $G^*$ of $G$;  since each dimer $e$ crosses two faces of the dual graph, let the tile corresponding to $e$ be the union of these two faces. In this paper we will focus on a graph $G$ which is the subgraph of the grid graph of $\Z^2$ satisfying $|x|+|y|\leq n$. Since $\Z^2$ is its own dual graph, this is the study of domino tilings of the Aztec Diamond introduced in \cite{EKS}.\\ One can also define a set of surfaces in $\R^3$ called height functions which are in bijection with the dimer configurations of $G$. When the graph and probability measure have a certain double periodicity, \cite{K.O.S} find that three classes of limiting Gibbs measures are possible. The three classes are referred to as \emph{frozen} (solid), \emph{rough} (liquid) and \emph{smooth} (gas). Each of the three classes has a distinct decay of correlations for distant dimers, namely fixed (or zero decay), polynomial and exponential decay. 
 
 The above classes are also related to the fluctuations and average slope of the associated random height functions. At a large but finite scale, the induced measure on height functions localises on functions with distinct domains. The limiting height function in these domains either has a smoothly varying, strictly non-zero curvature or zero curvature and a flat shape - called facets. The rough regions, in which the height functions curve, are rough in the sense that they have logarithmic fluctuations around their macroscopic shape. 
In these regions one expects the height fluctuations to be given by the Gaussian free field in an appropriate limit. The facets can have either no fluctuations (frozen) or bounded fluctuations (``smooth" but with small Poisson-like dislocations). The typical behaviour of the macroscopic height function at a rough to facet boundary point $x_0$ is expected to be like $\sim(x-x_0)^{3/2}$, see \cite{CL}, the Pokrovsky-Talapov law. This has been proven rigorously for frozen boundaries in a large class of dimer models in \cite{ADPZ}. In the context of random height functions our results relate to the correlations of the microscopic gradients of the height function at a rough-smooth boundary.
 
  Specifically, this article contains a uniform analysis of the decay of dimer correlations in different directions in a large two-periodic Aztec Diamond in the transition region between the rough and smooth phases. We study this decay in the finite model since if we let the size of the Aztec Diamond go to infinity first we only see a smooth pure phase in the transition region. Our analysis also applies to an infinite planar model with the same graph and weights induced to have the right slope of the height function. One result
 in this setting is the following theorem.
 \begin{theorem}
 Consider dimers in any large enough bounded region of the infinite grid graph (with vertices given by $\Z^2$). There are Gibbs measures in the \emph{rough phase} (or \emph{liquid phase}) for which these dimers have exponential decay of correlations.
 \begin{proof}
 See section \ref{infgraphdiscuss}.
 \end{proof}
 \label{theoremgibbsrough}
 \end{theorem}
This behavior of the correlations appears when the slope of the model is such that we are very close to a smooth phase. We will see that as we get closer to the smooth phase, the larger the corresponding region in which the dimers experience exponential decay of correlations. However for much larger distances we see the decay of correlations typically associated with the rough phase. Thus, although we are in a rough phase, if we look locally it seems like we are in a smooth phase and the rough phase only manifests itself at sufficiently long distances.
The behavior at the rough-smooth boundary in the finite two-periodic Aztec diamond is discussed informally
in Section \ref{dimerdimerexample}, and rigorously in Section \ref{corrdire1}.

 As a further general background, the interpolation between exponential and power law decay of correlations has already been observed in a variety of two-dimensional models in statistical mechanics \cite{CL}. Whilst rigorous results are limited \cite{FS}, physicists have used renormalisation group/mean field techniques along with numerical simulations to analyse the transition region of some of these models. As such, it is expected that a sharp transition occurs between the two regimes (exponential and power law decay). This is known as a Kosterlitz-Thouless phase transition, which was originally characterised by the appearance of "topological defects".  For example, in the $XY$ model the transition is associated with the appearance of vortices. The interpretation is that these vortices destroy the quasi-long range order present in the lower temperature regime giving rise to short range order and exponential decay of correlations. We discuss a kind of analogue of this and Figure \ref{AztecDiamondpics1} for the infinite dimer model at the end of section \ref{infgraphdiscuss}.

\subsection{Definition of the model} Consider the subset of $\Z^2$  given by $V = W\cup  B$, where
\begin{align*}
W=\{(i,j): i \ \text{mod } 2=1, \ j \ \text{mod } 2=0, 1
\leq i\leq2n-1, 0\leq j\leq 2n\}
\end{align*}
and
\begin{align*}
B=\{(i,j): i \ \text{mod } 2=0, \ j \ \text{mod } 2=1, 1
\leq i\leq2n, 0\leq j \leq 2n-1\}.
\end{align*}
We define the vertex set $V$ as the vertex set of the \emph{Aztec Diamond graph} $AD$ of size $n$ with corresponding edge set given by all $b-w=\pm \vec{e}_1, \pm\vec{e}_2$ for all $b\in B, w\in W$, where $\vec{e}_1=(1,1)$, $\vec{e}_2=(-1,1)$. For an Aztec Diamond of size $n=4m, m\in \N_{>0}$ define the weight as a function $w$ from the edge set into $\R_{> 0}$ such that the edges contained in the smallest cycle surrounding the point $(i,j)$ where $(i+j)$ mod $4=2$, have weight $a\in(0,\infty)$ and the edges contained in the smallest cycle surrounding the point $(i,j)$ where $(i+j)$ mod $4=0$ have weight $b\in (0,\infty)$. Each of these cycles is the boundary of a face of $AD$ and we call each of these faces an $a$ face ($b$ face) if the edges on its boundary each have weight $a$ ($b$). 
We divide the white and black vertices into two different types. For $i\in\{0,1\}$,
\begin{align*}
&B_i=\{(x_1,x_2)\in B : x_1+x_2 \text{ mod } 4 = 2i+1\},
\end{align*}
and
\begin{align*}
&W_i= \{(x_1,x_2)\in W : x_1+x_2 \text{ mod } 4=2i+1\}.
\end{align*}\\
\begin{figure}[h]
\centering
\includegraphics[width = 0.5\textwidth]{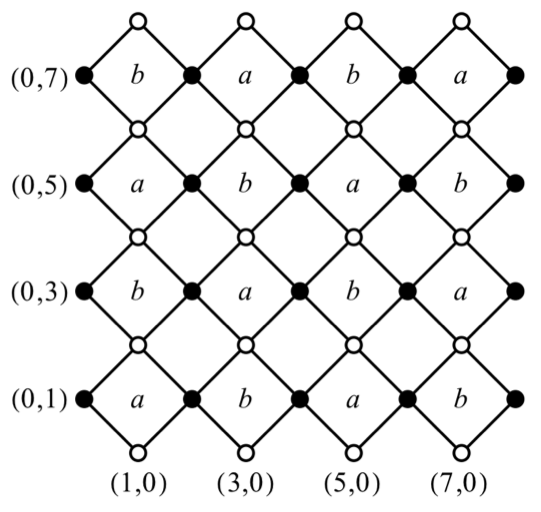}
\caption[Aztec Diamond graph]
	{The two-periodic Aztec Diamond graph for $n=4$ with the $a$ and $b$ faces labelled.}
	\label{AztecDiamondgraphn=4}
\end{figure}
Define a probability measure $\P_{Az}$ on the finite set of all dimer configurations $\mcM(AD)$ of $AD$. For a dimer configuration $\omega\in\mcM(AD)$,
\begin{align*}
\P_{Az} (\omega) = \frac{1}{Z}\prod_{e\in \omega}w(e), \qquad \text{where}  \ \ Z=\sum_{\omega\in \mcM(AD)}\prod_{e\in\omega } w(e)
\end{align*}
is the partition function and the product is over all edges in $\omega$.
We call the probability space corresponding to $\P_{Az}$ and $\mcM(AD)$ the \emph{two-periodic Aztec Diamond}, and note that the setup here is the same as in \cite{C/J}.

\subsection{Kasteleyn's approach and dimer statistics} The classical approach to analyse the statistical behaviour of random dimer configurations of large bipartite graphs $G$ is to follow an idea introduced by Kasteleyn. In this approach, one puts signs $+1,-1$ (called a Kasteleyn orientation) into a submatrix of the weighted adjacency matrix indexed by $B'\times W'$ for the black and white vertices, $B'$ and $W'$, of $G$. The resulting matrix $K$ is called the Kasteleyn matrix and has the property that the partition function of the dimer model is equal to the absolute value of the determinant of $K$. The key idea in this approach is that the signs one introduced to construct the matrix $K$ cause the non-zero terms in the sum of the determinant of $K$ to all have the same sign. In general, the Kasteleyn orientation is not unique, and its values need not be restricted to $1,-1$. Here we introduce the Kasteleyn matrix $K_{a,b}$ that we use for the two-periodic Aztec Diamond model of size $n=4m$. Define
\begin{align}
K_{a,b}(x,y)=\begin{cases} a(1-j)+bj & \text{if } y=x+\vec{e}_1, x\in B_{j}\\
i(aj+b(1-j)) & \text{if } y=x+\vec{e}_2, x\in B_{j}\\
aj+b(1-j) & \text{if } y=x-\vec{e}_1, x\in B_{j}\\
i(a(1-j)+bj) & \text{if }  y=x-\vec{e}_2,x\in B_{j}\\
0 & \text{otherwise}
\end{cases}
\end{align}
 where $i=\sqrt{-1},$ $j\in\{0,1\} $. For dimers $e_1,...,e_n$, define the $n$-point correlation function
 \begin{align}
 \rho_n(e_1,...,e_n)=\P_{Az}(\omega\in \mcM(AD): e_1,...,e_n\in \omega).
  \end{align}
One can use the determinantal expression of the partition function to show that collections of dimers form a determinantal point process. Indeed, a theorem from Kenyon \cite{K} (see also \cite{Joh}) gives that for $e_i=(b_i,w_i), \ i=1,...,n$, the $n$-point correlation functions are 
\begin{align}\label{rhocorr}
\rho_n(e_1,...,e_n)=\det(L(e_i,e_j))_{i,j=1}^n
\end{align}
with correlation kernel
\begin{align}\label{Lkernel}
L(e_i,e_j)=K_{a,b}(b_i,w_i)K_{a,b}^{-1}(w_j,b_i).
\end{align}
In the above, $K_{a,b}^{-1}(w_j,b_i)$ is the inverse of the Kasteleyn matrix $K_{a,b}$ evaluated at $(w_j,b_i)$.
Define the \emph{dimer-dimer correlation} $\text{corr}(e_1,e_2)$ of two dimers $e_1,e_2$ to be the covariance of their indicators $\ind_{e_1\in \omega}, \ind_{e_2\in \omega}$. We have
\begin{align}
\label{covarianceform}
\text{corr}(e_1,e_2)&=\mathbb{E}_{Az}[\ind_{e_1\in \omega, e_2\in \omega}]-\mathbb{E}_{Az}[\ind_{e_1\in \omega}]\mathbb{E}_{Az}[\ind_{e_2\in \omega}]\\
 &=\rho_2(e_1, e_2)-\rho_1(e_1)\rho_1(e_2)\nonumber\\
&=-K_{a,b}(b_1,w_1)K_{a,b}(b_2,w_2)K_{a,b}^{-1}(w_2,b_1)K_{a,b}^{-1}(w_1,b_2),\nonumber
\end{align}
from \eqref{rhocorr} and \eqref{Lkernel}.

\subsection{Dimer locations}
 Fix $a\in(0,1)$ and take, without loss of generality, $b=1$. It is convenient to set $c=a/(1+a^2)\in(0,1/2)$. Our goal is to take a very large $n$ and compute uniform asymptotics of the correlation of two edges separated by a growing distance but for distances much smaller than $n$. We want to investigate the decay of correlations in the transition region 
between the rough and smooth phases of the model.
 As in \cite{C/J} we only consider dimers in the the bottom left quadrant of $AD$. Let $-1<\xi<0$ so that $(n(1+\xi),n(1+\xi))$ is the coordinate varying over the diagonal of the bottom left quadrant of the Aztec Diamond. 
We will use $(n(1+\xi), n(1+\xi))$ as a reference point so we want it to have integer coordinates, i.e. we assume that 
 \begin{align*}
n(1+\xi)\in 2\N_{\geq 0}\cap [n(1-\frac{1}{2}\sqrt{1+2c},n(1-\frac{1}{2}\sqrt{1-2c})].
\end{align*}
The reference point should be close to the asymptotic rough-smooth boundary. 
Following the notation in \cite{C/J}, let $\xi_c=-\frac{1}{2}\sqrt{1-2c}$ so that $\xi=\xi_c$ corresponds to the limiting rough-smooth boundary. We will consider $\xi$ close to $\xi_c$.

The locations of the dimers can be specified as follows. Let $\ind_{a_F}$ be an indicator function on the edge set so that $\ind_{a_F}(e)=1$ if the edge $e$ has weight $a$ and $\ind_{a_F}(e)=0$ otherwise.
Take two dimers and label their positions as follows: Let $\eps_1,\eps_2\in\{0,1\}$. For $i=1,2$, the dimer $e_i$ is given by a pair of vertices $(x_{ \eps_1}^{(i)},y_{\eps_2}^{(i)})\in W_{\eps_1}\times B_{\eps_2}$ where
\begin{align}
&x_{\eps_1}^{(i)}=(n(1+\xi)+2r_1^{(i)}+1, n(1+\xi)+2r_2^{(i)}+2\eps_1^{(i)}(2\ind_{a_F}(e_i)-1))\label{xdimervertex}\\
&y_{\eps_2}^{(i)}=(n(1+\xi)+2r_1^{(i)}+1-(1-2\eps_2^{(i)})(2\ind_{a_F}(e_i)-1), n(1+\xi)+2r_2^{(i)}+2\ind_{a_F}(e_i)-1)\label{ydimervertex}
\end{align}
and $(r_1^{(i)}, r_2^{(i)})\in\{ (r_1^{(i)}, r_2^{(i)})\in \Z^2 : r_1^{(i)}+r_2^{(i)}\in2\Z  \ \text{and}  \ |2r_1^{(i)}|+|2r_2^{(i)}| \leq r\}$.  

Note that the sum of  $r_1^{(i)}$ and $r_2^{(i)}$ is always even and $(n(1+\xi)+1,n(1+\xi))$ is always a $W_0$-type vertex. So $x_0^{(i)}=(n(1+\xi)+2r_1^{(i)}+1,n(1+\xi)+2r_2^{(i)})$ generates all possible $W_0$ vertices within a distance $r$ of $(n(1+\xi)+1,n(1+\xi))$.  In \eqref{xdimervertex} and \eqref{ydimervertex} we have written the locations of the vertices $x$ and $y$ corresponding to the edge $e=(x,y)\in W\times B$ as the sum of two coordinate vectors, the first corresponds to the $W_0$ vertex in the same $a$ or $b$ face as $e$ and the second is $(0,2\eps_1^{(i)}(2\ind_{a_F}(e_i)-1))$ for the white vertex $x$ or $((-1+2\eps_2^{(i)})(2\ind_{a_F}(e_i)-1),2\ind_{a_F}(e_i)-1)$ for the black vertex $y$.

\begin{example} \rm
There are many possibilities for the location and orientation of the dimers. For our discussion of the correlations we will consider the following case.
We take two dimers $e_1=(x^{(1)},y^{(1)})$ and $e_2=(x^{(2)},y^{(2)})$ in two separate $a$ faces and both connecting the vertices of type $W_0$ and $B_0$, so $(x^{(i)},y^{(i)})\in W_0\times B_0$ for $i=1,2$. We place one along the main diagonal. Explicitly
\begin{align}
&x^{(1)}=(n(1+\xi)+1,n(1+\xi)), && y^{(1)}=(n(1+\xi),n(1+\xi)+1),\\
&x^{(2)}=(n(1+\xi)+1+2r_1,n(1+\xi)+2r_2),&&y^{(2)}=(n(1+\xi)+2r_1,n(1+\xi)+1+2r_2)
\end{align}
for $r_1,r_2\in \Z$ such that $r_1+r_2$ is even and $2(|r_1|+|r_2|)\leq r$.\label{example}
\end{example}
\begin{figure}[ht]
\centering
\includegraphics[width=0.6\textwidth]{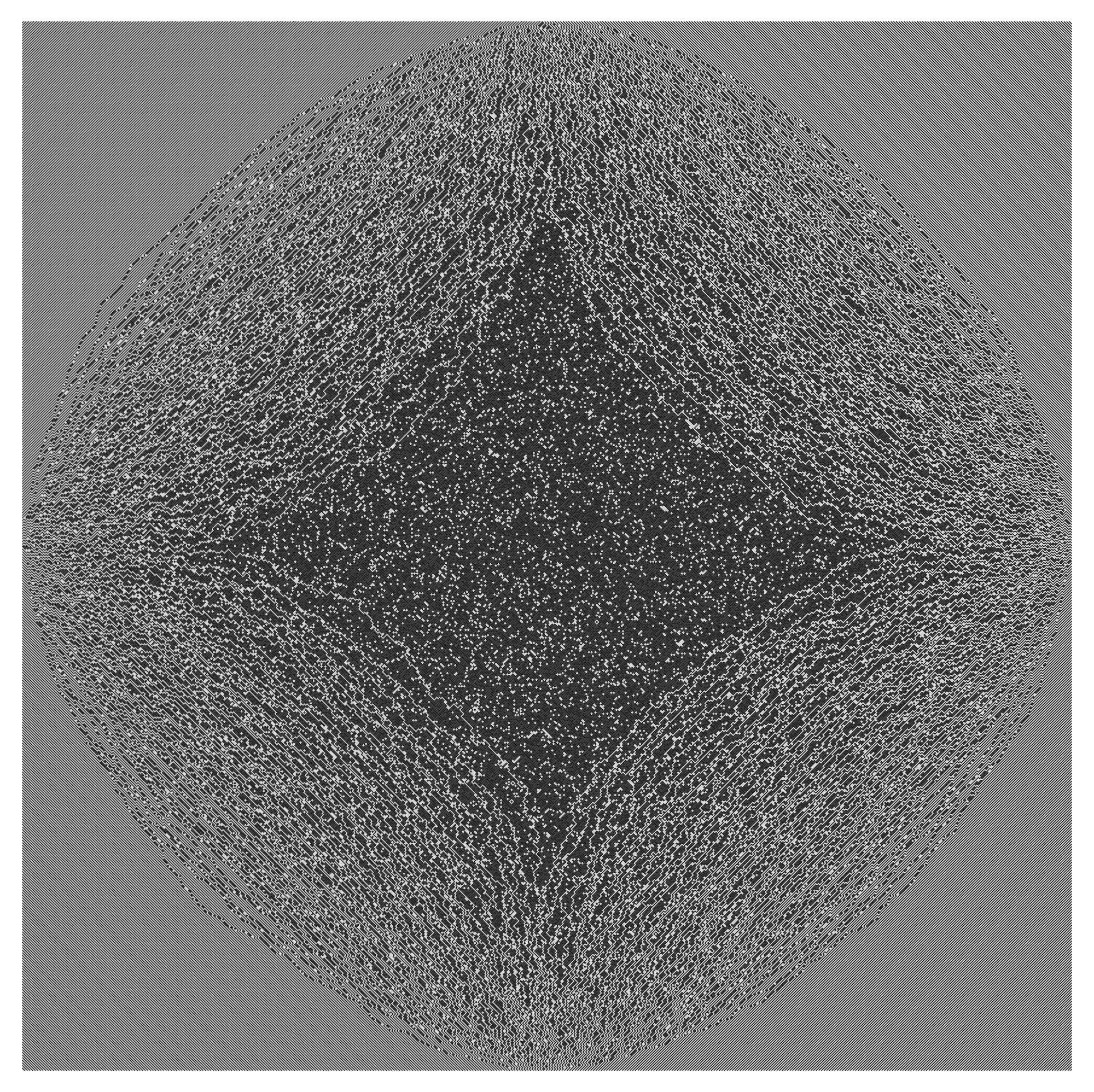}
\caption[Aztec Diamond]
	{A drawing of a sample of the two-periodic Aztec Diamond with $a=0.5$ and $b=1$. The picture consists of eight different grey-scale colours to emphasize the smooth (gas) region. The dimers with weight $a$ are coloured lighter grey than those with weight 1.}
	\label{AztecDiamondpics}
\end{figure}
\begin{figure}[ht]
\centering
\includegraphics[width=0.6\textwidth]{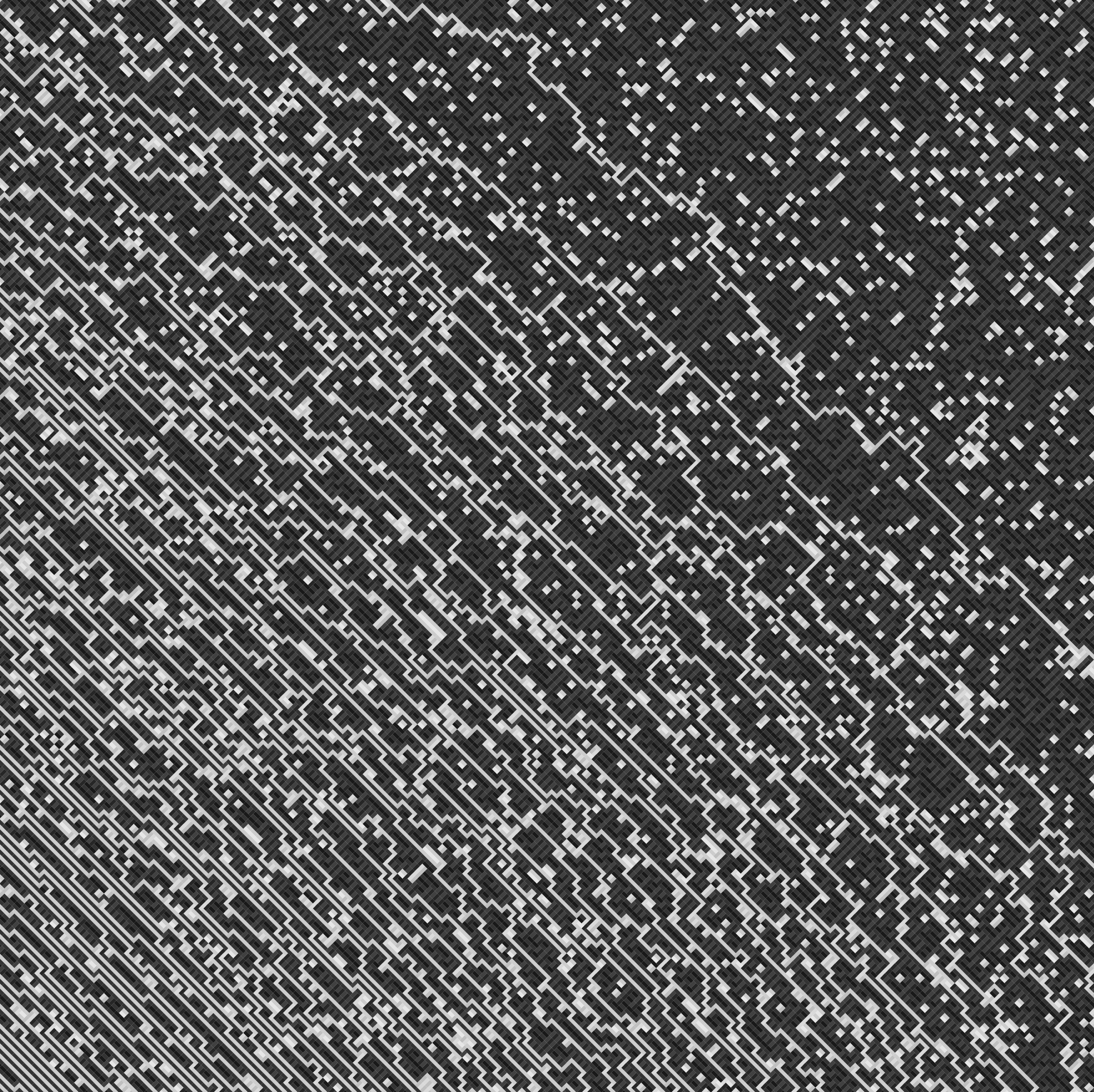}
\caption[Aztec Diamond1]
	{The rough-smooth transition region in the lower left quadrant.}
	\label{AztecDiamondpics1}
\end{figure}

In order to investigate the dimer-dimer correlation (\ref{covarianceform}) asymptotically we need asymptotic formulas for $K_{a,1}^{-1}$. To prove these asymptotic results we need a good formula for $K_{a,1}^{-1}$ as a starting point. A useful formula for $K_{a,1}^{-1}$ was first given in \cite{C/J}. Other versions of the result, not giving the inverse Kasteleyn matrix but a related particle kernel, have been derived in \cite{DK} and in \cite{BD} using a completely different approach. The rough-smooth boundary has also been investigated in \cite{BCJ1} and \cite{BCJ2} with the aim of establishing that right at the boundary we can see the Airy kernel point process.
We will use the formula for $K_{a,1}^{-1}$ derived in \cite{C/J}. Before we can give the formula we first need to define the objects that come into it.
For $\eps_1,\eps_2\in\{0,1\}$, we write
\begin{equation}
h(\eps_1,\eps_2)=\eps_1(1-\eps_2)+\eps_2(1-\eps_1).\label{HHH}
\end{equation}
In many cases the coordinates of the vertices giving the position of the dimers will enter into formulas via the following expressions. Let $(x_1,x_2)\in W_{\eps_1}$, $(y_1,y_2)\in B_{\eps_2}$ and define
\begin{align}
&k_1=\frac{x_2-y_2-1}{2}+h(\eps_1,\eps_2), && \ell_1=\frac{y_1-x_1-1}{2},\label{klh1}\\
& k_2=k_1+1-2h(\eps_1,\eps_2) , && \ell_2=\ell_1+1.
\label{klh2}
\end{align}
A basic role is played by the analytic function
\begin{align}\label{Gfunction}
G: \C \setminus i[-\sqrt{2c},\sqrt{2c}] \rarrow \C ;  \ w\mapsto \frac{1}{\sqrt{2c}}(w-\sqrt{w^2+2c})
\end{align}
with $\sqrt{w^2 +2c}=\exp(\frac{1}{2}\log(w+i\sqrt{2c})+\frac{1}{2}\log(w-i\sqrt{2c}))$ and where the logarithm takes its argument in $(-\pi/2,3\pi/2)$. Let $\sqrt{1/w^2 +2c}$  denote the previous square root evaluated at $1/w$. 
From the definition, we have the symmetries
\begin{align}
\overline{\sqrt{w^2+2c}}=\sqrt{\overline{w}^2+2c}, && -\sqrt{w^2+2c}=\sqrt{(-w)^2+2c}
\label{squarerootsymmetries}
\end{align}
which give
\begin{align}
&\overline{G(w)}=G(\overline{w}),&& G(w)=G(-w).
\label{Gsymmetries}
\end{align}

Let $k,\ell$ be non-zero integers. For $k,\ell> 0$ define
\begin{align}
E_{k,\ell}=\frac{i^{-k-\ell}}{2(1+a^2)2\pi i}\int_{\Gamma_1} \frac{dw}{w}\frac{G(w)^\ell G(1/w)^k}{\sqrt{w^2+2c}\sqrt{1/w^2+2c}},
\label{EKL1}
\end{align}
and then for all $k,\ell$ define
\begin{align}
E_{k,\ell}&=E_{|k|,|\ell|}.\label{ADPAD}
\end{align}
For vertices $(x_1,x_2)\in W_{\eps_1}$, $(y_1,y_2)\in B_{\eps_2}$, we define
\begin{align}
\K_{1,1}^{-1}(x_1,x_2,y_1,y_2)=-i^{1+h(\eps_1,\eps_2)}(a^{\eps_2}E_{k_1,\ell_1}+a^{1-\eps_2}E_{k_2,\ell_2}),
\label{K_11}
\end{align}
where we used (\ref{klh1}) and (\ref{klh2}). This is the inverse Kasteleyn matrix in an infinite smooth phase planar model, see Section \ref{SubsecGibbs}.
In the asymptotic analysis of the inverse Kasteleyn matrix $K_{a,1}^{-1}$ in \cite{C/J}, the following saddle-point function appears
\begin{align}
g_\xi(w)=\log w-\xi \log G(w)+\xi \log G(w^{-1}).
\label{gxi}
\end{align} 
We recall from \cite{C/J} that for $-\sqrt{1+2c}/2<\xi<\xi_c$, $g_\xi$ has four critical points $\pm\omega_c, \pm\overline{\omega}_c$, $\omega_c=e^{i\theta_c}$,
 $\theta_c\in[0,\pi/2]$, where $g'_\xi(\pm\omega_c)=g'_\xi(\pm\overline{\omega_c})=0$. These four points satisfy the saddle-point equation
\begin{align}
-\frac{1}{\xi}=\frac{w}{\sqrt{w^2+2c}}+\frac{1}{w\sqrt{1/w^2+2c}} \label{critptsss}
\end{align}
by \eqref{firstderiv}, and will hence depend on $\xi$.
Define
\begin{align}
\tilde{C}_{\omega_c}(k,\ell)=\frac{i^{-k-\ell}}{2(1+a^2)2\pi i}\int_{\Gamma_{\omega_c}}\frac{dw}{w}\frac{G(w)^\ell G(1/w)^k}{\sqrt{w^2+2c}\sqrt{1/w^2+2c}},
\end{align}
where the contour $\Gamma_{\omega_c}=\{e^{i\theta}:\theta\in(\theta_c,\pi-\theta_c)\}\cup \{e^{i\theta}:\theta\in (-\pi+\theta_c,-\theta_c)\}$ has positive orientation around the origin. 
For $x=(x_1,x_2)\in W_{\eps_1}$, $y=(y_1,y_2)\in B_{\eps_2}$, define
\begin{align}
C_{\omega_c}(x,y) =-i^{1+h(\eps_1,\eps_2)}\big(a^{\eps_2}\tilde{C}_{\omega_c}(k_1,\ell_1)+a^{1-\eps_2}\tilde{C}_{\omega_c}(k_2,\ell_2)\big).\label{Comegacccc}
\end{align} 
The elements of the inverse Kasteleyn matrix are given by the following theorem.
\begin{theorem}\label{ThmInvKast}
For $n=4m$, $m\in\N_{>0}$, $(x_1,x_2)\in W_{\eps_1}$, $(y_1,y_2)\in B_{\eps_2}$ with $\eps_1,\eps_2\in \{0,1\}$ then 
\begin{align}
K_{a,1}^{-1}((x_1,x_2),(y_1,y_2))&=\K_{1,1}^{-1}((x_1,x_2),(y_1,y_2))-C_{\omega_c}((x_1,x_2),(y_1,y_2))+ R_{\eps_1,\eps_2}((x_1,x_2),(y_1,y_2))\label{inverseKastelements2}\\
&\qquad\qquad+B^*_{\eps_1,\eps_2}(a,(x_1,x_2),(y_1,y_2)).\nonumber
\end{align}
\end{theorem}
We will not define $B^*_{\eps_1,\eps_2}$ here. By lemma $3.6$ in \cite{C/J}  there are positive constants $C_1,C_2$ such that $\abs{B^*_{\eps_1,\eps_2}(a,x_1,x_2,y_1,y_2)}\leq C_1e^{-C_2 n}$ when the vertices $(x_1,x_2)$, $(y_1,y_2)$ are  $\ell_1$-distance $O(\sqrt{n})$ from the line $(n(1+\xi),n(1+\xi))$, $-1<\xi<0$. This means that for the purposes of this paper it will be negligible. The theorem follows from the results in \cite{C/J}, see Section \ref{SecUniformbound}. The term $R_{\eps_1,\eps_2}$ is an error term in this paper.
We have the following bound for $R_{\eps_1,\eps_2}$.
\begin{proposition}
Let \begin{align}
&(x_1,x_2)=(n(1+\xi)+2a_1-1,n(1+\xi)+2a_2)\label{changeofcoordinates},\\ 
&(y_1,y_2)=(n(1+\xi)+2b_1,n(1+\xi)+2b_2-1)\nonumber.
\end{align}
If $\delta^*_n>0$ is a  sequence going to zero as $n\rarrow \infty$  then there exists $B,C>0$ such that for any $n$ large enough and $0< \xi_c-\xi\leq \delta^*_n$
\begin{align}
\abs{R_{\eps_1,\eps_2}(x_1,x_2,y_1,y_2)}\leq B|G(\omega_c)|^{b_1-b_2+a_2-a_1}\min(\frac{1}{n^{1/3}},\frac{1}{\sqrt{n\sqrt{\xi_c-\xi}}})\label{propstate}
\end{align}
uniformly for all $\abs{a_i},\abs{b_i}\leq \max\{n^{1/3},\sqrt{n\sqrt{\xi_c-\xi}}\}$ with $i=1,2$.
\label{uniformbound}
\end{proposition}

We see that what we need is uniform asymptotics for $\K_{1,1}^{-1}$ and $C_{\omega_c}$. The main results of this paper is the bound in Proposition \ref{uniformbound} and 
uniform asymptotics for $\K_{1,1}^{-1}$ and $C_{\omega_c}$.

\begin{theorem}\label{ThmMain} The asymptotic formulas (\ref{EklAs1}), (\ref{EklAs2}), (\ref{tempCom12}), and (\ref{tempd2wd}) below together with (\ref{K_11}) and (\ref{Comegacccc}), give uniform asymptotics for $\K_{1,1}^{-1}$ and $C_{\omega_c}$, and hence for $K_{a,1}^{-1}$ by (\ref{inverseKastelements2}).
\end{theorem}

Although these formulas give the asymptotics we need, it is not immediate to see how the the covariance (\ref{covarianceform}) behaves if we fix $n$ large and study how the covariance between two dimers in the transition region between the smooth and rough phases behaves as we increase the distance between them. We will discuss this in the setting of Example \ref{example} in Section \ref{SecDiscussion}. Here we give a more informal summary.

\subsection{Discussion of dimer-dimer correlation in example 1.}\label{dimerdimerexample}
We first formulate a corollary of the analysis for the covariance in the setting of Example \ref{example} with the dimers separated along the diagonal. The proof of this corollary is found in subsection \ref{corrdire1} and the discussion contained therein.
\begin{corollary}
Let $\delta^*_n>0$ be any sequence going to zero with $n$ and take $n$ so large that $\delta^*_n>0$ is sufficiently small. Assume that $\xi$ satisfies  $n\delta^*_n>n(\xi_c-\xi)>>n^{1/3}$. Let the distance between the two dimers $r=2\sqrt{r_1^2+r_2^2}$ lie in an interval $(r_{\text{min}},r_{n})$ where $r_{\text{min}}>0$ is large but fixed and $r_{n}$ is very large (growing with $n$) but less than $\sqrt{n\sqrt{\xi_c-\xi}}$.
Take the direction pointing from one of the dimers to the other to be parallel to the main diagonal, that is, parallel to the vector $(1,1)$. 

There are constants $c_1,c_2,c_3>0$ such that as $r$ varies from  $r_{\text{min}}$ to $c_1\log(\frac{1}{\sqrt{\xi_c-\xi}})$, the correlation 
\begin{align}
\text{$\text{corr}(e_1,e_2)$ is exponentially decaying in $r$},
\end{align}
for $r$ from $c_1\log(\frac{1}{\sqrt{\xi_c-\xi}})$ to $c_2\frac{1}{\sqrt{\xi_c-\xi}}$, we see a transition to
\begin{align}\label{Exnodecay}
\text{corr}(e_1,e_2)\,\, \text{is constant} \text{ or "no decay"},
\end{align}
for $r$ from $c_2\frac{1}{\sqrt{\xi_c-\xi}}$ to $r_n$ we see a transition to
\begin{align}
\text{corr}(e_1,e_2)\,\, \text{has decay like $1/r^2$ (with oscillations)}.
\end{align}
The oscillations have a long period compared to the lattice spacing, the constant in \eqref{Exnodecay} is $\propto \xi_c-\xi$.
\end{corollary}

The sign of the correlation changes depending on which type of dimers we pick. Here is some intuition behind the results. Think of the above case when the direction pointing from one dimer to the other is parallel to the main diagonal. There are two length scales in the problem. One is the lattice spacing and the other is the typical distance $1/\sqrt{\xi_c-\xi}$ between the paths that we see in Figure \ref{AztecDiamondpics1}. These are the long (corridor) paths which connect sides of the Aztec diamond defined in \cite{BCJ2}. As we increase the distance $r$ between the dimers we first see the smooth phase exponential decay which takes place in the order of the lattice spacing. After that the correlations in some sense come from the paths, similar to what should happen at the rough-frozen boundary. In terms of the distance between the paths the lattice spacing is very short and hence we see constant correlations proportional to $\xi_c-\xi$. At some point when the distance $r$ is of order $1/\sqrt{\xi_c-\xi}$, $r\sim d/\sqrt{\xi_c-\xi}$, we start to see the type of decay that we have in a rough phase, the correlations decay like $\frac{\sin^2((\xi_c-\xi)d}{d^2}$. 

We can also consider the decay of the covariance in other directions, for example: take the direction pointing from one dimer to the other to be parallel to the anti-diagonal, i.e the vector $(-1,1)$. Let $0<m<1$ be small and $M>1$ large. As $r$ varies from $r_{\text{min}}$ to $c_3\log(\frac{1}{\sqrt{\xi_c-\xi}})$, 
\begin{align}
\text{corr}(e_1,e_2)\,\,\text{is exponentially decaying in $r$},
\end{align}
 then as $r$ varies from $\log(\frac{1}{\sqrt{\xi_c-\xi}})$ to $m/(\xi_c-\xi)$,
\begin{align}
\text{corr}(e_1,e_2)\,\, \text{ decays like $1/\sqrt{r}$},
\end{align}
then as $r$ varies from $m/(\xi_c-\xi)$ to $M/(\xi_c-\xi)$ we see a transition to
\begin{align}
\text{corr}(e_1,e_2)\,\, \text{has $1/r$ decay},
\end{align}
and then as $r$ varies from $M/(\xi_c-\xi)$ to $r_n$ we see a transition to
\begin{align}
\text{corr}(e_1,e_2)\,\, \text{decays like $1/r^2$ (with no oscillations)}.
\end{align}
For a discussion on arbitrary directions, see Section \ref{SecDiscussion}. 

\subsection{Gibbs measures and the infinite planar graph}\label{SubsecGibbs}
\begin{figure}[h]
\centering
\includegraphics[width = 0.4\textwidth]{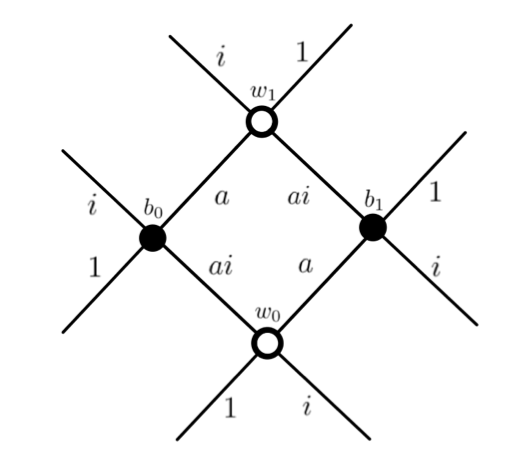}
\includegraphics[width = 0.4\textwidth]{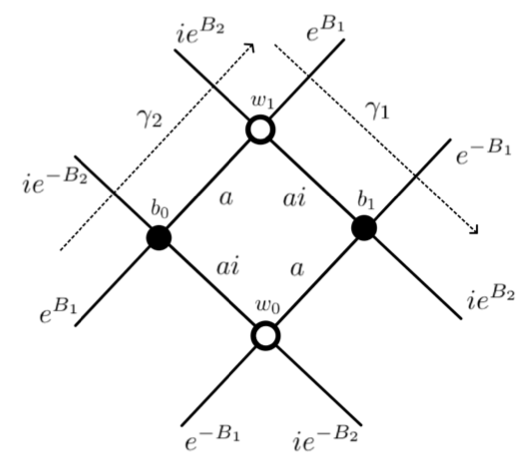}
\caption[FD]
	{Left: the fundamental domain $G^1$ with corresponding Kasteleyn matrix entries. Right: the fundamental domain of $G^1_{e^{B_1},e^{B_2}}$ with paths $\gamma_1,\gamma_2$.}
	\label{FundamentalDomainMagnetAltered}
\end{figure}

In \cite{K.O.S}, the authors describe the set of all translation invariant Gibbs measures of the infinite bipartite planar dimer models with a certain double periodicity. 
They also give an approach to compute the full plane inverse Kasteleyn matrix. The infinite planar graph $G=\tilde{V}\cup\tilde{E}$ relevant for the two-periodic Aztec diamond is the infinite version of what we defined above, and is defined as follows:
 for $i\in\{0,1\}$, let
\begin{align}
\tilde{B}_i=\{(x,y)\in \Z^2: x \text{ mod }2=0, y \text{ mod }2 =1, x+y \text{ mod } 4=2i+1\},\nonumber\\
\tilde{W}_i=\{(x,y)\in \Z^2: x \text{ mod }2=1, y \text{ mod }2 =0, x+y \text{ mod } 4=2i+1\}\nonumber
\end{align}
where $\tilde{B}=\tilde{B}_0\cup \tilde{B}_1$ are the black vertices, $\tilde{W}=\tilde{W}_0\cup \tilde{W}_1$ are the white vertices and $\tilde{V}=\tilde{W}\cup \tilde{B}$. The edge set $\tilde{E}$ is all edges of the form $b-w=\pm \vec{e}_1+,\pm\vec{e}_2$, for all $b\in \tilde{B},w\in \tilde{W}$. Let the edges contained in the smallest cycle surrounding the point $(i,j)$ where $(i+j)$ mod $4=2$, have weight $a\in(0,\infty)$ and the edges contained in the smallest cycle surrounding the point $(i,j)$ where $(i+j)$ mod $4=0$ have weight $1$.

As a reference for the following we note the unpublished Lectures on the Dimer model, Toninelli, F. 
The graph $G$ (together with its weights) is doubly periodic in the sense that the set $T$ of shifts of the form $n\vec{e_1}+m\vec{e_2}$, $(n,m)\in (2\Z)^2$ preserve the colour of the vertices and edge weights. Define the graph $G_1$ to be $G$ mod the same shifts so that $G$ is a collection of copies of $G_1$ obtained by applying all shifts $T$. From its definition, $G_1$ has periodic boundary conditions in both directions $\vec{e}_1$ and $\vec{e}_2$, so one can also think of $G_1$ as a graph on the torus. Define the edge weights and Kasteleyn orientation of $G_1$ as in figure \ref{FundamentalDomainMagnetAltered}, which induces a Kasteleyn orientation on all of $G$. The graph $G_1$ is called the fundamental domain of $G$, note that it has vertices $w_0\in \tilde{W}_0$, $b_0=w_0+\vec{e}_2$, $b_1=w_0+\vec{e}_1$, $w_1=b_0+\vec{e}_1$.

The authors of \cite{K.O.S} introduce "magnetic coordinates" to parametrise the set of Gibbs measures.  Following \cite{K.O.S}, let $(B_1,B_2)\in \R^2$ be the magnetic coordinates for the graph $G$ and take two paths $\gamma_1$, $\gamma_2$ in each fundamental domain as shown in Figure \ref{FundamentalDomainMagnetAltered}. Multiply the edge weight of an edge $e=(b,w)$ by $e^{B_1}$ ($e^{-B_1}$) if $\gamma_1$ crosses $e$ such that the white vertex $w$ is on its right (left). Similarly multiply the edge weight of $e$ by $e^{B_2}$ ($e^{-B_2}$) if $\gamma_2$ crosses $e$ such that $w$ is on its right (left). This yields the same graph with a different set of edge weights denoted $G_{e^{B_1},e^{B_2}}$. Define the fundamental domain of $G_{e^{B_1},e^{B_2}}$ to be $G^1_{e^{B_1},e^{B_2}}$ as shown in the Figure \ref{FundamentalDomainMagnetAltered}. The magnetic coordinates introduced in this way re-weight the average slope of the corresponding height functions.  Note that $B_1=B_2=0$ gives zero average slope in directions $\vec{e}_1$, $\vec{e}_2$ and gives the limiting smooth phase in the two-periodic Aztec diamond we are considering.

For $L\in 2\N_{>0}$, define $G^L=(V_L,E_L)$ to be the graph obtained by applying all shifts of the form $n\vec{e}_1+m\vec{e}_2, (n,m)\in (2\Z\cap [-L,L])^2$ to $G^1$ and with periodic boundary conditions,  observe that $G=\cup_L G^L$. If $B_L$ are the black vertices and $W_L$ are the white vertices so that $V_L=B_L\cup W_L$ then viewing the Kasteleyn matrix $K^L$ of the graph $G^L$ as an operator $\C^{W_L}\rarrow \C^{B_L}$, \cite{K.O.S} block diagonalise $K^L$ (along with three slightly modified variants of $K^L$). The authors of \cite{K.O.S} then use an extension of Kasteleyn's theory for dimer models on the torus to perform a limiting argument in $L$ which yields a double contour integral formula for the free energy (per fundamental domain) of $G$, and posit a formula for the inverse Kasteleyn matrix.  Following \cite{K.O.S}, denote the magnetically altered Kasteleyn matrix for the fundamental domain $G_1$ by $K(z,w)$ where one multiplies the edge weight of an edge $e=(b,w)$ by $z$ (or $1/z$) if $\gamma_1$ crosses $e$ with the white vertex on its right (or left). Likewise multiply the edge weight of $e$ by $w$ (or $1/w$) if $\gamma_2$ crosses $e$ such that $w$ is on its right (or left). So for $i,j\in\{0,1\}$ we obtain
\begin{align}
\big( K(z,w)\big)_{b_{i},w_{j}}=\begin{pmatrix}
i(a+1/w) & a+z\\
a+1/z & i(a+w)
\end{pmatrix}_{i+1,j+1}.
\end{align}
The magnetically altered Kasteleyn matrix in magnetic coordinates $(B_1,B_2)$ corresponding to $G^1_{e^{B_1},e^{B_2}}$ is $K(e^{B_1}z,e^{B_2}w)$. 

Suppose that $x\in \tilde{W}_{\eps_1}$ and $y\in \tilde{B}_{\eps_2}$, $\eps_1,\eps_2\in\{0,1\}$. Let $(u,v)\in \Z^2$ be such that $u (2\vec{e}_1)+v(2\vec{e}_2)$ is the translation to get from the fundamental domain containing $x$ to the fundamental domain containing $y$. The whole plane inverse Kasteleyn matrix of $G_{s_1,s_2}$ for the entries $x,y$ in magnetic coordinates $(\log(s_1),\log(s_2))$ is
\begin{align}
\K_{s_1,s_2}^{-1}(x,y)=\frac{s_1^{-u+1}s_2^{-v+1}}{(2\pi i)^2}\int_{\Gamma_{s_1}}\frac{dz}{z}\int_{\Gamma_{s_2}}\frac{dw}{w}\frac{Q(z,w)_{\eps_1+1,\eps_2+1}}{P(z,w)}z^uw^v
\label{infinvKasteleyn}
\end{align}
where $\Gamma_r$ is a circle or radius $r$ around the origin,
\begin{align}
Q(z,w)=\begin{pmatrix}
i(a+w) & -(a+z)\\
-(a+1/z) & i(a+1/w)
\end{pmatrix}
\end{align}
and 
\begin{align}
P(z,w)=-2-2a^2-\frac{a}{w}-aw-\frac{a}{z}-az
\end{align}
is called the characteristic polynomial. Note that $K(z,w)^{-1}=Q(z,w)/P(z,w)$. The formula for $\K_{s_1,s_2}^{-1}$ differs compared to the one in \cite{C/J} by multiplication of $s_1^{-u+1}s_2^{-v+1}$.

If we set $s_1=s_2=1$ in (\ref{infinvKasteleyn}) we get $\K_{1,1}^{-1}$ which agrees with $\K_{1,1}^{-1}$ given by (\ref{K_11}) as shown in \cite{C/J}.
We summarise a derivation of \eqref{K_11} from \eqref{infinvKasteleyn}. Make the change of variables $z=-u_1u_2$, $w=u_2/u_1$ for $(u_1,u_2)\in \Gamma_1^2$ in the double contour integral \eqref{infinvKasteleyn} so that the characteristic polynomial is $P(z,w)=-2(1+a^2)+a(u_1-1/u_1)(u_2-1/u_2)$. Then perform a small deformation in the $u_1,u_2$ variables so that $(u_1,u_2)\in \Gamma_R^2$ for $R<1$ very close to $1$ (avoiding the zeros of $P$). Then observe that $u\rarrow \sqrt{\frac{c}{2}}(u-1/u)$ is a bijection from $\{u;|u|<1\}$ to $\C\setminus i[-\sqrt{2c},\sqrt{2c}]$ with inverse $w\rarrow G(w)$. Then if we make another change of variables $u_i=G(w_i)$ for $i=1,2$ the characteristic polynomial becomes $-2(1+a^2)(1-w_1w_2)$.  Under the condition that $k,\ell\geq 0$, one obtains a single contour integral from the pole $w_1w_2=1$. One then checks that this single contour integral formula agrees with  \eqref{K_11} in each of the cases of vertices $\eps_1,\eps_2\in\{0,1\}$. The formula holds for all $k,\ell$ by a symmetry argument. For more information see \cite{C/J}.

In fact the whole plane inverse Kasteleyn matrix (\ref{infinvKasteleyn}) can be related to the quantities whose asymptotics we analyze. We bring forward lemma 3.3 in \cite{C/J}.
\begin{lemma}
For $x=(x_1,x_2)\in W_{\eps_1}, y=(y_1,y_2)\in B_{\eps_2}$, and $s_1=1, s_2=1/|G(\omega_c)|^2$ we have
\begin{align}
 \K_{s_1,s_2}^{-1}(x,y)=|G(\omega_c)|^{2v-2}\Big [\K_{1,1}^{-1}(x,y)-C_{\omega_c}(x,y)\Big],
 \label{Kr1r2K11Comegac}
\end{align}
where $(u,v)\in \Z^2$ is such that $u(2\vec{e}_1)+v(2\vec{e}_2)$ is the translation to get from the fundamental domain containing $x$ to the fundamental domain containing $y$.
\end{lemma}

This means that Propositions \ref{propEKL}, \ref{Comegacinnerdisc} and \ref{propeklcomegac} also yield uniform asymptotics for the inverse Kasteleyn matrix for the infinite planar dimer model $G_{s_1,s_2}$. 

\subsection{Discussion of dimer-dimer correlations on the infinite planar graph}\label{infgraphdiscuss}
Using \eqref{infinvKasteleyn} we can reformulate the discussion of dimer to dimer correlations in section \ref{dimerdimerexample} in terms of the infinite planar graph $G_{s_1,s_2}$. Instead of varying the distance $n(\xi_c-\xi)$ of the two dimers to rough-smooth boundary, we can vary the magnetic coordinates in such a way that the Gibbs measure of the infinite dimer model varies between rough and smooth phases. Note that above we fixed $n$ and chose $\xi$ depending on $n$ so the $n$-dependence now sits in the magnetic coordinates
$(\log(s_1),\log(s_2))=(0,\log(1/|G(\omega_c)|^2))$. Using lemma \ref{thetacxi} one can compute $1/|G(\omega_c)|^2=1/|G(i)|^2+c_4(\xi_c-\xi)+O(\xi_c-\xi)^{3/2}$ for some $c_4>0$ depending on $a$, so we have
\begin{align}
s_1=1 &&s_2=1/|G(\omega_c)|^2=1/|G(i)|^2+c_4(\xi_c-\xi)+O(\xi_c-\xi)^{3/2},
\end{align}
where $1/|G(i)|^{2}=2a/(\sqrt{1+a^2}-1+a)^2>1$. We see that varying $\xi$ corresponds to varying $s_2$ while keeping $s_1$ fixed.

Take two dimers, $e_i=(x^{(i)},y^{(i)})\in \tilde{W}_0\times \tilde{B}_0$, $i=1,2$ with $e_1$ placed arbitrarily. Define the coordinates of $e_2$ as
\begin{align}
x^{(2)}=x^{(1)}+2(r_1,r_2), && y^{(2)}=y^{(1)}+2(r_1,r_2)
\end{align}
where $r_1,r_2\in \Z$ and $r_1+r_2$ is even. Let $\xi_c-\xi>0$ be sufficiently small and let the distance  between the two dimers $r=2\sqrt{r_1^2+r_2^2}$ lie in $[r_{\min},r_n)$ ($r_{\min}$ is large but fixed and $r_n=\infty$). Note that this allows for any direction between the dimers. We see that the discussion in Section \ref{dimerdimerexample}, and more precisely in Section \ref{SecDiscussion}, of the dimer-dimer correlation between the dimers $e_1,e_2$ translates into this setting. 

The results of this paper in conjunction with \cite{K.O.S} show that for the collection of weights $s_1=e^{B_1}=1$,  $1\leq s_2=e^{B_2} \leq 1/|G(i)|^2$ the corresponding points $(B_1,B_2)$ lie in the bounded complementary component of the Amoeba of $P(z,w)$, which corresponds to a single unique (smooth) Gibbs measure with an average slope of $(0,0)$. When $\varphi_c>0$ but small and $s_1=1, s_2=1/|G(ie^{-i\varphi_c})|^2$ the corresponding points $(B_1,B_2)$ lie in the interior of the Amoeba of $P(z,w)$ and for each $\varphi_c>0$ there is a unique (rough) Gibbs measure with average slope $t \vec{e}_1$ for some $t\in(0,1]$. From this we see Theorem \ref{theoremgibbsrough} holds.

An interpretation is that when we have a non-zero average slope, there are infinitely long paths related to the level curves of the height function that are present. These infinite paths appear to imply quasi-long range order and power law decay. As we transition to the smooth phase these infinite paths typically become further and further apart, corresponding to a decreasing average slope of the height function. The length scales that are much smaller than the typical distance between these infinite paths experience exponential decay of correlations. If we move to the smooth phase the paths move infinitely far apart, disappearing and leaving only short range order.

\subsection*{Acknowledgements} We thank Craig Tracy for raising the question of how to understand the change from exponential to power law decay for dimer correlations
in the two-periodic Aztec  diamond at a meeting in Oberwolfach several years ago. We also thank Sunil Chhita for providing code to draw the figures \ref{AztecDiamondpics} and \ref{AztecDiamondpics1} and for helpful comments on the introduction.

\section{Asymptotic results}

In this section we formulate the precise asymptotic results for $E_{k,\ell}$ and $\tilde{C}_{\omega_c}(\ell,k)$.
The critical point $\omega_c$ defined above is a function of $\xi$ and lies on the unit circle so we can write $\omega_c(\xi)=e^{i\theta_c(\xi)}$, where $\theta_c(\xi)\in (0,\pi/2)$. It  satisfies the equation \eqref{critptsss}.  For convenience, we also define the function  
 \begin{align}
 \varphi_c(\xi)=\pi/2-\theta_c(\xi)
 \end{align} 
 since this will be a natural quantity that is small and can be related to $\xi_c-\xi$, see Lemma \ref{thetacxi}.
 Note that when $\xi=\xi_c$, $g_\xi$ has a double critical point at both $ i$ and $-i$ so $\theta(\xi_c)=\pi/2$ and hence $\varphi(\xi_c)=0$.
 
 We need to assume that $\xi_c-\xi$ is small, which means that we are close to the asymptotic rough-smooth boundary.
 \begin{assumption} We assume throughout the paper that $\xi$ satisfies $0<\xi_c-\xi<\delta$, where $\delta<\xi_c+\frac 12\sqrt{1+2c}$ is a small number determined by a finite number of conditions below. We will not discuss the explicit value of $\delta$.
\label{gxiassump}
\end{assumption}
We now relate the critical point $\omega_c$ and $\varphi_c$ to the value of $\xi$ near $\xi_c$.

\begin{lemma} Let $\xi$ be such that Assumption \ref{gxiassump} holds and let $\omega_c=e^{i\theta_c}$, $\theta_c=\pi/2-\varphi_c\in[0,\pi/2]$ satisfy $g'_\xi(\omega_c)=0$. Then there exists a bounded function $R_1(\varphi_c)$ such that
\begin{align}
\xi_c-\xi=\frac{4c(1+c)\xi_c^2}{(1-2c)^{5/2}}\varphi_c^2+R_1(\varphi_c)\varphi_c^4.
\label{xiintermsofthetac}
\end{align}
Furthermore, there are bounded functions $R_2(\xi), R_3(\xi)$ such that 
\begin{align}
\omega_c-i=\sqrt{\frac{(1-2c)^{5/2}}{4\xi_c^2c(1+c)}}\sqrt{\xi_c-\xi}+R_2(\xi)(\xi_c-\xi),
\label{omegacintermsofxi}
\end{align}
and
\begin{align}
\varphi_c=\sqrt{\frac{(1-2c)^{5/2}}{4\xi_c^2c(1+c)}}\sqrt{\xi_c-\xi}+R_3(\xi)(\xi_c-\xi).
\label{thetacintermsofxi}
\end{align}
\label{thetacxi}
\end{lemma}
The proof of this lemma is given in Section \ref{proofxiomegac}. Note that Assumption \ref{gxiassump} and Lemma \ref{thetacxi} imply that there is a $\delta_1>0$ such that $0<\varphi_c<\delta_1$ for all of the $\xi$ we consider.

For the asymptotic analysis of $E_{k,\ell}$ we need another saddle-point function.
For $\alpha\in [-1,1]$, define the function $\tilde{g}_\alpha: \C\setminus (i(-\infty,1/\sqrt{2c}]\cup i[-\sqrt{2c},\sqrt{2c}]\cup i[\sqrt{2c},\infty))\rarrow \C$ by
\begin{align}
\widetilde{g}_{\alpha}(w)=\log G(w)+\alpha \log G(1/w),
\end{align}
where both logarithms have branch cuts on the negative real axis. Taking a derivative yields
\begin{align}
\widetilde{g}'_\alpha (w) = \frac{-1}{\sqrt{w^2+2c}}+\frac{\alpha}{w^2}\frac{1}{\sqrt{1/w^2+2c}},
\label{gasderiv}
\end{align}
and setting this equal to zero gives the equation $2cw^4+(1-\alpha^2)w^2-2c\alpha^2=0$. Solving this gives the four critical points 
\begin{align}
\pm i\left(\frac{1-\alpha^2+\sqrt{(1-\alpha^2)^2+16c^2\alpha^2)}}{4c}\right)^{1/2}, \quad
\pm \left(\frac{-(1-\alpha^2)+\sqrt{(1-\alpha^2)^2+16c^2\alpha^2)}}{4c}\right)^{1/2}.\label{walpha1}
\end{align} It is straight forward to see that the four points solve $\widetilde{g}'_\alpha(w)=0$ if $0<\alpha\leq1$ (and that there is no solution for $-1\leq\alpha\leq 0$) by using \eqref{squarerootsymmetries}.
Let
\begin{align}
w_\alpha = \left(\frac{1-\alpha^2+\sqrt{(1-\alpha^2)^2+16c^2\alpha^2)}}{4c}\right)^{1/2} \in [1,1/\sqrt{2c})
\end{align}
so that $iw_\alpha$ is a critical point of $\tilde{g}_\alpha$.

Note that we have the properties
 \begin{align}
 E_{k,\ell}=E_{\ell,k},&& \tilde{C}_{\omega_c}(k,\ell)=\tilde{C}_{\omega_c}(\ell,k), \label{symmeklcomegac}
 \end{align} 
seen by the change of variables $w\rarrow 1/w$ in the definitions.
By (\ref{symmeklcomegac}) we can assume that $|k/\ell| \leq 1$ without loss of generality. 
We are now in position to formulate our main asymptotic results.

\begin{proposition}
Let the integers $\ell,k$ be such that $\ell+k$ is even and $\alpha=\abs{k}/\abs{\ell}$ lie in a compact subset of $(0,1]$. As $|\ell|\rarrow \infty$ there is a positive constant $d_1$ and a bounded function $R_4(k,\ell)$ such that
\begin{align}\label{EklAs1}
E_{k,\ell} = \frac{(-1)^{k}|G(iw_\alpha)|^{\abs{\ell}}|G(1/(iw_\alpha))|^{\abs{k}}}{(1+a^2)\sqrt{2\pi \abs{\tilde{g}''_\alpha(iw_\alpha)\ell}} w_\alpha\sqrt{w_\alpha^2-2c}\sqrt{1/w_\alpha^2-2c}}\left (1+R_4(k,\ell)/|\ell|)\right ).
\end{align}
For fixed $k$,  there is a positive constant $d_2$ and a bounded function $R_5(k,\ell)$ such that
\begin{align}\label{EklAs2}
E_{k,\ell}=\frac{(-1)^{k}|G(\frac{i}{\sqrt{2c}})|^{|\ell|}|G(-i\sqrt{2c})|^{|k|}(1-a^2)^4(1+a)}{(2a(1+a^2)^9)^{1/4}2\sqrt{\pi |\ell|}}\Big (1+R_5(k,\ell)/|\ell|^{3/2})\Big) 
\end{align}
 as $|\ell|\rarrow \infty$.
 \label{propEKL}
\end{proposition}

Given a small $\eps>0$, we define $\Lambda_\eps$ to be the set of real numbers where we stay at least $\eps$ away from all $n\pi$, $n\neq 0$, i.e.
\begin{equation}\label{Lambdaeps}
\Lambda_\eps=\{x\in\mathbb{R}\,;\,|x-n\pi|\ge\eps, \text{all $n\in\mathbb{Z}\setminus\{0\}$}\}.
\end{equation}
We use the convention that $\sin(bx)/x=b$ when $x=0$.
\begin{proposition}
\label{Comegacinnerdisc}
Fix $\gamma\in(1/2,1]$ and $\eps>0$ small. Consider non-zero integers $\ell,k$ such that $\ell+k$ is even, $\alpha=k/\ell$ lies in $[-1,1]$ and $\varphi_c^{2-\gamma}(|\ell|+|k|)\leq 1$. if $\varphi_c\ell(1-\alpha)/\sqrt{1-2c}\in \Lambda_\eps$, then
\begin{align}
\tilde{C}_{\omega_c}(k,\ell)=\frac{(-1)^{k}|G(i)|^{\ell+k}}{(1-a)^2\pi}\Big (\frac{\sin((\ell-k)\varphi_c/\sqrt{1-2c})}{(\ell-k)/\sqrt{1-2c}}(1+R_{6}(\xi,\ell,k))\Big )\label{tempCom12}
\end{align}  
where we have the estimate
\begin{align}
|R_{6}(\xi,\ell,k)|\leq C\varphi_c^{-1+2\gamma},\label{temp2zmhh}
\end{align} 
for some constant $C>0$ that only depends on $\eps$.
\end{proposition}
Set $c'=c/(1-2c)^{3/2}$. For $\theta \in \R$ define
\begin{align}
F(\theta):=\arg G(ie^{-i\theta})-\pi/2=\frac{\theta}{\sqrt{1-2c}}-R_{7}(-\theta)\theta^3.\label{Ftheta}
\end{align}
It follows from \eqref{tempImpsi12} that $R_{7}$ is bounded.

\begin{proposition}\label{propeklcomegac}
Fix $\tilde{\alpha}$ in $[-1,1]$, $\gamma\in(1/2,1]$ and $\eps>0$ small. Consider non-zero integers $\ell,k$ such that 
$\ell+k$ is even, $k/\ell=\tilde{\alpha}+\kappa_\ell\in[-1,1]$ and $(|\ell|+|k|)\varphi_c^{2-\gamma}\geq 1$. Here $\kappa_\ell$ is an arbitrary sequence indexed by $\ell$ such that $|\kappa_\ell|<3/|\ell|$.

Assume that $F(\varphi_c)(\ell-k)\in \Lambda_\eps$. Then, when $|\ell|$ is large enough, there is a bounded function $ R_{8}(\xi,\ell,k)$ such that
\begin{align}
&E_{k,\ell}-\tilde{C}_{\omega_c}(k,\ell)\label{tempd2wd}
\\
&=\frac{(-1)^{k}|G(\omega_c)|^{\ell+k}}{\pi(1-a)^2}\begin{cases}
-e^{(\ell+k)c'\varphi_c^2}\int_0^{\varphi_c}\cos\Big ((\ell-k)F(\theta)\Big) e^{-(\ell+k)c'\theta^2}d\theta, & \ell+k< -2\\
-((\ell-k)/\sqrt{1-2c})^{-1}\sin((\ell-k)F(\varphi_c)),&\ell+k\in \{-2,0,2\}\\
e^{(\ell+k)c'\varphi_c^2}\int_{\varphi_c}^{\infty}\cos\Big ((\ell-k)F(\theta)\Big) e^{-(\ell+k)c'\theta^2}d\theta, & \ell+k>2
\end{cases}
\nonumber\\
&\qquad \times\Big(1+R_{8}(\xi,\ell,k)\Big).\nonumber
\end{align}
Furthermore, there is a constant $C>0$ such that
\begin{align}
|R_{8}(\xi,\ell,k)|\leq C\max(\frac{1}{\sqrt{|\ell|}},\varphi_c). \label{tempR54}
\end{align}

\end{proposition}


\section{Discussion}\label{SecDiscussion}
In this section we discuss the decay of the covariance between pairs of dimers in the transition region between the rough and smooth phases in the setting of Example \ref{example}. Recall the definitions \eqref{klh1}, \eqref{klh2}.
For $i=1,2$ label the evaluation of $k_i, \ell_i$  at $(x^{(1)},y^{(2)})$ as $k^{1,2}_i, \ell^{1,2}_i$ and label their evaluation at $(x^{(2)},y^{(1)})$ as $k^{2,1}_i, \ell^{2,1}_i$. We have  
\begin{align}
k_1^{1,2}=-r_2 -1, \ \ell_1^{1,2}=r_1-1, \ 
k_2^{1,2}=-r_2 , \  \ell_2^{1,2}=r_1
\label{tempk1l1}
\end{align}
and
\begin{align}
k_1^{2,1}=r_2 -1, \ \ell_1^{2,1}=-r_1-1, \ 
k_2^{2,1}=r_2 , \ \ell_2^{2,1}=-r_1.
\label{tempkl12}
\end{align}
This motivates us to think of
\begin{align}
(\ell,k)=r(-\sigma_1,\sigma_2)
\end{align}
(where $\sigma_1^2+\sigma_2^2=1$) in the sense that $r\sim\sqrt{r_1^2+r_2^2}$ gives the distance between two dimers and $\pm(\sigma_1,\sigma_2)$ gives the direction. 

Consider sequences $\delta_n, \delta^*_n, r_n>0$ such that $\delta_n\leq\delta^*_n$,
\begin{align}
 \delta_n,\delta^*_n\rarrow 0, \quad r_n\rarrow \infty \label{assump111}
\end{align}
and
\begin{align}
r_n=o(\sqrt{n\sqrt{\delta_n}}), \quad \delta_n n^{2/3}\rarrow\infty,
\label{assumprn}
\end{align}
as $n\to\infty$.
We assume that $r_{\min}\leq \sqrt{r_1^2+r_2^2}\leq r_n$ and $\delta_n\leq \xi_c-\xi\leq \delta^*_n$, where $r_{\min}$ is a positive number that will be taken large enough. For example we could take $\delta_n=n^{-1/2}$ and $\delta^*_n = 1/\log(\log(n))$. Note that \eqref{assumprn} implies
\begin{align}\label{rndom}
\frac 1{r_n}>>\frac 1{ \sqrt{n\sqrt{\xi_c-\xi}}}=o(n^{-1/3}).
\end{align}
 Recall Lemma \ref{thetacxi}, in particular \eqref{thetacintermsofxi} which gives a positive constant $d_3$ such that
 \begin{align}
\varphi_c=d_3\sqrt{\xi_c-\xi}+R_3(\xi)(\xi_c-\xi).
\end{align} 
 We see there are constants $c_1,C_1>0$ 
 \begin{align} 
 c_1\sqrt{\xi_c-\xi}\leq\varphi_c\leq C_1\sqrt{\xi_c-\xi}. \label{temppassa}
 \end{align}
The dimer-dimer correlation is given by formula \eqref{covarianceform}, which in the setting of Example \ref{example} becomes
\begin{align}
\text{corr}(e_1,e_2)=a^2K_{a,1}^{-1}(x^{(2)},y^{(1)})K_{a,1}^{-1}(x^{(1)},y^{(2)}).\label{corr12}
\end{align}
Let $(i,j)=(1,2)$ or $(2,1)$. 
By Theorem \ref{ThmInvKast}
\begin{align}
K_{a,1}^{-1}(x^{(i)},y^{(j)})=\K_{1,1}^{-1}(x^{(i)},y^{(j)})-C_{\omega_c}(x^{(i)},y^{(j)})+R_{0,0}(a,x^{(i)},y^{(j)})+O(e^{-Cn}).\label{temp23ws}
\end{align}
 Putting the coordinates in Example \ref{example} into \eqref{changeofcoordinates} yields Proposition \ref{uniformbound} as
\begin{align}
|R_{0,0}(a,x^{(i)},y^{(j)})|&\leq B'|G(\omega_c)|^{r_i-r_j}\frac{1}{\sqrt{n\sqrt{\xi_c-\xi}}}\label{remainderbd}\\
&\leq B'e^{Cr_n\varphi_c^2}|G(i)|^{r_i-r_j}\frac{1}{\sqrt{n\sqrt{\xi_c-\xi}}}.\nonumber
\end{align}
In the second inequality in \ref{remainderbd} we used the following approximation; Set $\alpha=0$ in \eqref{psithetaalpha} to get $\mcR[\psi(\theta_c)]=\log|G(\omega_c)|$, and then use \eqref{tempz4f4} to obtain
\begin{align}
|G(\omega_c)|^{r_i-r_j}= e^{(r_i-r_j)O(\varphi_c^2)}|G(i)|^{r_i-r_j}.\label{estGomegac}
\end{align}
Now we use \eqref{K_11} and \eqref{Comegacccc}  to write
\begin{align}
\K_{1,1}^{-1}(x^{(i)},y^{(j)})-C_{\omega_c}(x^{(i)},y^{(j)})=-\sqrt{-1}\big(E_{k_1^{i,j},l_1^{i,j}}-\tilde{C}_{\omega_c}(k_1^{i,j},l_1^{i,j})+a(E_{k_2^{i,j},l_2^{i,j}}-\tilde{C}_{\omega_c}(k_2^{i,j},l_2^{i,j}))\big)\label{tempz5f4}
\end{align}
where we recall  $l_1^{1,2}, k_1^{1,2}$... as given in \eqref{tempk1l1}, \eqref{tempkl12}.

\subsection{Dimer-dimer correlations parallel to $\vec{e}_1$}
\label{corrdire1}
We look at dimers separated along the diagonal, that is, we let $r_1=r_2=r>0$.  This implies $k_1^{2,1}=l_1^{1,2}=r-1$, $k_1^{1,2}=l_1^{2,1}=-r-1$ and $k_2^{1,2}=l_2^{2,1}=-r$, $k_2^{2,1}=l_2^{1,2}=r$. 
We obtain the formula
\begin{align}
K_{a,1}^{-1}(x^{(i)},y^{(j)})=-\sqrt{-1}\big(E_{r-1,-r-1}+aE_{-r,r}-\tilde{C}_{\omega_c}(r-1,-r-1)-a\tilde{C}_{\omega_c}(-r,r)\big)+R_{0,0}(a,x^{(i)},y^{(j)}),\label{tempz5f5}
\end{align}
where we used \eqref{symmeklcomegac}. We will use $\simeq$ to say that two expressions are equal if we neglect subdominant terms.
Take $r_{\min}$ so large that we can neglect remainders in Proposition \ref{propEKL}. Note that $|G(i)|=(1-\sqrt{1-2c})/\sqrt{2c}<1$.
Although other length scales can be analysed, for simplicity take $\delta_n=\xi_c-\xi=\delta_n^*$. 
 \begin{customthm}{3.2}\label{CorCov1}
Take a very large $n$, a small $\eps>0$ and fix $\gamma\in(1/2,1)$. Take $\xi_c-\xi=\delta_n$, which gives $\varphi_c\sim d\sqrt{\delta_n}$ with an explicit constant $d$. Assume that $r_{\min}<r<r_n$,  
\begin{align}
2r_n\varphi_c^{2-\gamma}\le 1, \label{TEMP212}
\end{align}
and $2r\varphi_c/\sqrt{1-2c}\in \Lambda_\eps$, where $ \Lambda_\eps$ is defined by \eqref{Lambdaeps}.
Then
\begin{equation}\label{CorrAs1}
\text{\rm corr}(e_1,e_2)\simeq-a^2\Big (\frac{|G(i)|^{2r}\sqrt{1-a}}{\sqrt{8\pi r a \sqrt{1+a^2}}}+\frac{a-|G(i)|^{-2}}{(1-a)^2\pi}\frac{\sin(2r\varphi_c/\sqrt{1-2c})}{2r/\sqrt{1-2c}}\Big)^2
\end{equation}
\end{customthm}


\begin{proof}
It follows from our assumptions that $r_n\varphi_c^2=o(1)$, together with \eqref{remainderbd} and \eqref{rndom} we have the estimate
\begin{equation*}
R_{0,0}(a,x^{(i)},y^{(j)})=o(n^{-1/3}).
\end{equation*}
An evaluation using \eqref{g''alphaas} gives
\begin{align}
E_{-r,r}=E_{r,r}\simeq \frac{(-1)^{r}|G(i)|^{2r}}{\sqrt{8\pi c \ r  (1-2c)^{1/2}}(1+a^2)}.\label{temp2rds}
\end{align}
We use
lemma \ref{alphaalpha'} with $\alpha=(r-1)/(r+1)$, $\alpha'=1$ to get $\mcR[\tilde{g}_{\alpha}(iw_\alpha)]=\mcR[\tilde{g}_1(i)]-2\log|G(i)|/r+O(1/r^2)$. We can use this in Proposition \ref{propEKL}, to get
\begin{align}
E_{r-1,-r-1}=E_{r-1,r+1}\simeq-\frac{(-1)^r|G(i)|^{2r}}{\sqrt{8\pi c r(1-2c)^{1/2}}(1+a^2)}.
\end{align}
We now apply Proposition \ref{Comegacinnerdisc} to the $\tilde{C}_{\omega_c}$ appearing in \eqref{tempz5f5}. By \eqref{temppassa}, the remainder $R_{61}$ in \eqref{tempCom12} is very small compared to $1$ and so we neglect this term. It follows from Proposition \ref{Comegacinnerdisc} that we can write
\begin{align}
\tilde{C}_{\omega_c}(r-1,-r-1)+a\tilde{C}_{\omega_c}(-r,r)\simeq (-1)^{r}\frac{(a-|G(i)|^{-2})}{(1+a^2)\pi}\frac{\sin(2r\varphi_c/\sqrt{1-2c})}{2r\sqrt{1-2c}} \label{Comegactemp21}
\end{align}
where  $a-|G(i)|^{-2}<0$.
Now from \eqref{tempz3e4} we obtain
\begin{align}
|\tilde{C}_{\omega_c}(r-1,-r-1)+a\tilde{C}_{\omega_c}(-r,r)|\geq C\min(\frac{\varphi_c}{1-2c},\frac{1}{2r}).\label{temptempbound345f3}
\end{align}
Combining this with \eqref{rndom} we see that the error term $o(n^{-1/3})$ can be neglected. A computation using \eqref{tempz5f5} now gives \eqref{CorrAs1}.
\end{proof}

Observe how the behaviour of 
\begin{align}
\frac{\sin(2r\varphi_c/\sqrt{1-2c})}{2r/\sqrt{1-2c}}
\end{align}
changes as a function of $r$. If $2r\varphi_c$ is small compared to $\sqrt{1-2c}$ it is close to the constant $\varphi_c$ ($n$ is fixed),
whereas for $2r\varphi_c>\sqrt{1-2c}$ it starts to slowly oscillate and decay like $1/r$.
This leads to three regimes for the correlation as a function of $r$

\noindent \emph{Regime I}.
As $r$ varies from $r_{\min}$ to $\frac{1}{\log|G(i)|^{-2}}\log\frac{1}{\varphi_c}$,
\begin{align}
\text{corr}(e_1,e_2)\simeq -\frac{a(1-a)}{8\pi\sqrt{1+a^2}}\frac{|G(i)|^{4r}}r,
\end{align}
which has exponential decay in $r$. 

\noindent \emph{Regime II}. As $r$ varies from $\frac{1}{\log|G(i)|^{-2}}\log\frac{1}{\varphi_c}$ to small compared to $\sqrt{1-2c}/(2\varphi_c)$ then 
\begin{align}
\text{corr}(e_1,e_2)\simeq -\Big(\frac{a(a-|G(i)|^{-2})}{\pi(1-a)^2}\Big)^2\varphi_c^2
\end{align}
which has no decay, so the correlation is constant.

\noindent \emph{Regime III}.
As $r$ varies from $\sqrt{1-2c}/(2\varphi_c)$ to $r_n$, 
\begin{align}
\text{corr}(e_1,e_2)\simeq -\Big(\frac{a(a-|G(i)|^{-2})}{2\pi (1-a)\sqrt{1+a^2}}\Big)^2\frac{\sin^2(2r\varphi_c/\sqrt{1-2c})}{r^2} 
\end{align}
which oscillates and has a decay like $1/r^2$.

\subsection{Dimer-dimer correlations parallel to $\vec{e}_2$}
Now we look at dimers separated along the anti-diagonal. So we instead take $r_1=-r_2=r>0$ for which we have $k_1^{1,2}=\ell_1^{1,2}=r-1$, $k_1^{2,1}=\ell_1^{2,1}=-r-1$, $k_2^{2,1}=l_2^{2,1}=-r$, $k_2^{1,2}=l_2^{1,2}=r$. So we obtain
\begin{align}
K^{-1}_{a,1}(x^{(i)},y^{(j)})=&-\sqrt{-1}\big(E_{(-1)^jr-1,(-1)^jr-1}+aE_{(-1)^jr,(-1)^jr}\label{remtemp}\\&-\tilde{C}_{\omega_c}((-1)^jr-1,(-1)^jr-1)-a\tilde{C}_{\omega_c}((-1)^jr,(-1)^jr)\big)\nonumber\\&+R_{0,0}(a,x^{(i)},y^{(j)})+O(e^{-Cn}). \nonumber
\end{align}
 \begin{customthm}{3.3}
Take a very large $n$, a small $\eps>0$ and fix $\gamma\in(1/2,1)$. Take $\xi_c-\xi=\delta_n$, which gives $\varphi_c\sim d\sqrt{\delta_n}$ with an explicit constant $d$. If $r_{\min}<r<\varphi_c^{\gamma-2}/2$ then
\begin{align}
\text{\rm corr}(e_1,e_2)\simeq\frac{a(1-a)}{8\pi  \sqrt{1+a^2}}\frac{|G(i)|^{4r}}{r}+\frac{a^{3/2}(a-|G(i)|^{-2})^2}{\sqrt{8\pi^3(1-a)^5\sqrt{1+a^2}}}\frac{\varphi_c}{\sqrt{r}}.\label{corrtempc2s}
\end{align}
If instead $\varphi_c^{\gamma-2}/2\leq r<r_n$ then
\begin{align}
\text{\rm corr}(e_1,e_2)\simeq \frac{a^2(a-|G(\omega_c)|^{-2})^2}{2\pi^2ra(1-a)\sqrt{1+a^2}}\int_{\sqrt{2rc'}\varphi_c}^\infty\int_0^{\sqrt{2rc'}\varphi_c}e^{t_1^2-t_2^2}dt_1dt_2\label{corr1ws}
\end{align}
where $c'=a\sqrt{1+a^2}/(1-a)^2$.
\begin{proof}
We can use \eqref{temp2rds}
 \begin{align}
 E_{(-1)^jr-1,(-1)^jr-1}+aE_{(-1)^jr,(-1)^jr}=\frac{(-1)^r|G(i)|^{2r}}{\sqrt{8\pi r\  c(1-2c)^{1/2}}(1+a^2)}\Big(a-|G(i)|^{-2(-1)^j}+O(1/r)\Big) \label{temp22ea}
 \end{align}
 where the remainder $O(1/r)$ comes from estimating the difference between $E_{(-1)^jr-1,(-1)^jr-1}$ and $E_{(-1)^jr,(-1)^jr}$. Clearly \eqref{temp22ea} decays exponentially.
 Take $r_{\min}$ large enough that the remainders in Proposition \ref{Comegacinnerdisc} are small. We can use Proposition \ref{Comegacinnerdisc} to write that for $r\leq1/(2\varphi_c^{2-\gamma} )$,
\begin{align}
\tilde{C}_{\omega_c}((-1)^jr-1,(-1)^jr-1)+a\tilde{C}_{\omega_c}((-1)^jr,(-1)^jr)
\simeq(-1)^r|G(i)|^{(-1)^j2r}\frac{(a-|G(i)|^{-2})\varphi_c}{\pi(1-a)^2}. \label{temp21313}
\end{align}
Observe that there are no oscillations.
By \eqref{assumprn} and \eqref{remainderbd} we see the remainder in \eqref{remtemp} is small compared to \eqref{temp21313} for $n$ large.
 So 
\begin{align}
\text{corr}(e_1,e_2)\simeq&-a^2\Big( \frac{|G(i)|^{2r}(a-|G(i)|^{-2})}{\sqrt{8\pi c \ r  (1-2c)^{1/2}}(1+a^2)}-|G(i)|^{2r}\frac{(a-|G(i)|^{-2})\varphi_c}{\pi(1-a)^2} \Big)\nonumber\\&\quad\times\Big( \frac{|G(i)|^{2r}(a-|G(i)|^{2})}{\sqrt{8\pi c \ r  (1-2c)^{1/2}}(1+a^2)}-|G(i)|^{-2r}\frac{(a-|G(i)|^{-2})\varphi_c}{\pi(1-a)^2} \Big),\nonumber
\end{align} 
now \eqref{corrtempc2s} follows by using \eqref{temp24dsa} with $\theta=-\pi/4$ which gives
\begin{align}
(a-|G(i)|^{-2})(|G(i)|^2-a)=(1-a)^2
\end{align}
and since $1/\sqrt{r}>>\varphi_c$ in this regime.

If $1/(2\varphi_c^{2-\gamma})\leq r\leq r_n $, by substitution in Proposition \ref{propeklcomegac} we obtain
\begin{align}
&E_{r,r}-\tilde{C}_{\omega_c}(r,r)\simeq\frac{(-1)^r|G(\omega_c)|^{2r}}{\pi(1-a)^2}\frac{D_-(\sqrt{2rc'}\varphi_c)}{\sqrt{2rc'}}
\label{temp2essa}\\
&E_{-r,-r}-\tilde{C}_{\omega_c}(-r,-r)\simeq-\frac{(-1)^r|G(\omega_c)|^{-2r}}{\pi(1-a)^2}\frac{D_+(\sqrt{2rc'}\varphi_c)}{\sqrt{2rc'}}
\label{temp23we}
\end{align}
where 
\begin{align}
D_+(z)=e^{-z^2}\int_0^ze^{t^2}dt, &&
D_-(z)=e^{z^2}\int_z^\infty e^{-t^2}dt.
\end{align} $D_+, D_-$ are known as the Dawson function and Mills ratio, respectively. These functions are related to the imaginary and complementary error functions. We use 7.8.7 in \cite{DLMF} to write
\begin{align}
D_+(z)<\frac{1-e^{-z^2}}{z}\label{Dupper}
\end{align}
for $z>0$. We apply the upper bound $|e^{-z^2}-1|\leq z^2e^{z^2}$ to \eqref{Dupper} on $(0,1]$. We also have that for the same interval,  the lower bounds $e^{- z^2}\geq e^{-1}$, and $t\geq0$, $e^{t^2}\geq 1$ hold, hence
 \begin{align}
e^{-1}z \leq D_+(z)< ez.
 \label{dawsonest1}
 \end{align} For  $z\geq 1$, inequality (1) in \cite{C.T} gives a lower bound and we apply the upper bound $1-e^{-z^2}\leq 1-e^{-1}$ to \eqref{Dupper}  to get
\begin{align}
\frac{1}{2z}\leq D_+(z)< \frac{1-e^{-1}}{z}.\label{dawsonest2}
\end{align} 
By 7.8.3 in \cite{DLMF} we have
\begin{align}
\frac{\sqrt{\pi}}{2\sqrt{\pi}z+2}\leq D_-(z)< \frac{1}{z+1}\label{millest}
\end{align}
for $z\geq 0$.
Hence 
\begin{align}
&\frac{D_+(\sqrt{2rc'}\varphi_c)}{\sqrt{2rc'}}\geq \begin{cases}
e^{-1}\varphi_c, &\sqrt{2rc'}\varphi_c\in(0,1]\\
(4rc'\varphi_c)^{-1} & \sqrt{2rc'}\varphi_c\in(1,\infty),
\label{tempm23e2e}
\end{cases}\\
&\frac{D_-(\sqrt{2rc'}\varphi_c)}{\sqrt{2rc'}}\geq \sqrt{\frac{\pi}{8c'}}\frac{1}{\sqrt{2\pi c'}r\varphi_c+\sqrt{r}}
\label{tempm23we}
\end{align}
We can use \eqref{assumprn} to show that both \eqref{tempm23e2e},\eqref{tempm23we} are much greater than $1/\sqrt{n\sqrt{\xi_c-\xi}}$. Hence it follows from \eqref{temp2essa},  \eqref{temp23we} and \eqref{remainderbd} that the remainders in \eqref{remtemp} are negligible.
So we have
\begin{align}
K_{a,1}^{-1}(x^{(1)},y^{(2)})&\simeq -i\big(E_{r,r}-\tilde{C}_{\omega_c}(r,r)+a(E_{r-1,r-1}-\tilde{C}_{\omega_c}(r-1,r-1)\big)\\
&\simeq -i\Big(\frac{(-1)^r|G(\omega_c)|^{2r}}{\pi(1-a)^2}\Big[\frac{D_-(\sqrt{2rc'}\varphi_c)}{\sqrt{2rc'}}(a-|G(\omega_c)|^{-2})+O(\frac{1}{r})\Big]\Big)\label{teplml2}
\end{align}
and
\begin{align}
K_{a,1}^{-1}(x^{(2)},y^{(1)})&\simeq i\Big(\frac{(-1)^r|G(\omega_c)|^{-2r}}{\pi(1-a)^2}\Big[\frac{D_+(\sqrt{2rc'}\varphi_c)}{\sqrt{2rc'}}(a-|G(\omega_c)|^{-2})+O(\frac{1}{r})\Big]\Big).\label{3mplwf}
\end{align}
The $O(1/r)$ error appearing in \eqref{teplml2}  arises from the difference in evaluating $D_-(\sqrt{2 \cdot c'}\varphi_c)/\sqrt{2\cdot c'}$ at $r$ and $r-1$, similarly for $D_+$ and \eqref{3mplwf}. By \eqref{tempm23e2e} and \eqref{tempm23we} these errors are negligible in the current regime. Hence by \eqref{corr12} we have \eqref{corr1ws} for $1/(2\varphi_c)^{2-\gamma}\leq r \leq r_n$
\end{proof}
 \end{customthm}

We bring forward the asymptotics contained in (42.6) from \cite{Atlas} and (7.12.1), (7.6.2) in \cite{DLMF} (for which we note $D_-(z)=\sqrt{\pi}e^{z^2}\text{erfc}(z)/2$). For $z>0$,
\begin{align}
D_+(z)=\begin{cases} z+O(z^3) & \text{for $z$ small} \\ \frac{1}{2z} +O(\frac{1}{z^3}) &\text{for $z$ large} \end{cases},
\quad D_-(z)=\begin{cases}\frac{\sqrt{\pi}}{2}+O(z)&\text{for $z$ small} \\\frac{1}{2z}+O(\frac{1}{z^3})&\text{for $z$ large}\end{cases}.
\label{D-D+asymps}
\end{align}

We see the following distinct decay rates.

\noindent \emph{Regime I}.
As $r$ varies from $r_{\min}$ to $\frac{1}{\log|G(i)|^{-4}}\log\frac{1}{\varphi_c}$ ,
\begin{align}
\text{corr}(e_1,e_2)\simeq \frac{a (1-a)}{8\pi  \sqrt{1+a^2}}\frac{|G(i)|^{4r}}{r}
\end{align} 
 which decays exponentially. 
 
 For the next regime, we use \eqref{corrtempc2s} and \eqref{estGomegac}  when $r$ varies from $\frac{1}{\log|G(i)|^{-4}}\log\frac{1}{\varphi_c}$ to $1/(2\varphi_c^{2-\gamma})$. We then use \eqref{D-D+asymps} on \eqref{corr1ws} when $r$ varies from $1/(2\varphi_c^{2-\gamma})$ to small compared to $1/(2c'\varphi_c^2)$. These two sub-regimes have the same leading order term which is just an artefact of how we proved the asymptotic formulas. We combine them into one regime.
 
 \noindent \emph{Regime II}.
As $r$ varies from $\frac{1}{\log|G(i)|^{-4}}\log\frac{1}{\varphi_c}$ to small compared to $1/(2c'\varphi_c^2)$,
\begin{align}
\text{corr}(e_1,e_2)\simeq\frac{a^{3/2}(a-|G(\omega_c)|^{-2})^2}{\sqrt{8\pi^{3} (1-a)^5\sqrt{1+a^2}}} \frac{\varphi_c}{\sqrt{r}} 
\end{align}
which decays like $1/\sqrt{r}$. 

 \noindent \emph{Regime III}.
As $r$ varies from small compared to $1/(2c'\varphi_c^2)$ to large compared to $1/(2c'\varphi_c^2)$, corr$(e_1,e_2)$ is given by \eqref{corr1ws} and the integral is bounded below by a positive constant. Hence corr$(e_1,e_2)$ decays like $1/r$.

For the next regime, we use \eqref{D-D+asymps} in \eqref{corr1ws}.

\noindent \emph{Regime IV}.
As $r$ varies from large compared to $1/(2c'\varphi_c^2)$ to $r_n$,
\begin{align}
\text{corr}(e_1,e_2)\simeq\frac{(1-a)^2(a-|G(\omega_c)|^{-2})^2}{(4\pi)^2(1+a^2)}\frac{1}{(\varphi_cr)^2}
\end{align}
which decays like $1/r^2$. One can also show that both $a-|G(i)|^{-2}<0$ and $|G(i)|^2-a<0$, hence we see that the correlation is positive in the four regimes above.\\
\subsection{Dimer-dimer correlations at an arbitrary angle.}

Now we analyse a third case which interpolates between the two previous cases. We consider $2r_n\varphi_c^{2-\gamma}\leq1$ for a fixed $\gamma\in(1/2,1)$.  For an arbitrarily small fixed $\eps>0$, consider 
\begin{align}
(r_1,r_2)=\sqrt{2}r(\cos(\theta),\sin(\theta))=:r(\sigma_1,\sigma_2),\quad \theta\in [-\pi/4,-\eps]\cup [\eps,\pi/4-\eps].\label{sig1sig2}
\end{align}
if $\theta=\pi/4$ this would correspond to separating the two dimers along the diagonal, when $\theta=-\pi/4$  this corresponds to the anti-diagonal. In particular, it is interesting to examine the behaviour for $\theta$ close to $-\pi/4$.

Let $\sigma:=\sigma_2/\sigma_1=\tan(\theta)$, so $-1\leq \sigma<1$. We now have  $l_1^{1,2}+1=l_2^{1,2}=r\sigma_1=-l_2^{2,1}=-l_1^{2,1}-1$ and $k_1^{2,1}+1=k_2^{2,1}=r\sigma_2=-k_2^{1,2}=-k_1^{1,2}-1$. 
To state a formula for the correlation in the current setting we require the definitions of some functions, these functions come into the rate of exponential decay and leading term constants. Define $\tilde{\sigma}=\text{sign}(\theta)$, sign$(0)=1$ and note $\sigma_1>0$. Define the piecewise continuous functions
\begin{align}
 &h_{\pm}(\theta)=\mcR[\tilde{g}_{|\sigma|}(iw_{|\sigma|})]\pm(1-\sigma)\log|G(i)|,\\
 &g_{\pm}(\theta)=f(|\sigma|)\frac{a-|G(iw_{|\sigma|})|^{\pm1}|G(1/(iw_{|\sigma|}))|^{\mp\tilde{\sigma}}}{a-|G(i)|^{-2}}
\end{align}
where $f(\alpha)=(\sqrt{2\pi|g''_\alpha(iw_\alpha)|}w_\alpha\sqrt{w_\alpha^2-2c}\sqrt{1/w_\alpha^2-2c})^{-1}$, $\alpha\in(0,1]$.

 \begin{customthm}{3.4}
Take a very large $n$, a small $\eps>0$ and fix $\gamma\in(1/2,1)$. Take $\xi_c-\xi=\delta_n$, which gives $\varphi_c\sim d\sqrt{\delta_n}$ with an explicit constant $d$. Assume that $r_{\min}<r<r_n$,  
\begin{align}
2r_n\varphi_c^{2-\gamma}\le 1, \label{TEMP211a2}
\end{align}
and $r(\sigma_1+\sigma_2)\varphi_c/\sqrt{1-2c}\in \Lambda_\eps$, where $ \Lambda_\eps$ is defined by \eqref{Lambdaeps}. Then 
\begin{align}
\text{\rm corr}(e_1,e_2)\simeq-\frac{a^2(a-|G(i)|^{-2})^2}{(1+a^2)^2}\prod_{q=\pm}\Big(\frac{g_q(\theta)}{\sqrt{r\sigma_1}}e^{r\sigma_1h_q(\theta)}-\frac{\sin(r(\sigma_1+\sigma_2)\varphi_c/\sqrt{1-2c})}{r\pi(\sigma_1+\sigma_2)\sqrt{1-2c}}\Big).\label{corrarbangle}
\end{align}
\begin{proof}
We use Theorem \ref{ThmMain} to establish \eqref{temp2kss} and \eqref{t3ed}.  By \eqref{K_11}, we have
\begin{align}
\K_{1,1}^{-1}(x^{(1)},y^{(2)})=-i(E_{|r\sigma_2|+\tilde{\sigma}_2,r\sigma_1-1}+aE_{|r\sigma_2|,r\sigma_1}),\label{templh1r}\\
\K_{1,1}^{-1}(x^{(2)},y^{(1)})=-i(E_{|r\sigma_2|-\tilde{\sigma}_2,r\sigma_1+1}+aE_{|r\sigma_2|,r\sigma_1}).\nonumber
\end{align}

Let $\alpha^{1,2}=(|r\sigma_2|+\tilde{\sigma}_2)/(r\sigma_1-1)$ and observe that $C_\eps\leq|\alpha^{1,2}|\leq 1$, where $C_\eps$ is small, for large $r$ by \eqref{sig1sig2}.  By \eqref{EklAs1}, we have
\begin{align}
\K_{1,1}^{-1}(x^{(1)},y^{(2)})\simeq -i\frac{f(\alpha^{1,2})}{\sqrt{r\sigma_1-1}}\frac{(-1)^{|r\sigma_2|+\tilde{\sigma}_2}}{1+a^2}e^{(r\sigma_1-1)\mcR[\tilde{g}_{\alpha^{1,2}}(iw_{\alpha^{1,2}})]}-ia\frac{f(|\sigma|)}{\sqrt{r\sigma_1}}\frac{(-1)^{|r\sigma_2|}}{1+a^2}e^{r\sigma_1\mcR[\tilde{g}_{|\sigma|}(iw_{|\sigma|})]}.\label{temp2kss}
\end{align}
We have $\alpha^{1,2}-|\sigma|=(|\sigma|+\tilde{\sigma}_2)/(r\sigma_1)+O(1/r^2)$ and $\frac{f(\alpha^{1,2})}{\sqrt{r\sigma_1-1}}=\frac{f(|\sigma|)}{\sqrt{r\sigma_1}}+O(1/r)$ so 
\begin{align}
\K_{1,1}^{-1}(x^{(1)},y^{(2)})\simeq -i\frac{f(|\sigma|)(-1)^{r|\sigma_2|}}{\sqrt{r\sigma_1}(1+a^2)}|G(iw_{|\sigma|})|^{r\sigma_1}|G(1/(iw_{|\sigma|}))|^{r|\sigma_2|}\Big (a-|G(iw_{|\sigma|})|^{-1}|G(1/(iw_{|\sigma|}))|^{\tilde{\sigma}}\Big)
\end{align}
by lemma \ref{alphaalpha'}.
Similarly,
\begin{align}
\K_{1,1}^{-1}(x^{(2)},y^{(1)})\simeq -i\frac{f(|\sigma|)(-1)^{r|\sigma_2|}}{\sqrt{r\sigma_1}(1+a^2)}|G(iw_{|\sigma|})|^{r\sigma_1}|G(1/(iw_{|\sigma|}))|^{r|\sigma_2|}\Big (a-|G(iw_{|\sigma|})||G(1/(iw_{|\sigma|}))|^{-\tilde{\sigma}}\Big).
\end{align}
We have
\begin{align}
C_{\omega_c}(x^{(i)},y^{(j)})\simeq-i(a-|G(i)|^{-2})\frac{(-1)^{-r\sigma_2}|G(i)|^{r_i-r_j}}{(1+a^2)\pi}\frac{\sin(r(\sigma_1+\sigma_2)\varphi_c/\sqrt{1-2c})}{r(\sigma_1+\sigma_2)\sqrt{1-2c})}.\label{t3ed}
\end{align}
We see the remainder $R_{0,0}$ in \eqref{temp23ws} is much smaller than \eqref{t3ed} via \eqref{remainderbd}.
\end{proof}
 \end{customthm}

Note the following facts about $h_+$ and  $h_-$, since $(1+|\sigma|)\log|G(i)|\leq(1-\sigma)\log|G(i)|\leq (1-|\sigma|)\log|G(i)|$ lemma \ref{galphaalpha} gives
\begin{align}
h_+\leq h_-\leq 0.
\end{align}
Since $h_+-h_-=2(1-\sigma)\log|G(i)|$ we have $h_+(\theta)<0$, it is also easy to compute $h_-(-\pi/4)=0$. In fact, one can show that there is a positive $c_5>0$ such that $h_-(\theta)=-c_5(\theta+\pi/4)^2+O(\theta+\pi/4)^3$.

We have the following decay regimes. Let $m>0$ be such that $m/\sqrt{1-2c}$ is small.

\noindent \emph{Regime I}.
As $r$ varies from $r_{\min}$ to $\frac{1}{\sigma_1|h_+(\theta)|}\log(\frac{1}{\varphi_c})$.
\begin{align}
\text{corr}(e_1,e_2)\simeq -\text{sign}(\theta)\frac{ac\sqrt{1-2c w_{|\sigma|}^2}}{2\pi \sqrt{(1-\sigma^2)^2+16c^2\sigma^2}}\frac{e^{r\sigma_1(h_-(\theta)+h_+(\theta))}}{r\sigma_1}\label{regime1general}
\end{align}
which decays exponentially. Note we used \eqref{g-g+} here.

\noindent \emph{Regime II}.
As $r$ varies from $\frac{1}{\sigma_1|h_+(\theta)|}\log(\frac{1}{\varphi_c})$ to $\min(\frac{1}{\sigma_1|h_-(\theta)|}\log(\frac{1}{\varphi_c}),m/((\sigma_1+\sigma_2)\varphi_c)$
\begin{align} 
\simeq \frac{a^2(a-|G(i)|^{-2})^2g_-(\theta)}{(1-a)^2(1+a^2)}\frac{\varphi_ce^{r\sigma_1 h_-(\theta)}}{\pi\sqrt{r\sigma_1}}
\end{align}
which decays exponentially for $\theta\neq -\pi/4$, and like $1/\sqrt{r}$ for $\theta=-\pi/4$.

\noindent \emph{Regime III}.
As $r$ varies from $\min(\frac{1}{\sigma_1|h_-(\theta)|}\log(\frac{1}{\varphi_c}),m/((\sigma_1+\sigma_2)\varphi_c))$ to $m/((\sigma_1+\sigma_2)\varphi_c)$
\begin{align}
\simeq -\Big(\frac{a(a-|G(i)|^{-2})}{\pi(1-a)^2}\Big)^2\varphi_c^2
\end{align}
which has zero or no decay.

\noindent \emph{Regime IV}.
As $r$ varies from $m/((\sigma_1+\sigma_2)\varphi_c)$ to $r_n$,
\begin{align}
\simeq -\Big(\frac{a(a-|G(i)|^{-2})}{\pi(1-a)\sqrt{1+a^2}}\Big)^2\frac{\sin^2(r(\sigma_1+\sigma_2)\varphi_c/\sqrt{1-2c})}{r^2(\sigma_1+\sigma_2)^2}.
\end{align}
which oscillates and decays like $1/r^2$.

Hence we see that for $\theta>-\pi/4$ close to $-\pi/4$, there are two distinct exponential decay rates; a faster exponential decay in Regime I followed by a slower exponential decay in Regime II. As we vary the angle $\theta$ closer to $-\pi/4$, Regime II increases in size, and the decay in Regime II changes from exponential decay to a decay like $1/\sqrt{r}$.
\begin{remark}
A simpler expression can be given for the product $g_-(\theta)g_+(\theta)$ that appears when one expands the product over $q$ in \eqref{corrarbangle}. This comes into the formula for Regime I \eqref{regime1general}. Through some manipulation one can show the identity
\begin{align}
\label{temp24dsa}&(a-|G(iw_{|\sigma|})||G(1/(iw_{|\sigma|}))|^{-\tilde{\sigma}})(a-|G(iw_{|\sigma|})|^{-1}|G(1/(iw_{|\sigma|}))|^{\tilde{\sigma}})\\&=\tilde{\sigma}(1+a^2)\sqrt{w_{|\sigma|}^2-2c}\sqrt{1/w_{|\sigma|}^2-2c}.\nonumber
\end{align}
The identity \eqref{temp24dsa} together with $w_{|\sigma|}^2\sqrt{1/w_{|\sigma|}^2-2c}=\alpha\sqrt{w_{|\sigma|}^2-2c}$ and \eqref{g''alphaas} can be used to obtain 
\begin{align}
g_-(\theta)g_+(\theta)=\tilde{\sigma}\frac{(1+a^2)\sqrt{1-2c w_{|\sigma|}^2}}{2\pi \sqrt{(1-\sigma^2)^2+16c^2\sigma^2}}\frac{1}{(a-|G(i)|^{-2})^2}\label{g-g+}
\end{align}
\end{remark}

\section{Asymptotics of $\K_{1,1}^{-1}$}
In \cite{C/J}, the asymptotics of $E_{k,\ell}$ are computed when $\pm \ell$ is close to $k$, i.e when the vertices lie close to diagonal or cross-diagonal from one another. Here we extend the asymptotics to arbitrary angles. Note from the symmetry relation we can restrict our attention to the case $\abs{k}\leq\abs{\ell}$ without loss of generality. Let  $\alpha = \frac{|k|}{|\ell|}$ and write
\begin{align}
E_{k,\ell}=\frac{i^{-|k|-|\ell|}}{2(1+a^2)2\pi i}\int_{\Gamma_1} \frac{dw}{w}\frac{\exp[{\abs{\ell} \widetilde{g}_{\alpha}(w)]}}{\sqrt{w^2+2c}\sqrt{1/w^2+2c}}
\label{ekl}
\end{align}
with saddle point function $\widetilde{g}_{\alpha}(w)$.

The integrand in \eqref{ekl} is analytic in the region $\C\setminus \left(i[-\sqrt{2c},\sqrt{2c}]\cup i[1/\sqrt{2c},\infty)\cup i(-\infty,-1/\sqrt{2c}]\right)$, so by Cauchy's deformation theorem, for $r>1$, we can deform $ \Gamma_1$ to $\Gamma_{w_\alpha}(r)=\cup_{j=0}^3 \gamma_j(r)$ where 
\begin{align*}
\gamma_0(r)=\{r+it: t\in[-w_\alpha,w_\alpha]\} && \gamma_1(r) = \{-t+iw_\alpha: t\in (-r,r)\},
\end{align*}
as sets $\gamma_3(r)=-\gamma_1(r)$, $\gamma_2(r)=-\gamma_0(r)$ and as curves each $\gamma_j(r)$ has positive orientation counter-clockwise around the origin. For $\eta_1,\eta_2, r\geq 0$, where $\eta_2$ is small define $\gamma_{\eta_1,\eta_2}(r)$ as
\begin{align}
\{\frac{i}{\sqrt{2c}}-\eta_2 i +t:t\in[0,\eta_2)\}\cup \{\frac{i}{\sqrt{2c}}+t i+\eta_2:t\in[-\eta_2,\eta_1)\} \cup\{\frac{i}{\sqrt{2c}}+i\eta_1+\eta_2+t:t\in[0,r)\}
\end{align}
which is a path from $i/\sqrt{2c}+i\eta_1+\eta_2+r$ to $i/\sqrt{2c}-\eta_2 i$ consisting of three straight lines. Note we define the orientation of this curve to be in the direction of the path beginning at $i/\sqrt{2c}+i\eta_1+\eta_2+r$ and ending at $i/\sqrt{2c}-\eta_2 i$. 

Let $\mathcal{R}[z]$ and $\mathcal{I}[z]$ denote the real and imaginary part of a complex number $z$, respectively.
\begin{lemma} For $k,\ell\in \Z$,
\label{eklready}
\begin{align}
E_{k,\ell}&=\frac{i^{-|k|-|\ell|}}{2(1+a^2)2\pi i}\int_{\gamma_1(\infty)} \frac{dw}{w}\frac{(1+(-1)^{|\ell|+|k|})\exp[{\abs{\ell} \widetilde{g}_{\alpha}(w)]}}{\sqrt{w^2+2c}\sqrt{1/w^2+2c}}.
\label{tempz4e6}
\end{align}
If $\ell+k$ is even, $\eta_1,\eta_2\geq 0$ then 
\begin{align}
E_{k,\ell}=\frac{i^{-|k|-|\ell|}}{(1+a^2)2\pi }\mcI\Big [\int_{\gamma_{\eta_1,\eta_2}(\infty)} \frac{dw}{w}\frac{\exp[{\abs{\ell} \widetilde{g}_{\alpha}(w)]}}{\sqrt{w^2+2c}\sqrt{1/w^2+2c}}\Big ].
\label{tempz4e6e4}
\end{align}
\begin{proof}
One can compute the following expansions for $G$, for $w$ large
\begin{align}
G(w)=-\sqrt{\frac{c}{2}}\frac{1}{w}+O\Big(\frac{1}{w^2}\Big)
\label{Gwlarge}
\end{align}
and for $w$ small
\begin{align}
G(w)=
\begin{cases}
-1+O(w), & \mcR(w)>0\\
1+w+O(w^2), &\mcR(w)\leq 0.
\end{cases}
\label{Gwsmall}
\end{align}
Take \eqref{ekl} and  perform the deformation of the contour described above and in the limit as $r\rarrow \infty$, the contributions from $\gamma_0(r)$ and $\gamma_2(r)$ vanish by standard estimates using \eqref{Gwlarge} and \eqref{Gwsmall}. Then \eqref{tempz4e6} follows from \eqref{squarerootsymmetries}, \eqref{Gsymmetries}. Similarly, \eqref{tempz4e6e4} follows from a further deformation to the contour above and applications of \eqref{squarerootsymmetries}, \eqref{Gsymmetries}.
\end{proof}
\end{lemma}
Next we give a lemma concerning the existence of descent paths. 
\begin{lemma}
\label{gasdesc}
For $\alpha\in[0,1]$, $\beta\in [1,1/\sqrt{2c}]$  the mapping $(0,\infty)\ni t\rarrow \mathcal{R} [\widetilde{g}_\alpha (i \beta+t)]$ is strictly decreasing. Moreover, if $\alpha=0$ the same mapping is strictly decreasing for all $\beta\in [1,\infty)$.
\begin{proof}
We use an integral representation of
\begin{align*}
\frac{1}{\sqrt{w^2+2c}}=\frac{1}{\pi}\int_{-\sqrt{2c}}^{\sqrt{2c}}\frac{1}{w-i s} \frac{\text{d} s}{\sqrt{2c-s^2}},
\end{align*}
which yields the correct branch cut and sheet of the square root.
From \eqref{gasderiv} write
\begin{align*}
\widetilde{g}'_\alpha(w)&=\frac{1}{\pi}\int_{-\sqrt{2c}}^{\sqrt{2c}}\frac{-1}{w-is}+\frac{\alpha}{w-i s w^2}\frac{\text{d} s}{\sqrt{2c-s^2}}\\
&=\frac{1}{\pi}\int_{0}^{\sqrt{2c}}\frac{-1}{w-is}-\frac{1}{w+is}+\alpha\left ( \frac{1}{w-isw^2}+\frac{1}{w+isw^2}\right ) \frac{\text{d} s}{\sqrt{2c-s^2}}\\
&=: \frac{1}{\pi}\int_{0}^{\sqrt{2c}}f_\alpha(w,s) \frac{\text{d} s}{\sqrt{2c-s^2}}. 
\end{align*}
Now,
\begin{align}
\frac{1}{t}\mathcal{R} f_\alpha(i\beta +t,s)&=\frac{-1}{t^2+(\beta-s)^2}-\frac{1}{t^2+(s+\beta)^2}\label{tempz4z3}\\
&+\alpha \left (\frac{1+2 s\beta}{t^2(1+2s\beta)^2+(\beta-s(t^2-\beta^2))^2}+\frac{1-2s\beta}{t^2(1-2s\beta)^2+(\beta+s(t^2-\beta^2))^2}\right )\nonumber\\
&=\frac{-1}{s^2-2s\beta+t^2+\beta^2}+\frac{-1}{s^2+2s\beta+t^2+\beta^2}\label{realfovert}\\
&+\alpha \left (\frac{1+2 s\beta}{(t^2+\beta^2)(1+2s\beta+s^2(t^2+\beta^2))}+\frac{1-2s\beta}{(t^2+\beta^2)(1-2s\beta+s^2(t^2+\beta^2))}\right ).
\nonumber
\end{align}
Let $t<\beta$, the fact that the denominators in \eqref{realfovert} are greater than zero and $1\pm2s\beta+s^2(t^2+\beta^2)\geq 1\pm2s\beta$  yields
\begin{align*}
\frac{1}{t}\mathcal{R} f_\alpha(i\beta +t,s)\leq \frac{-1}{t^2+(s-\beta)^2}+\frac{-1}{t^2+(s+\beta)^2}+\frac{2\alpha}{t^2+\beta^2}.
\end{align*}
The assertion that the previous expression is less than zero is equivalent to
\begin{align*}
\frac{2\alpha}{t^2+\beta^2}s^4+(4\alpha\frac{t^2-\beta^2}{t^2+\beta^2}-2)s^2-2(t^2+\beta^2-2)<0.
\end{align*}
The quartic above factorises into the form $2\alpha(s^2-C_-)(s^2-C_+)/(t^2+\beta^2)$ where explicitly
\begin{align*}
C_{\pm}&=\frac{1}{4\alpha}\left (2t^2(1-2\alpha)+2\beta^2(1+2\alpha))\pm\sqrt{(2t^2(1-2\alpha)+2\beta^2(1+2\alpha))^2+16\alpha(t^2+\beta^2)^2(1-\alpha)}\right ),
\end{align*}
clearly $C_-<0$ and $C_+>\frac{1}{2\alpha}(t^2(1-2\alpha)+\beta^2(1+2\alpha))>t^2(1/\alpha -1)+\beta^2>2c$ where the second last inequality holds since $t<\beta$.

For $t\geq \beta$, the function $x\rarrow 1/(1+x)$ is convex for $x>-1$ so the sum of the first two terms in \eqref{realfovert} is bounded above by
\begin{align*}
&\frac{-1}{(t^2+\beta^2)s^2-2s\beta+t^2+\beta^2}+\frac{-1}{(t^2+\beta^2)s^2+2s\beta+t^2+\beta^2}\\&\leq \frac{1}{t^2+\beta^2}\left (\frac{-1}{1+s^2-\frac{2s\beta}{t^2+\beta^2}}+\frac{-1}{1+s^2+\frac{2s\beta}{t^2+\beta^2}}\right )\\
&\leq \frac{1}{t^2+\beta^2}\left( \frac{-2}{1+s^2}\right),
\end{align*}
so from \eqref{realfovert}
\begin{align*}
(t^2+\beta^2)\mathcal{R} f_\alpha(i\beta +t,s)/t<\frac{-2}{1+s^2}+1+\frac{1-2s\beta}{1-2s\beta+s^2(t^2+\beta^2)}.
\end{align*}
The assertion that the previous expression is less than zero is now equivalent to
\begin{align*}
s^2((t^2+\beta^2)s^2-4\beta s+2-t^2+\beta^2)<0.
\end{align*}
The previous quartic factorises into the form $(t^2+\beta^2)s^2(s-B_-)(s-B_+)$ where 
\begin{align*}
B_\pm=\frac{2\beta}{t^2+\beta^2}\pm\sqrt{1+2\frac{2\beta^2-1}{t^2+\beta^2}}.
\end{align*}
It is obvious that since $\beta\geq1$, $B_+>1>\sqrt{2c}$ and $B_-\leq0$ is equivalent to 
\begin{align*}
4\frac{\beta^2}{(t^2+\beta^2)^2}\leq1+2\frac{2\beta^2-1}{t^2+\beta^2}
\end{align*}
which is true when $t\geq \beta$. So the mapping $(0,\infty)\times [0,\sqrt{2c}]\rarrow \R$ such that $(t,s)\rarrow\mathcal{R}f_\alpha(i\beta+t,s)$ is negative which proves the first result. The extension to $\beta\in[1,\infty)$ when $\alpha=0$ is immediate from \eqref{tempz4z3}.
\end{proof}
\label{lemmadec}
\end{lemma}
Since $\mathcal{R}[\widetilde{g}_\alpha(i\beta-t)]=\mathcal{R}[\widetilde{g}_\alpha(i\beta+t)]$ for $t\in(-\infty,\infty)$, lemma \ref{lemmadec} also proves $(0,\infty)\rarrow \R;$ $t\rarrow \mathcal{R}[\widetilde{g}_\alpha(i\beta-t)]$ is decreasing. 

\begin{proof}[Proof of Proposition \ref{propEKL}]
We perform a saddle point analysis on the right hand side of \eqref{tempz4e6} and begin by proving the first statement. Consider \eqref{gasderiv}, a computation reveals
\begin{align*}
\widetilde{g}_\alpha'(w)+w\widetilde{g}_\alpha''(w)=2c\left (\frac{-1}{(w^2+2c)^{3/2}}-\frac{\alpha}{w^2}\frac{1}{(1/w^2+2c)^{3/2}}\right )
\end{align*}
and since $\widetilde{g}_\alpha'(iw_\alpha)=0$, further computation yields
\begin{align}
\widetilde{g}''_\alpha(iw_\alpha)=\frac{-2c}{(w_\alpha^2-2c)^{3/2}}\left (\frac{1}{w_\alpha}+\frac{w_\alpha^3}{\alpha^2}\right ).\label{g''alphaas}
\end{align}
Take $\eps>0$ small and less than $\min(1/\sqrt{2c}-w_\alpha, w_\alpha-\sqrt{2c})$ and write \eqref{tempz4e6} as
\begin{align*}
E_{k,\ell}=C_{k,\ell}\left (\int_{\gamma_1(\eps/2)} \frac{dw}{w}\frac{\exp[{\abs{\ell} \widetilde{g}_{\alpha}(w)]}}{\sqrt{w^2+2c}\sqrt{1/w^2+2c}}+\int_{\gamma_1(\infty)\setminus \gamma_1(\eps/2)} \frac{dw}{w}\frac{\exp[{\abs{l} \widetilde{g}_{\alpha}(w)]}}{\sqrt{w^2+2c}\sqrt{1/w^2+2c}}\right )
\end{align*}
where $C_{k,\ell}=\frac{(1+(-1)^{|\ell|+|k|})i^{-|k|-|\ell|}}{2(1+a^2)2\pi i}$. We parametrise $\gamma_1(\infty)$ by 
\begin{align*}
w(t)=iw_\alpha-t, \quad t\in(-\infty,\infty).
\end{align*}
Taylor's theorem yields 
\begin{align}
\widetilde{g}_\alpha(w(t))-\widetilde{g}_\alpha(iw_\alpha)=\widetilde{g}_\alpha ''(iw_\alpha)\frac{t^2}{2}+t^3R_{9}(t,\alpha),\label{giwalphaexpand}
\end{align}
where
\begin{align*}
R_{9}(t,\alpha)=\frac{1}{2\pi i} \int_{\partial\B_{\eps}(iw_\alpha)}dz\frac{\widetilde{g}_{\alpha}(z)}{(z-w(t))(z-iw_\alpha)^{3}}(=:R_{9}(t)).
\end{align*}
Here $\partial \B(iw_\alpha,\eps)=\{\eps e^{i\theta} + iw_\alpha: \theta\in[0,2\pi)\}$, so $\abs{R_{9}(t,\alpha)}\leq C$ for $t\in[-\eps/2,\eps/2]$ and all $\alpha$. Since $g_{\alpha}''(iw_\alpha)$ is less than some negative number for all $\alpha$, we can take $\eps$ so small that
\begin{align}
g_{\alpha}''(iw_\alpha)/2+|tR_{9}(t,\alpha)|<-b\label{temp4txcv}
\end{align}
for some $b>0$ uniformly in $\alpha$ and $t\in[-\eps/2,\eps/2]$.
Setting $\beta=w_\alpha$ in lemma \ref{gasdesc}, we have a descent contour so it follows from \eqref{giwalphaexpand} and \eqref{temp4txcv} that there are positive constants $C_1, C_2$ so that
\begin{align*}
\abs{\int_{\gamma_1(\infty)\setminus \gamma_1(\eps/2)} \frac{dw}{w}\frac{\exp[{\abs{\ell}(\widetilde{g}_{\alpha}(w)- \widetilde{g}_{\alpha}(iw_\alpha))]}}{\sqrt{w^2+2c}\sqrt{1/w^2+2c}}}&\leq C_1\sup_{w\in\gamma_1(\infty)\setminus\gamma_1(\eps/2)}\exp[-\abs{\ell}\abs{\mathcal{R}(\widetilde{g}_\alpha(w)-\widetilde{g}_\alpha(iw_\alpha))}]\\
&\leq C_2e^{-b\abs{\ell}\eps^2/4}.
\end{align*}
For the integral over $\gamma_1(\eps/2)$, using the same parametrisation, \eqref{giwalphaexpand} gives
\begin{align}
&\int_{\gamma_1(\eps/2)} \frac{dw}{w}\frac{\exp[{\abs{\ell} \widetilde{g}_\alpha(w)]}}{\sqrt{w^2+2c}\sqrt{1/w^2+2c}}\nonumber\\
&=-\frac{\exp[\abs{\ell}\widetilde{g}_\alpha(iw_\alpha)]}{V(iw_\alpha)}\int_{-\eps/2}^{\eps/2}\exp[\abs{\ell}(\widetilde{g}_\alpha''(iw_\alpha)t^2/2+t^3R_{9}(t))]V(iw_\alpha)/V(iw_\alpha-t)dt\label{tempint12}
\end{align}
where $V(w)=w\sqrt{w^2+2c}\sqrt{1/w^2+2c}$. We require three bounds. Taylors theorem applied to $t\rarrow V(iw_\alpha)/V(iw_\alpha-t)$ at $t=0$ yields
\begin{align}
&\abs{\int_{-\eps/2}^{\eps/2}\exp[\abs{\ell}(\widetilde{g}_\alpha''(iw_\alpha)t^2/2+t^3R_{9}(t))]\frac{V(iw_\alpha)}{V(iw_\alpha-t)}dt-\int_{-\eps/2}^{\eps/2}\exp[\abs{\ell}(\widetilde{g}_\alpha''(iw_\alpha)t^2/2+t^3R_{9}(t))]dt}\nonumber\\
&\leq\int_{-\eps/2}^{\eps/2}|t\exp[\abs{\ell}(\widetilde{g}_\alpha''(iw_\alpha)t^2/2+t^3R_{9}(t))]|dt\leq C_3\int_{-\eps/2}^{\eps/2}|t|e^{-|\ell|bt^2}dt\leq \frac{C_4}{|\ell|^{3/2}}.\nonumber
\end{align}
We use the bound $|e^{t}-1|\leq|t|e^{|t|}$  to get
\begin{align}
&\abs{\int_{-\eps/2}^{\eps/2}\exp[\abs{\ell}(\widetilde{g}_\alpha''(iw_\alpha)t^2/2+t^3R_{9}(t))]dt-\int_{-\eps/2}^{\eps/2}\exp[\abs{\ell}(\widetilde{g}_\alpha''(iw_\alpha)t^2/2)]dt}\nonumber\\
&\leq \int_{-\eps/2}^{\eps/2}\Big|e^{|\ell|t^3R_{9}(t)}-1\Big|\exp[\abs{\ell}(\widetilde{g}_\alpha''(iw_\alpha)t^2/2)]dt\leq |\ell|\int_{-\eps/2}^{\eps/2}|t|^3\exp[|\ell|(g''_\alpha(iw_\alpha)t^2/2+|t^3R_{9}(t)|)]dt\nonumber\\
&\leq   C_5 |\ell|\int_{-\eps/2}^{\eps/2}|t|^3\exp[-|\ell|bt^2]dt\leq \frac{C_6}{|\ell|}.\nonumber
\end{align}
Finally,
\begin{align}
\abs{\int_{-\eps/2}^{\eps/2}\exp[|\ell|g_\alpha''(iw_\alpha)t^2/2]dt-\int_{-\infty}^{\infty}\exp[|\ell|g_\alpha''(iw_\alpha)t^2/2]dt}\leq C_7\frac{e^{-|\ell|b\eps^2/4}}{\eps}. 
\label{tempdiff12z}
\end{align}
Now take the absolute value of the integral in \eqref{tempint12} minus the second term of the difference in \eqref{tempdiff12z}. Write this as a sum of differences above, then use the triangle inequality and the three bounds. The main factor in \eqref{EklAs1} comes from the term $\frac{\exp[\abs{\ell}\widetilde{g}_\alpha(iw_\alpha)]}{V(iw_\alpha)}$ in \eqref{tempint12}.
 \\\\
The case $\alpha=|k|/|\ell|$ for fixed $k$ is handled differently.  Recall the integral in \eqref{tempz4e6e4} and write it as
\begin{align}
-\int_{\gamma'_{\eta_1,\eta_2}(\infty)} \frac{dw}{w}\frac{\exp[{\abs{\ell} \widetilde{g}_{0}(w)]}G(1/w)^{|k|}}{\sqrt{w^2+2c}\sqrt{1/w^2+2c}}.
\label{temprewrite232}
\end{align}
The minus sign in \eqref{temprewrite232} appears by defining $\gamma'_{\eta_1,\eta_2}(\infty)$ to be $\gamma_{\eta_1,\eta_2}(\infty)$ with reverse orientation.
In this case the asymptotics come from the branch point $i/\sqrt{2c}$, and the saddle point function $\widetilde{g}_0$ is analytic at this point.
Parametrise $w(t)=i(t+1/\sqrt{2c}) +\eta_2$ for $t\in[0,\eta_1]$ and observe that
\begin{align}
\sqrt{1/w(t)^2+2c}\rarrow \sqrt{t}\sqrt{2c/(t+1/\sqrt{2c})+\sqrt{2c}/(t+1/\sqrt{2c})^2}
\end{align} as $\eta_2\rarrow 0^+$. So
\begin{align}
\frac{G(1/w(t))^{|k|}}{w(t)\sqrt{w(t)^2+2c}\sqrt{1/w(t)^2+2c}}\rarrow \frac{f(t)}{\sqrt{t}}, \quad \text{as} \ \eta_2\rarrow 0^+
\end{align} for a function $f$. A computation gives
\begin{align}
f(0)=-G(-i\sqrt{2c})^{|k|}\sqrt{1-4c^2}\sqrt{2c+1}. 
\end{align} 
 We can get a bounded function $R_{10}(t)$ such that
\begin{align}
f(t)-f(0)=tR_{10}(t)\label{temptay2cs}
\end{align}
for $t\in[0,\eta_1]$.
From \eqref{gasderiv} we also have a bounded function $R_{11}(t)$ such that
\begin{align}
\widetilde{g}_0(i/\sqrt{2c}+it)-\widetilde{g}_0(i/\sqrt{2c})=-t\sqrt{2c/(1-4c^2)}+t^2R_{11}(t)\label{tempbound2ac}
\end{align}
for $t\in[0,\eta_1]$.
Let $b'=\sqrt{2c/(1-4c^2)}$.
By lemma \ref{gasdesc}, the straight line from $i/\sqrt{2c}+i\eta_1+\eta_2$ to infinity is a descent contour. Take $\eta_1>0$ so small that 
\begin{align}
-t\sqrt{2c/(1-4c^2)}+t^2|R_{11}(t)|<-t\sqrt{2c/(1-4c^2)}/2.
\end{align}
Then, by \eqref{temptay2cs} and \eqref{tempbound2ac}, there are positive constants $C_4,C_5$ such that
\begin{align}
\abs{\int_{i/\sqrt{2c}+i\eta_1+\R_{>0}} \frac{dw}{w}\frac{\exp[{\abs{\ell} \widetilde{g}_{0}(w)]}G(1/w)^{|k|}}{\sqrt{w^2+2c}\sqrt{1/w^2+2c}}}\leq C_4e^{|\ell|\mcR[\widetilde{g}_0(i/\sqrt{2c})]-C_5|\ell|\eta_1}.
\end{align}
Take $\eta_2\rarrow 0^+$ in \eqref{temprewrite232}. We consider a sequence of approximations and bound their differences.
First, we have the estimate,
\begin{align}
&\abs{\int_0^{\eta_1}\frac{dt}{\sqrt{t}}f(t)e^{|\ell|\widetilde{g}_0(i/\sqrt{2c}+it)}-\int_0^{\eta_1}\frac{dt}{\sqrt{t}}f(0)e^{|\ell|\widetilde{g}_0(i/\sqrt{2c}+it)}}\nonumber\\
&\leq C\int_0^{\eta_1}dt \sqrt{t} e^{|\ell|\mcR\widetilde{g}_0(i/\sqrt{2c}+it)}\leq Ce^{|\ell|\mcR\widetilde{g}_0(i/\sqrt{2c})}\int_0^{\eta_1}dt \sqrt{t} e^{-|\ell|b't/2}\leq C'\frac{e^{|\ell|\mcR\widetilde{g}_0(i/\sqrt{2c})}}{|\ell|^{3/2}},\nonumber
\end{align}
secondly,
\begin{align}
&\abs{\int_0^{\eta_1}\frac{dt}{\sqrt{t}}f(0)e^{|\ell|\widetilde{g}_0(i/\sqrt{2c}+it)}-\int_0^{\eta_1}\frac{dt}{\sqrt{t}}f(0)e^{|\ell|(\widetilde{g}_0(i/\sqrt{2c})-b't)}}\nonumber\\
&\leq \int_0^{\eta_1}\frac{dt}{\sqrt{t}}f(0)e^{|\ell|(\mcR\widetilde{g}_0(i/\sqrt{2c})-b't)}\Big |e^{R_{11}(t)t^2}-1\Big|\leq C'''\int_0^{\eta_1}dt \ t^{3/2}e^{|\ell|(\mcR\widetilde{g}_0(i/\sqrt{2c})-b't+t^2|R_{11}(t)|)}\nonumber\\
&\leq C''''e^{|\ell|\mcR\widetilde{g}_0(i/\sqrt{2c})}\int_0^{\eta_1}dt t^{3/2}e^{-|\ell|b't/2}\leq C^{(v)}\frac{e^{|\ell|\mcR\widetilde{g}_0(i/\sqrt{2c})}}{|\ell|^{5/2}},\nonumber
\end{align}
and finally,
\begin{align}
&\abs{\int_0^{\eta_1}\frac{dt}{\sqrt{t}}f(0)e^{|\ell|\widetilde{g}_0(i/\sqrt{2c})-\ell b't}-\int_0^{\infty}\frac{dt}{\sqrt{t}}f(0)e^{|\ell|\widetilde{g}_0(i/\sqrt{2c})-\ell b't}}\leq C^{(vi)}e^{|\ell|\mcR\widetilde{g}_0(i/\sqrt{2c})}\frac{e^{-|\ell|b'\eta_1}}{|\ell|b'\sqrt{\eta_1}}.\nonumber
\end{align}
Furthermore,
\begin{align}
\int_0^\infty \frac{dt}{\sqrt{t}}f(0)e^{|\ell|\widetilde{g}_0(i/\sqrt{2c})-\ell b't}=-\sqrt{\frac{\pi}{|\ell|}}G(\frac{i}{\sqrt{2c}})^{|\ell|}G(-i\sqrt{2c})^{|k|}(1-4c^2)^{3/2}\sqrt{\sqrt{2c}+1/\sqrt{2c}}.\label{temp1das}
\end{align}
Now note that when $\ell+k$ is even, $G(\frac{i}{\sqrt{2c}})^{|\ell|}G(-i\sqrt{2c})^{|k|}$ is a real number. The main contribution to the integral in \eqref{tempz4e6e4} is the right hand side of \eqref{temp1das}. If we write the difference between the main contribution and the integral as sum of the differences above, then by the triangle inequality, \eqref{EklAs2} holds.
\end{proof} 

\begin{remark}
We can rewrite the expressions for $E_{k,l}$ in lemma \ref{propEKL} as
\begin{align}
E_{k,l}=\frac{\cos(\pi({|k|+|\ell|})/2)G(iw_\alpha)^{\max{(|k|,|\ell|)}}G(1/(iw_\alpha))^{\min{(|k|,|\ell|)}}}{(1+a^2)\sqrt{2\pi |\tilde{g}''_\alpha(iw_\alpha)\max(|k|,|\ell|)|}w_\alpha\sqrt{w_\alpha^2+2c}{\sqrt{1/w_\alpha^2+2c}}}(1+o(1))
\label{temprewritesymmEKL}
\end{align}
when $\alpha=\min{(\frac{|k|}{|\ell|},\frac{|\ell|}{|k|})}$ varies in compact subset of $(0,1]$.
If  instead one of $k$ or $\ell$ is fixed, then 
\begin{align}
E_{k,l}=\frac{\cos(\frac{\pi\max(|k|,|\ell|)}{2})G(\frac{i}{\sqrt{2c}})^{\max(|k|,|\ell|)}G(-i\sqrt{2c})^{\min(|k|,|\ell|)}}{2(1+a^2)\sqrt{\pi\max(|k|,|\ell|)}}\frac{\sqrt{(1-4c^2)^{3}(2c+1)}}{(2c)^{1/4}}(1+o(1)).
\end{align}
\end{remark}
We know that $|G(i)|<1$, from which we can discern that the asymptotics of $E_{k,l}$ are exponentially decaying when $\alpha=1$, since then $iw_\alpha=i$ so $\mcR[\tilde{g}_1(iw_1)]=2\log|G(i)|<0$. Here we give a  lemma that shows $E_{k,l}$ is exponentially decaying for $0\leq \alpha\leq 1$.
\begin{lemma} For $0\leq \alpha< 1$,
\begin{align}
\mcR[\tilde{g}_{\alpha}(iw_\alpha)]\leq(1+\alpha)\log|G(i)|<0.
\end{align}
\label{galphaalpha}
\begin{proof}
Consider the integral
\begin{align}
\int_{\sqrt{2c}}^x\frac{du}{\sqrt{u^2-2c}}=-\log\Big(\frac{x}{\sqrt{2c}}-\sqrt{\Big(\frac{x}{\sqrt{2c}}\Big)^2-1}\Big).
\end{align}
We see that
\begin{align}
\mcR[\tilde{g}_\alpha(iw_\alpha)]=-\int_{\sqrt{2c}}^{w_\alpha}\frac{du}{\sqrt{u^2-2c}}-\alpha\int_{\sqrt{2c}}^{1/w_\alpha}\frac{du}{\sqrt{u^2-2c}}
\end{align}
and 
\begin{align}
(1+\alpha)\log|G(i)|=-(1+\alpha)\int_{\sqrt{2c}}^1\frac{du}{\sqrt{u^2-2c}}
\end{align}
where $1\leq w_\alpha\leq 1/\sqrt{2c}$.
Define
\begin{align}
s(\alpha)&:=\mcR\tilde{g}_\alpha(iw_\alpha)-(1+\alpha)\log|G(i)|\\
&=-\int_1^{w_\alpha}\frac{du}{\sqrt{u^2-2c}}+\alpha\int_{1/w_\alpha}^1\frac{du}{\sqrt{u^2-2c}}.\nonumber
\end{align}
We want to show that $s(\alpha)\leq 0$.
Substituting $u=1/v$ yields
\begin{align}
\int_{1/w_\alpha}^1\frac{du}{\sqrt{u^2-2c}}=\int_1^{w_\alpha}\frac{1}{\sqrt{1-2cv^2}}\frac{dv}{v},
\end{align}
and hence 
\begin{align}
s(\alpha)=\int_1^{w_\alpha}\frac{\alpha}{u\sqrt{1-2cu^2}}-\frac{1}{\sqrt{u^2-2c}}du.\label{tempintsalpha}
\end{align}
For $1\leq u\leq 1/\sqrt{2c}$,
\begin{align}
\frac{\alpha}{u\sqrt{1-2cu^2}}\leq \frac{1}{\sqrt{u^2-2c}},
\end{align}
which is equivalent to
\begin{align}
2cu^4-(1-\alpha^2)u^2-2c\alpha^2\leq 0
\end{align}
or
\begin{align}
\frac{1-\alpha^2-\sqrt{(1-\alpha^2)^2+16c^2\alpha^2}}{4c}\leq u^2\leq w_\alpha^2
\end{align}
Hence the integrand in \eqref{tempintsalpha} is less than or equal to zero, so $s(\alpha)\leq 0$.
\end{proof}
\end{lemma}
Finally, we have a short lemma that we use to extract the leading order terms of $\K_{1,1}^{-1}$ from Proposition \ref{propEKL}.
\begin{lemma}\label{alphaalpha'}
For $0< \alpha,\alpha'\leq 1$,
\begin{align}
\tilde{g}_\alpha(iw_\alpha)=\tilde{g}_{\alpha'}(iw_{\alpha'})+\log(G(1/(iw_{\alpha'}))(\alpha-\alpha')+O(\alpha-\alpha')^2.
\end{align}
\begin{proof}
Since $\alpha\rarrow \tilde{g}_\alpha(iw_\alpha)$ is smooth, we use Taylor's theorem and note that
\begin{align}
\left. {\frac{d}{d\alpha}}\right |_{\alpha'}\tilde{g}_\alpha(iw_\alpha)&=\log(G(1/(iw_{\alpha'})))+i\frac{dw_{\alpha'}}{d\alpha}\left.\frac{d}{dw}\right |_{iw_{\alpha'}}\tilde{g}_{\alpha'}(w)\nonumber
\end{align}
by the chain rule. The lemma follows since $iw_{\alpha'}$ is a critical point of $\tilde{g}_{\alpha'}$.
\end{proof}
\end{lemma}

\section{Uniform asymptotics of $C_{\omega_c}$}\label{Comegacasymptsection}
The goal of this section is to prove propositions \ref{Comegacinnerdisc} and \ref{propeklcomegac},
first we prove a lemma about descent paths.
\begin{lemma}
$\abs{G(e^{i\theta})}$ is strictly increasing for $\theta\in(0,\pi/2)$.
\label{Ginc}
\begin{proof}
Since $\abs{G(e^{i\theta})}^2=G(e^{i\theta})G(e^{-i\theta})$, taking the logarithm we can see the statement of this lemma is equivalent to  assertion that the function $\{e^{i\theta}: \theta\in (0,\pi/2)\}\rarrow \R$ such that $\widetilde{g}_1(w)=\log(G(w))+\log(G(1/w))$ is strictly increasing over $\theta$. From
\begin{align*}
w\widetilde{g}_1'(w)=-\frac{w}{\sqrt{w^2+2c}}+\frac{1}{w}\frac{1}{\sqrt{1/w^2+2c}},
\end{align*}
changing variables $w(\theta)=e^{i\theta}$ and from $\sqrt{(\overline{w})^2+2c}=\overline{\sqrt{w^2+2c}}$,
\begin{align*}
\frac{d}{d\theta} \widetilde{g}_1(e^{i\theta})= 2 \ \mathcal{I}\left(\frac{e^{i\theta}}{\sqrt{(e^{i\theta})^2+2c}}\right).
\end{align*}
Once again we use the integral representation of the reciprocal  square root to write
\begin{align*}
\mathcal{I}\left(\frac{e^{i\theta}}{\sqrt{(e^{i\theta})^2+2c}}\right)&=\frac{1}{\pi}\int_{-\sqrt{2c}}^{\sqrt{2c}}\mathcal{I}\left (\frac{e^{i\theta}}{e^{i\theta}-is}\right )\frac{ds}{\sqrt{2c-s^2}}\\
&=\frac{1}{\pi}\int_0^{\sqrt{2c}}\frac{4s^2 \sin(\theta)\cos(\theta)}{|e^{i\theta}+is|^2|e^{i\theta}-is|^2}\frac{ds}{\sqrt{2c-s^2}}
\end{align*}
and since the integrand is positive for $0<s<\sqrt{2c}$, $0<\theta<\pi/2$ the lemma follows.
\end{proof}
\end{lemma}
\begin{corollary}
\label{corr}
For $\alpha \in (-1,\infty)$ ($\alpha\in (-\infty,-1)$) the function $(0,\pi/2)\ni\theta\mapsto \mathcal{R}[g_\alpha(ie^{-i\theta})]$ is strictly decreasing (strictly increasing). If $\alpha=-1$ the same function is zero.
\begin{proof}
Due to lemma \ref{Gsymmetries}
\begin{align}
\mathcal{R}[g_\alpha(ie^{-i\theta})]&=\log\abs{G(ie^{-i\theta})}+\alpha\log\abs{G(1/(ie^{-i\theta}))}\nonumber\\
&=(1+\alpha)\log|G(e^{i(\pi/2-\theta)})|\nonumber
\end{align}
so the statement follows from lemma \ref{Ginc}.
\end{proof}
\label{galphadec}
\end{corollary}

We write
\begin{align}
\psi(\theta):=\widetilde{g}_{\alpha}(e^{i\theta})=(1+\alpha)\log|G(e^{i\theta})|+i(1-\alpha)\arg(G(e^{i\theta})).\label{psithetaalpha}
\end{align}
 to shorten the expressions.
  We require a few facts and approximations that will be used multiple times in the following proofs. From \eqref{psithetaalpha} and \eqref{Gsymmetries} we have
 \begin{align}
 \psi(\theta)=\overline{\psi(\pi-\theta)}+i\pi(1-\alpha),\label{psisym}
 \end{align}
 from which we see that $\mcI[\psi(\pi/2+\theta)]$ is odd and $\mcR[\psi(\pi/2+\theta)]$ is even on $[-\pi/2,\pi/2]$. In particular, \eqref{psisym} gives
 \begin{align}
 \mathcal{I}[\psi(\pi/2)]=(1-\alpha)\pi/2 \quad \text{and} \quad \mcR[\psi'(\pi/2)]=\mcR[\psi'''(\pi/2)]=0\label{psievals1}
 \end{align}
   A computation yields \begin{align}
\psi'(\theta)=(1+\alpha)\mathcal{R}\big [\frac{-ie^{i\theta}}{\sqrt{(e^{i\theta})^2+2c}}\big ]+i(1-\alpha)\mathcal{I}\big [ \frac{-ie^{i\theta}}{\sqrt{(e^{i\theta})^2+2c}}\big ]
\label{psi'1}
\end{align}
from which we get 
\begin{align}
\psi'(\pi/2)=-i\frac{(1-\alpha)}{\sqrt{1-2c}}.\label{psievals2}
\end{align} Further computation gives 
\begin{align}
\psi''(\pi/2)=-\frac{2c(1+\alpha)}{(1-2c)^{3/2}}.\label{psievals3}
\end{align} 
Taylor's theorem gives a bounded function $R_{12}(\theta,\theta_c,\alpha)$ such that
\begin{align}
\psi(\theta)-\psi(\theta_c)=\psi'(\theta_c)(\theta-\theta_c)+\psi''(\theta_c)(\theta-\theta_c)^2/2+\psi'''(\theta_c)(\theta-\theta_c)^3/3!+R_{12}(\theta,\theta_c,\alpha)(\theta-\theta_c)^4,
\label{taylorpsithetac4fc}
\end{align}
and there are also bounded functions $R_{13}(\theta_c,\alpha), R_{14}(\theta_c,\alpha),R_{15}(\theta_c,\alpha)$ such that
\begin{align}
\label{taylorpsi''thetac1e1}\psi'(\theta_c)&=\psi'(\pi/2)-\psi''(\pi/2)\varphi_c+\frac{\psi'''(\pi/2)}{2}\varphi_c^2 +R_{13}(\theta_c,\alpha)\varphi_c^3,\\
\psi''(\theta_c)&=\psi''(\pi/2)-\psi'''(\pi/2)\varphi_c+R_{14}(\theta_c,\alpha)\varphi_c^2,\nonumber\\
\psi'''(\theta_c)&=\psi'''(\pi/2)+R_{15}(\theta_c,\alpha)\varphi_c\nonumber.
\end{align}
From \eqref{psievals1}, \eqref{taylorpsithetac4fc} and \eqref{taylorpsi''thetac1e1} we obtain
\begin{align}
\label{realtaylor5x1q}\mcR[\psi(\theta)-\psi(\theta_c)]=[\varphi_c\psi''(\pi/2)+\mcR[R_{13}(\theta_c,\alpha)]\varphi_c^3](\theta_c-\theta)
+[\psi''(\pi/2)+\mcR[R_{14}(\theta_c,\alpha)]\varphi_c^2]\frac{(\theta_c-\theta)^2}{2}\\-\mcR[R_{15}(\theta_c,\alpha)]\varphi_c(\theta_c-\theta)^3+\mcR[R_{12}(\theta,\theta_c,\alpha)](\theta_c-\theta)^4.\nonumber
\end{align}
We have a bounded function $R_{16}(\theta,\alpha)$ such that
\begin{align}
\psi(\pi/2+\theta)=\psi(\pi/2)+\psi'(\pi/2)\theta+\psi''(\pi/2)\theta^2/2+\theta^3R_{16}(\theta,\alpha),\label{psitaylorpi2}
\end{align}
which using \eqref{psievals1}, \eqref{psievals2} gives a bounded function $R_{8}(\theta)$ on $\R$ such that
\begin{align}
\mcI[\psi(\pi/2+\theta)]=(1-\alpha)(\pi/2-\theta/\sqrt{1-2c}+R_{8}(\theta)\theta^3)\label{tempImpsi12}
\end{align}
We are now ready for the proof of Proposition \ref{Comegacinnerdisc}.
\begin{proof}[Proof of Proposition \ref{Comegacinnerdisc}]
We parametrise $w(\theta)=e^{i\theta}$ and use the fact that $l+k$ is even to write
\begin{align}
\tilde{C}_{\omega_c}(k,\ell)&=\frac{i^{-k-\ell}}{2(1+a^2)2\pi i}(1+(-1)^{\ell+k})\int_{\Gamma_{\omega_c}\cap\HH}\frac{dw}{w}\frac{G(w)^\ell G(1/w)^k}{\sqrt{w^2+2c}\sqrt{1/w^2+2c}}\label{temp234321}\\
&=\frac{i^{-k-\ell}}{2(1+a^2)\pi }\int_{\theta_c}^{\pi-\theta_c}\frac{\exp{(\ell \psi(\theta))}}{|e^{2i\theta}+2c|}d\theta
\nonumber
\end{align}
We rewrite \eqref{psitaylorpi2} to get a $R_{17}(\theta,\alpha)$ such that
\begin{align}
\psi(\theta)-\psi(\pi/2)=-\psi'(\pi/2)(\pi/2-\theta)+R_{17}(\theta,\alpha)(\pi/2-\theta)^2.
\label{tempz4f4}
\end{align} We use \eqref{tempz4f4} and the bound $|e^{t}-1|\leq |t|e^{|t|}$ to get constants $C_1,C_2>0$ such that
\begin{align}
&\abs{\int_{\theta_c}^{\pi-\theta_c}\frac{\exp{(\ell \psi(\theta))}}{|e^{2i\theta}+2c|}d\theta-\int_{\theta_c}^{\pi-\theta_c}\frac{\exp{(\ell \psi(\pi/2)-\psi'(\pi/2)(\pi/2-\theta))}}{|e^{2i\theta}+2c|}d\theta}\label{tempdiffb24}
\\&\leq C_1|\ell|\varphi_c^2e^{\ell\mcR[\psi(\pi/2)]}\int_{\theta_c}^{\pi-\theta_c}e^{|\ell R_{17}(\theta,\alpha)|(\pi/2-\theta)^2}d\theta\nonumber \\ &\leq C_1e^{\ell\mcR[\psi(\pi/2)]}\varphi_c^{1+\gamma}e^{C_2\varphi_c^\gamma}.\nonumber
\end{align}
Next the second term in the difference in \eqref{tempdiffb24} can be approximated by \eqref{taylorinnd3f}. We get 
\begin{align}
&\abs{\int_{\theta_c}^{\pi-\theta_c}\frac{\exp(\ell\psi(\pi/2)-\ell\psi'(\pi/2)(\pi/2-\theta))}{|e^{2i\theta}+2c|}d\theta-\int_{\theta_c}^{\pi-\theta_c}\frac{\exp(\ell\psi(\pi/2)-\ell\psi'(\pi/2)(\pi/2-\theta))}{1-2c}d\theta}\label{tempdiffr5fv}
\\
&\leq C_4e^{\ell\mcR[\psi(\pi/2)]}\varphi_c^3\leq C_4e^{l\mcR[\psi(\pi/2)]}\frac{\varphi_c^{1+\gamma}}{|\ell|}.\nonumber
\end{align}
Finally, the second them in the difference in \eqref{tempdiffr5fv} is
\begin{align}
\int_{\theta_c}^{\pi-\theta_c}\frac{\exp(\ell\psi(\pi/2)-\ell\psi'(\pi/2)(\pi/2-\theta))}{1-2c}d\theta=\frac{2e^{\ell\psi(\pi/2)}\sin((\ell-k)\varphi_c/\sqrt{1-2c})}{(\ell-k)\sqrt{1-2c}} \label{tempsincontribution3d3}
\end{align}
The two bounds \eqref{tempdiffb24}, \eqref{tempdiffr5fv} together with \eqref{tempsincontribution3d3} and \eqref{temp234321} give
a bounded function $R_{18}(\xi,\ell,k)$ such that
\begin{align}
\tilde{C}_{\omega_c}(k,\ell)=\frac{(-1)^{k}|G(i)|^{\ell+k}}{(1-a)^2\pi}\Big (\frac{\sin((\ell-k)\varphi_c/\sqrt{1-2c})}{(\ell-k)/\sqrt{1-2c}}+R_{18}(\xi,\ell,k)\varphi_c^{1+\gamma}\Big ).\label{tempCom1}
\end{align}
Next, for all $\eps>0$ small if $bx\in \Lambda_\eps$ then there is a $C(\eps)>0$ such that
\begin{align}
\abs{\frac{\sin(bx)}{x}}\geq C(\eps)\min (|b|,1/|x|)\label{tempz3e4}
\end{align}
where we used the lower bound $|\sin(x)/x|\geq 1-2|x|/\pi$ for $x\in [-\pi/2,\pi/2]$.
 Hence if $R_{6}$ is defined by \eqref{tempCom12},
\begin{align}
|R_{6}| &\leq C\varphi_c^{1+\gamma}/\min(\varphi_c/(1-2c),1/|\ell-k|)\\
&\leq C\max(\varphi_c^{\gamma}(1-2c),\varphi_c^{1+\gamma}|\ell-k|).\nonumber
\end{align} 
Now $\varphi_c^{2-\gamma}|\ell-k|\leq \varphi_c^{2-\gamma}(|\ell|+|k|)<1$ implies $\varphi_c^{1+\gamma}|\ell-k|\leq \varphi_c^{-1+2\gamma}$ and $\varphi_c^\gamma(1-2c)\leq\varphi_c^{-1+2\gamma}$, so we obtain \eqref{temp2zmhh}.
\end{proof}

Now we give three propositions which go into the proof of Proposition \ref{propeklcomegac}.

\begin{proposition}\label{Comegasin}
Let non-zero integers $\ell,k$ be such that $\ell+k\leq 2$ is even and $\alpha = k/\ell$ lies in a compact subset of $[-1,1)$. There is a bounded function $R_{19}(\xi,\ell,k)$ such that
\begin{align}
\tilde{C}_{\omega_c}(k,\ell)=\frac{(-1)^{k}|G(\omega_c)|^{\ell+k}}{\pi(1-a)^2}\Big[\frac{\sin{((\ell-k)(F(\varphi_c)))}}{(\ell-k)/\sqrt{1-2c}}+R_{19}(\xi,\ell,k)\frac{\varphi_c}{(\ell-k)}\Big].
\label{Cwcalphaminus11/2}
\end{align}
\begin{proof}
We recall \eqref{temp234321} as
\begin{align}
\tilde{C}_{\omega_c}(k,\ell)&=\frac{i^{-k-\ell}}{2(1+a^2)\pi }\int_{\theta_c}^{\pi-\theta_c}\frac{\exp{(\ell \psi(\theta))}}{|e^{2i\theta}+2c|}d\theta.
\label{temp23432}
\end{align}
Integrating by parts twice, we have the integral in \eqref{temp23432} equal to
\begin{align}
\frac{1}{\ell}\big (e^{\ell\psi(\pi-\theta_c)}b(\pi-\theta_c)-e^{\ell\psi(\theta_c)}b(\theta_c)\big )-\frac{1}{\ell^2}\big (e^{\ell\psi(\pi-\theta_c)}\frac{b'(\pi-\theta_c)}{\psi'(\pi-\theta_c)}-e^{\ell\psi(\theta_c)}\frac{b'(\theta_c)}{\psi'(\theta_c)}\big)\label{temp087}\\+\frac{1}{\ell^2}\int_{\theta_c}^{\pi-\theta_c}\exp( \ell\psi(\theta)) \frac{d}{d\theta}\Big (\frac{b'(\theta)}{\psi'(\theta)}\Big )d\theta
\nonumber
\end{align}
where
\begin{align}
 b(\theta)=1/[\psi'(\theta)|e^{2i\theta}+2c|].
\end{align}
From \eqref{psisym} we get $\psi'(\theta)=-\overline{\psi'(\pi-\theta)}$ and so $b(\theta)=-\overline{b(\pi-\theta)}$. This gives
\begin{align}
e^{\ell\psi(\pi-\theta_c)}b(\pi-\theta_c)-e^{\ell\psi(\theta_c)}b(\theta_c)
&=-2e^{\ell\mathcal{R}[\psi(\theta_c)]} \mathcal{R}\big [e^{ i\ell\mathcal{I}[\psi(\theta_c)]}b(\theta_c)\big ].\label{temp5g5}
\end{align}
Note that because of \eqref{psievals2} the term next to the exponential in the integrand in \eqref{temp087} is not defined when $\alpha=1$. It is bounded when $\alpha$ lies in a compact subset of $[-1,1)$, this follows from the change of variables $w=e^{i\theta}$ in lemma \eqref{gasderiv} and noting that $w_\alpha>1$ for $0<\alpha<1$. When $\ell+k<0$ and $\alpha<1$ we have $\ell<0$ so by corollary \ref{galphadec}, $l\mcR[\psi(\theta)]$ achieves its maximum over $[\theta_c,\pi-\theta_c]$ at the endpoints where $l\mcR[\psi(\theta_c)]=l\mcR[\psi(\pi-\theta_c)]$. For $\ell+k\in \{0,1,2\}$, $\ell \mcR[\psi(\theta)]$ is bounded trivially on $[\theta_c,\pi-\theta_c]$. From this we obtain a bounded function $R_{20}(\xi,\ell,k)$ such that 
\begin{align}
\int_{\theta_c}^{\pi-\theta_c}\frac{\exp{(\ell \psi(\theta))}}{|e^{2i\theta}+2c|}d\theta&=\frac{-2}{\ell}e^{\ell\mathcal{R}[\psi(\theta_c)]} \mathcal{R}\big [e^{ i\ell\mathcal{I}[\psi(\theta_c)]}b(\theta_c)\big ]+\varphi_c\frac{e^{\ell\mathcal{R}[\psi(\theta_c)]}}{\ell^2(1-\alpha)^2}R_{20}(\xi,\ell,k).\label{temp12412}
\end{align}
By \eqref{taylorpsi''thetac1e1} and \eqref{temomo} there is a bounded function $R_{21}(\xi,\alpha)$ such that
\begin{align}
|b(\theta_c)-b(\pi/2)|\leq 2\frac{|R_{21}(\xi,\alpha)|\varphi_c}{(1-2c)^2|\psi'(\pi/2)|^2}.\label{temptaylorbound2ss}
\end{align}
From \eqref{tempImpsi12}, \eqref{temptaylorbound2ss} and $b(\pi/2)=i/[(1-\alpha)\sqrt{1-2c}]$ there are bounded functions $R_{22}(\xi,\alpha), R_{23}(\xi,\ell,k)$ such that \eqref{temp12412} equals
\begin{align}
\label{tempsdf3} 2i^{\ell-k} e^{\ell\mathcal{R}[\psi(\theta_c)]}\Big(\Big[\frac{\sin{[(\ell-k)\varphi_c/\sqrt{1-2c}-(\ell-k)R_{7}(-\varphi_c)\varphi_c^3]}}{(\ell-k)\sqrt{1-2c}}\\\Big (1+R_{22}(\xi,\alpha)\varphi_c\Big)\Big ]+R_{23}(\xi,\ell,k)\frac{\varphi_c}{(\ell-k)^2}\Big).\nonumber
\end{align}
Hence \eqref{Cwcalphaminus11/2} follows by \eqref{temp23432}, \eqref{temp12412} and \eqref{tempsdf3}.
\end{proof}
\end{proposition}

 \begin{proposition}\label{steepestdescComega}
Let non-zero integers $\ell,k$ be such that $\ell+k< 0$ is even and $\alpha = k/\ell$ lies in a compact subset of $(-1,1]$. There is a bounded function $R_{24}(\xi,\ell,k)$ such that 
\begin{align}
\tilde{C}_{\omega_c}(k,\ell)&=\frac{(-1)^{k}  \ |G(\omega_c)|^{\ell+k}}{\pi(1-a)^2}\Big[e^{\frac{(\ell+k)c\varphi_c^2}{(1-2c)^{3/2}}}\int_{0}^{\varphi_c}\cos( (\ell-k)F(\theta))e^{-\frac{(\ell+k)c}{(1-2c)^{3/2}}\theta^2}d\theta\label{comegacrewriteg5g3}\\&\qquad\qquad+R_{24}(\xi.\ell,k)\frac{\varphi_c}{|\ell+k|}\Big].\nonumber
\end{align} 
\begin{proof}
We parametrise $\Gamma_{\omega_c}\cap \HH^+$ by $w(\theta)=e^{i\theta}$ for $\theta\in [\theta_c,\pi/2]$, use \eqref{integralsymmetries} and the fact that $\ell+k$ is even to write 
\begin{align}
\tilde{C}_{\omega_c}(k,\ell)&=\frac{i^{-k-\ell}}{2(1+a^2)2\pi i}(1+(-1)^{\ell+k})\Big[\int_{\Gamma_{\omega_c}\cap\HH^+}\frac{dw}{w}\frac{G(w)^\ell G(1/w)^k}{\sqrt{w^2+2c}\sqrt{1/w^2+2c}}\nonumber\\&\qquad\qquad-(-1)^{\ell+k}\overline{\int_{\Gamma_{\omega_c}\cap\HH^+}\frac{dw}{w}\frac{G(w)^\ell G(1/w)^k}{\sqrt{w^2+2c}\sqrt{1/w^2+2c}}}\Big]\nonumber\\
&=\frac{i^{-k-\ell}}{(1+a^2)2\pi i}2ie^{\ell\mcR{[\psi(\theta_c)]}}\mcR\Big[\int_{\theta_c}^{\pi/2}\frac{e^{\ell\psi(\theta)-\ell\mcR[\psi(\theta_c)]}}{|e^{2i\theta}+2c|}d\theta\Big].
\label{temp234323}
\end{align}
Now, \begin{align}
\frac{d}{d\theta}|e^{2i\theta}+2c|=\frac{2ic(e^{-2i\theta}-e^{2i\theta})}{|e^{2i\theta}+2c|}
\label{temomo}
\end{align}  which is zero when $\theta=\pi/2$, so there is a bounded function $R_{25}(\theta)$ such that 
\begin{align}
\frac{1}{|e^{2i\theta}+2c|}=\frac{1}{(1-2c)}+R_{25}(\theta)(\theta-\pi/2)^2.
\label{taylorinnd3f}
\end{align} Hence the integral in \eqref{temp234323} can be approximated with an error bound given by
\begin{align}
&\abs{\int_{\theta_c}^{\pi/2}\frac{e^{\ell\psi(\theta)-\ell\mcR[\psi(\theta_c)]}}{|e^{2i\theta}+2c|}d\theta-\int_{\theta_c}^{\pi/2}\frac{e^{\ell\psi(\theta)-\ell\mcR[\psi(\theta_c)]}}{1-2c}d\theta}
\leq C\int_{\theta_c}^{\pi/2}(\pi/2-\theta)^2e^{\ell\mcR[\psi(\theta)-\psi(\theta_c)]}d\theta\label{boundtempbound345f}
\end{align} for some $C>0$.
We will return to this bound. Also
note that the assumptions $\ell+k<0$ and $\alpha\in (-1,1]$ give $\psi''(\pi/2)>0$ and $\ell<0$. We now focus on approximating the real part of the second integral appearing in the difference in \eqref{boundtempbound345f}
\begin{align}
\mcR\Big (\int_{\theta_c}^{\pi/2}e^{\ell\psi(\theta)-\ell\mcR[\psi(\theta_c)]}\Big )=\int_{\theta_c}^{\pi/2}\cos(\ell\mcI[\psi(\theta)])e^{\ell\mcR[\psi(\theta)-\psi(\theta_c)]}d\theta.
\end{align}
Recalling \eqref{realtaylor5x1q}, we see that
\begin{align}
&-\psi''(\pi/2)\varphi_c(\theta-\theta_c)+\psi''(\pi/2)(\theta_c-\theta)^2/2\nonumber\\
&-|\mcR[R_{13}(\theta_c,\alpha)\varphi_c^3(\theta_c-\theta)+R_{14}(\theta_c,\alpha)\varphi_c^2(\theta_c-\theta)^2/2\nonumber\\&\qquad\qquad R_{15}(\theta_c,\alpha)\varphi_c(\theta_c-\theta)^3+R_{12}(\theta_c,\alpha)(\theta_c-\theta)^4]|\nonumber\\
&>-\frac{1}{4}\psi''(\pi/2)\varphi_c(\theta-\theta_c) \label{tempbound1x2}
\end{align}
uniformly for $\theta_c\leq \theta\leq\pi/2$. Note the inequality in \eqref{tempbound1x2} reverses upon multiplying both sides by $\ell<0$.
Now we use \eqref{realtaylor5x1q}, \eqref{tempbound1x2} and the bound $|e^t-1|\leq |t|e^{|t|}$ to obtain a constant $C>0$ such that 
\begin{align}
\label{tempbound1x3}&\abs{\int_{\theta_c}^{\pi/2}\cos(\ell\mcI[\psi(\theta)])e^{\ell\mcR[\psi(\theta)-\psi(\theta_c)]}d\theta-\int_{\theta_c}^{\pi/2}\cos(\ell\mcI[\psi(\theta)])e^{-\ell\psi''(\pi/2)(\theta-\theta_c)[\varphi_c-\frac{\theta-\theta_c}{2}]}d\theta}\\
&\leq C|\ell|\int_{0}^{\varphi_c}e^{-\ell\psi''(\pi/2)\varphi_c\theta/4}\Big(\varphi_c^3\theta+\varphi_c^2\theta^2+\varphi_c\theta^3+\theta^4\Big )d\theta\nonumber\\
&\leq 4C|\ell|\varphi_c^3\int_{0}^{\varphi_c}e^{-\ell\psi''(\pi/2)\varphi_c\theta/4} \ \theta d\theta= 4C\frac{\varphi_c}{|\ell| \ (\psi''(\pi/2))^2}\int_0^{\ell\psi''(\pi/2)\varphi_c^2}e^{-\theta/4} \ \theta d\theta.\nonumber
\end{align}
We return to the bound \eqref{boundtempbound345f}. By a similar argument to \eqref{tempbound1x3}
\begin{align}
&Ce^{\ell\mcR[\psi(\theta_c)]}\abs{\int_{\theta_c}^{\pi/2}(\pi/2-\theta)^2e^{\ell\mcR{\psi(\theta)-\psi(\theta_c)}}d\theta-\int_{\theta_c}^{\pi/2}(\pi/2-\theta)^2e^{-\ell\psi''(\pi/2)(\theta-\theta_c)[\varphi_c-(\theta-\theta_c)/2]}d\theta}\label{tempbound1x4}\\
&\leq Ce^{\ell\mcR[\psi(\theta_c)]}\int_{0}^{\varphi_c}(\varphi_c-\theta)^2e^{-\ell\psi''(\pi/2)\varphi_c\theta}\Big(\varphi_c^3\theta+\varphi_c^2\theta^2+\varphi_c\theta^3+\theta^4\Big )d\theta\nonumber\\
&\leq \frac{4Ce^{\ell\mcR[\psi(\theta_c)]}\varphi_c^3}{|\ell\psi''(\pi/2)|^2}\int_0^{\ell\psi''(\pi/2)\varphi_c^2}e^{-\theta/4} \ \theta d\theta.\label{templgrp4}
\end{align}
The second term in the difference in \eqref{tempbound1x4} is bounded by
\begin{align}
&Ce^{\ell\mcR[\psi(\theta_c)]}\int_{\theta_c}^{\pi/2}(\pi/2-\theta)^2e^{-\ell\psi''(\pi/2)(\theta-\theta_c)[\varphi_c-(\theta-\theta_c)/2]}d\theta\nonumber\\
&\leq Ce^{\ell\mcR[\psi(\theta_c)]}\int_{0}^{\varphi_c}\varphi_c^2e^{-\ell\psi''(\pi/2)\theta\varphi_c/2}d\theta\leq C_1e^{\ell\mcR[\psi(\theta_c)]}\frac{\varphi_c}{|\ell\psi''(\pi/2)|}.\label{temprhssd}
\end{align}
for some $C_1>0$. Hence by the triangle inequality \eqref{boundtempbound345f} is bounded above by the sum of \eqref{templgrp4} and the right hand side of \eqref{temprhssd}.
Next, because of \eqref{psisym}, the fact that $\mcI[\psi(\pi/2+\theta)]$ is odd on $[-\pi/2,\pi/2]$ and that $\ell \pi(1-\alpha)=\pi(\ell-k)$ is an even multiple of $\pi$, we have
\begin{align}
&\int_{\theta_c}^{\pi/2}\cos(\ell\mcI[\psi(\theta)])e^{-\ell\psi''(\pi/2)(\theta-\theta_c)[\varphi_c-\frac{\theta-\theta_c}{2}]}\nonumber\\
&=\int_0^{\varphi_c}\cos(\ell\mcI[\psi(\theta+\theta_c)])e^{\ell\psi''(\pi/2)[(\theta-\varphi_c)^2-\varphi_c^2]/2}d\theta\nonumber\\
&=e^{-\ell\psi''(\pi/2)\varphi_c^2/2}\int_{-\varphi_c}^0\cos(\ell\mcI[\psi(\pi/2+\theta)])e^{\ell\psi''(\pi/2)\theta^2/2}d\theta\nonumber\\
&=e^{-\ell\psi''(\pi/2)\varphi_c^2/2}\int_{0}^{\varphi_c}\cos(\ell\mcI[\psi(\pi/2-\theta)])e^{\ell\psi''(\pi/2)\theta^2/2}.d\theta\label{tempfinal1zz2}
\end{align}
The integral in \eqref{comegacrewriteg5g3} follows directly from \eqref{tempfinal1zz2} since \eqref{tempfinal1zz2} multiplied by $1/(1-2c)$ gives the main contribution to the integral in \eqref{temp234323}.
\end{proof}
\end{proposition}

We now note that by equation 4.21 in \cite{C/J}, the formula \eqref{EKL1} in fact holds for $k>0$ \textit{or} $\ell>0$ instead of just $k,\ell>0$ as we stated it. Hence for non-zero integers $k,\ell$ such that $k>0$ or $\ell>0$, define a function $D_{\omega_c}(k,\ell)= E_{k,\ell}-\tilde{C}_{\omega_c}(k,\ell)$ where
\begin{align}
D_{\omega_c}(k,\ell)=\frac{i^{-k-\ell}}{2(1+a^2)2\pi i}\int_{\tilde{\Gamma}_{\omega_c}}\frac{dw}{w}\frac{G(w)^\ell G(1/w)^k}{\sqrt{w^2+2c}\sqrt{1/w^2+2c}}
\label{Domegac}
\end{align}
and where $\tilde{\Gamma}_{\omega_c}=\Gamma_1\setminus \Gamma_{\omega_c}$  has positive orientation counterclockwise around the origin. Note that $D_{\omega_c}(k,\ell)=D_{\omega_c}(\ell,k)$ when $k,\ell>0$.

One can use the symmetries in \eqref{squarerootsymmetries} to get
\begin{align}
&\int_{\Gamma_{a,b}}\frac{dw}{w}\frac{G(w)^\ell G(1/w)^k}{\sqrt{w^2+2c}\sqrt{1/w^2+2c}}\label{integralsymmetries}=-\overline{\int_{\overline{\Gamma_{a,b}}}\frac{dw}{w}\frac{G(w)^\ell G(1/w)^k}{\sqrt{w^2+2c}\sqrt{1/w^2+2c}}}\\
&=(-1)^{\ell+k}\int_{-\Gamma_{a,b}}\frac{dw}{w}\frac{G(w)^\ell G(1/w)^k}{\sqrt{w^2+2c}\sqrt{1/w^2+2c}}=-(-1)^{\ell+k}\overline{\int_{-\overline{\Gamma_{a,b}}}\frac{dw}{w}\frac{G(w)^\ell G(1/w)^k}{\sqrt{w^2+2c}\sqrt{1/w^2+2c}}}\nonumber
\end{align}
where $\Gamma_{a,b}=\{e^{i\theta}:\theta\in[a,b]\}$, $0\leq a<b\leq \pi/2\}$ and each curve has orientation counterclockwise around the origin.

\begin{proposition}
 Let non-zero integers $\ell,k$ be such that $\ell+k> 0$ is even and $\alpha = k/\ell$ lies in a compact subset of $(-1,1]$. There exists  a bounded function $R_{26}(\xi,\ell,k)$  such that 
\begin{align}
D_{\omega_c}(k,\ell)&=\frac{i^{2k}e^{\ell\mcR[\psi(\theta_c)]}}{\pi(1+a^2)(1-2c)}\Big[e^{\frac{(\ell+k)c}{(1-2c)^{3/2}}\varphi_c^2}\int_{\varphi_c}^\infty \cos((\ell-k)F(\theta))e^{-\frac{(\ell+k)c}{(1-2c)^{3/2}}\theta^2}d\theta\label{Domcpropexp1x}\\
&+R_{26}(\xi,\ell,k)\Big (\frac{\varphi_c}{(\ell+k)(1+\alpha)}+\frac{1}{(\ell+k)^{3/2}}\Big)\Big]\nonumber
\end{align}
\begin{proof}
From the definition of $D_{\omega_c}(k,\ell)$ we see that
\begin{align}
D_{\omega_c}(k,\ell)&=\frac{i^{-k-\ell}}{2(1+a^2)2\pi i}(1+(-1)^{\ell+k})\Big[\int_{\tilde{\Gamma}_{\omega_c}\cap\HH^+}\frac{dw}{w}\frac{G(w)^\ell G(1/w)^k}{\sqrt{w^2+2c}\sqrt{1/w^2+2c}}\nonumber\\&\qquad\qquad-(-1)^{\ell+k}\overline{\int_{\tilde{\Gamma}_{\omega_c}\cap\HH^+}\frac{dw}{w}\frac{G(w)^\ell G(1/w)^k}{\sqrt{w^2+2c}\sqrt{1/w^2+2c}}}\Big]\nonumber\\
&=\frac{i^{-k-\ell}}{(1+a^2)2\pi i}2ie^{\ell\mcR[\psi(\theta_c)]}\mcR\Big[\int_{\theta_c-\delta}^{\theta_c} \frac{e^{\ell\psi(\theta)-\ell\mcR[\psi(\theta_c)]}}{|e^{2i\theta}+2c|}d\theta+\int_{0}^{\theta_c-\delta} \frac{e^{\ell\psi(\theta)-\ell\mcR[\psi(\theta_c)]}}{|e^{2i\theta}+2c|}d\theta\Big]
\label{Domegacrewritetempze1}
\end{align}
where we parametrised $\tilde{\Gamma}_{\omega_c}\cap \HH^+$ by $w(\theta)=e^{i\theta}$ for $\theta\in[0,\theta_c]$ and $\delta>0$ will be chosen small enough. The main contribution to $D_{\omega_c}$ comes from  the first integral in \eqref{Domegacrewritetempze1}.\\
 From \eqref{taylorinnd3f} we have a constant $C>0$ such that
\begin{align}
\abs{\int_{\theta_c-\delta}^{\theta_c} \frac{e^{\ell\psi(\theta)-\mcR[\psi(\theta_c)]}}{|e^{2i\theta}+2c|}d\theta-\int_{\theta_c-\delta}^{\theta_c} \frac{e^{\ell\psi(\theta)-\ell\mcR[\psi(\theta_c)]}}{1-2c}d\theta}\leq C\int_{\theta_c-\delta}^{\theta_c}(\theta-\pi/2)^2e^{\ell\mcR[\psi(\theta)-\psi(\theta_c)]}d\theta. \label{tempboundv4d}
\end{align}
We will return to this bound. First we will consider
\begin{align}
\mcR\Big [\int_{\theta_c-\delta}^{\theta_c}e^{\ell\psi(\theta)-\ell\mcR[\psi(\theta_c)]}d\theta\Big ]=\int_{\theta_c-\delta}^{\theta_c}\cos(\ell\mcI[\psi(\theta)])e^{\ell\mcR[\psi(\theta)-\psi(\theta_c)]}d\theta
\end{align}
since this contributes to the leading term of $D_{\omega_c}$.

Recall \eqref{realtaylor5x1q} and take $\delta$ so small that 
\begin{align}
&\label{tempboundvxa4}\psi''(\pi/2)\varphi_c(\theta_c-\theta)+\psi''(\pi/2)(\theta_c-\theta)^2/2\\
&+|\mcR[R_{13}(\theta_c,\alpha)\varphi_c^3(\theta_c-\theta)+R_{14}(\theta_c,\alpha)\varphi_c^2(\theta_c-\theta)^2/2\nonumber\\&\qquad\qquad R_{15}(\theta_c,\alpha)\varphi_c(\theta_c-\theta)^3+R_{12}(\theta_c,\alpha)(\theta_c-\theta)^4]|\nonumber\\
&<\frac{1}{2}\big [\psi''(\pi/2)\varphi_c(\theta_c-\theta)+\psi''(\pi/2)(\theta_c-\theta)^2/2\big ] \nonumber
\end{align}
for all $\theta\in [\theta_c-\delta,\theta_c]$.
Now we use \eqref{realtaylor5x1q}, the bound $|e^t-1|\leq |t|e^{|t|}$ and then \eqref{tempboundvxa4} to give a constant $C$ such that
\begin{align}
&\abs{\int_{\theta_c-\delta}^{\theta_c}\cos(\ell\mcI[\psi(\theta)])e^{\ell\mcR[\psi(\theta)-\psi(\theta_c)]}d\theta-\int_{\theta_c-\delta}^{\theta_c}\cos(\ell\mcI[\psi(\theta)])e^{\ell\psi''(\pi/2)[\varphi_c(\theta_c-\theta)+(\theta_c-\theta)^2/2]}d\theta}\nonumber\\
\leq &C\ell\int_0^\delta e^{\frac{\ell}{2}\big[\psi''(\pi/2)\varphi_c\theta+\psi''(\pi/2)\theta^2/2\big ]}\big [\varphi_c^3\theta+\varphi_c^2\theta^2+\varphi_c\theta^3+\theta^4\big ] d\theta.\label{temmp2312}
\end{align}
Rewrite the integral on the right hand side of \eqref{temmp2312} as a sum of four integrals. We bound these four integrals separately.  Here, and several times below, we will use that $\psi''(\pi/2)<0$ and $\ell>0$. The first bound is
\begin{align}
&\ell\int_0^\delta e^{\ell\big[\psi''(\pi/2)\varphi_c\theta+\psi''(\pi/2)\theta^2/2\big ]/2}\varphi_c^3\theta d\theta
\leq \ell\int_0^\delta e^{\ell\big [\psi''(\pi/2)\varphi_c\theta\big ]/2}\varphi_c^3\theta d\theta< \frac{\varphi_c}{\ell\psi''(\pi/2)^2}\int_0^\infty e^{-\theta/2}\theta d\theta,\label{tempbound345f}
\end{align}
the third is
\begin{align}
&\ell\int_0^\delta e^{\ell\big [\psi''(\pi/2)\varphi_c\theta+\psi''(\pi/2)\theta^2/2\big ]/2}\varphi_c\theta^3d\theta\leq \ell\int_0^\delta e^{\ell\psi''(\pi/2)\theta^2/4}\varphi_c\theta^3d\theta< \frac{\varphi_c}{\ell\psi''(\pi/2)^2}\int_0^{\infty}e^{-\theta^2/4}\theta^3d\theta,\label{tempbound312d}
\end{align}
and the fourth is
\begin{align}
&\ell\int_0^\delta e^{\frac{\ell}{2}\big [\psi''(\pi/2)\varphi_c\theta+\psi''(\pi/2)\theta^2/2\big ]}\theta^4d\theta\leq \ell\int_0^\delta e^{\ell\psi''(\pi/2)\theta^2/4}\theta^4d\theta< \frac{1}{(-\ell\psi''(\pi/2))^{3/2}}\int_0^{\infty}e^{-\theta^2/4}\theta^4d\theta.\label{tempbound4edv}
\end{align}
We bound the second integral slightly differently by making the substitution $\theta\rarrow (\ell^2\psi''(\pi/2)^2\varphi_c)^{-1/3}\theta$ in
\begin{align}
&\ell\int_0^\delta e^{\frac{\ell}{2}\big [\psi''(\pi/2)\varphi_c\theta+\psi''(\pi/2)\theta^2/2\big ]}\varphi_c^2\theta^2d\theta\nonumber\\
&=\frac{\varphi_c}{\ell\psi''(\pi/2)^2}\int_0^{\delta (\ell^2\psi''(\pi/2)^2\varphi_c)^{1/3}}e^{\frac{-1}{2}\big [(\ell|\psi''(\pi/2)|\varphi_c^2)^{1/3}\theta+(\ell|\psi''(\pi/2)|\varphi_c^2)^{-1/3}\theta^2/2\big ]}\theta^2d\theta\nonumber\\
&< \frac{\varphi_c}{\ell\psi''(\pi/2)^2}\int_0^{\infty}e^{\frac{-1}{2}\big [(\ell|\psi''(\pi/2)|\varphi_c^2)^{1/3}\theta+(\ell|\psi''(\pi/2)|\varphi_c^2)^{-1/3}\theta^2/2\big ]}\theta^2d\theta.
\label{tempbound5f3xx}
\end{align}
The integral in \eqref{tempbound5f3xx} is bounded because $\int_0^{\infty}e^{-t\theta/2-\theta^2/(4t)}\theta^2d\theta$ is bounded uniformly for all $t> 0$. We now return to the bound \eqref{tempboundv4d}. Similar to how we got the bound \eqref{temmp2312}, we obtain
\begin{align}
&\abs{\int_{\theta_c-\delta}^{\theta_c}(\theta-\pi/2)^2e^{\ell\mcR[\psi(\theta)-\psi(\theta_c)]}d\theta-\int_{\theta_c-\delta}^{\theta_c}(\theta-\pi/2)^2e^{\ell\psi''(\pi/2)[\varphi_c(\theta_c-\theta)+(\theta-\theta_c)^2/2]}d\theta}\label{temp123cag}\\
&\leq C\ell\int_{\theta_c-\delta}^{\theta_c}(\theta-\pi/2)^2e^{\ell\psi''(\pi/2)[\varphi_c(\theta_c-\theta)+(\theta-\theta_c)^2/2]}[\varphi_c^3(\theta_c-\theta)\nonumber\\&\qquad\qquad+\varphi_c^2(\theta_c-\theta)^2+\varphi_c(\theta_c-\theta)^3+(\theta_c-\theta)^4]d\theta\nonumber\\
\leq &C(2\delta^2+2\varphi_c^2)\Big(\ell\int_{\theta_c-\delta}^{\theta_c}e^{\ell\psi''(\pi/2)[\varphi_c(\theta_c-\theta)+(\theta-\theta_c)^2/2]}[\varphi_c^3(\theta_c-\theta)\label{tempbound6scz}\\&\qquad\qquad+\varphi_c^2(\theta_c-\theta)^2+\varphi_c(\theta_c-\theta)^3+(\theta_c-\theta)^4]d\theta\Big) \nonumber
\end{align}
The term appearing in the big brackets in \eqref{tempbound6scz} was bounded previously via \eqref{tempbound345f}, \eqref{tempbound4edv}, \eqref{tempbound312d} and \eqref{tempbound5f3xx}.
The second integral in the difference in \eqref{temp123cag} is 
\begin{align}
&\int_{\theta_c-\delta}^{\theta_c}(\theta-\pi/2)^2e^{\ell\psi''(\pi/2)[\varphi_c+(\theta-\theta_c)^2/2]}d\theta\nonumber\\
&\leq \int_0^{\delta}(2\theta^2+2\varphi_c^2)e^{\ell\psi''(\pi/2)[\varphi_c\theta+\theta^2/2]}d\theta\nonumber\\
&\leq \int_0^\delta 2\theta^2e^{\ell\psi''(\pi/2)\theta^2/2}d\theta+\int_0^\delta \varphi_c^2e^{\ell\psi''(\pi/2)\varphi_c\theta}d\theta\nonumber\\
&<\frac{2}{(-\ell\psi''(\pi/2))^{3/2}}\int_0^{\infty}\theta^2e^{-\theta^2/2}d\theta+\frac{2\varphi_c}{-\ell\psi''(\pi/2)}\int_0^{\infty}e^{-\theta}d\theta.
\end{align} Hence we have established an upper bound on the right hand side of \eqref{tempboundv4d}.
 The second integral appearing in the difference in \eqref{temmp2312} contributes to the main term and using similar manipulations leading to \eqref{tempfinal1zz2} we can rewrite it is as
\begin{align}
&\int_{\theta_c-\delta}^{\theta_c}\cos(\ell\mcI[\psi(\theta)])e^{\ell\psi''(\pi/2)[\varphi_c(\theta_c-\theta)+(\theta_c-\theta)^2/2]}d\theta\\
&=e^{-\ell\psi''(\pi/2)\varphi_c^2/2}\int_{\varphi_c}^{\delta+\varphi_c}\cos(\ell\mcI[\psi(\pi/2-\theta)])e^{\ell\psi''(\pi/2)\theta^2/2}d\theta.\nonumber
\end{align}
 We extend the integration in the last integral to infinity which gives an error term
\begin{align}
&\abs{\int_{\varphi_c}^{\delta+\varphi_c}\cos(\ell\mcI[\psi(\pi/2-\theta)])e^{\ell\psi''(\pi/2)\theta^2/2}d\theta-\int_{\varphi_c}^{\infty}\cos(\ell\mcI[\psi(\pi/2-\theta)])e^{\ell\psi''(\pi/2)\theta^2/2}d\theta}\label{tempbound7sz}\\
&\leq \int_{\delta+\varphi_c}^{\infty}e^{\ell\psi''(\pi/2)\theta^2/2}d\theta\leq 2\frac{e^{\ell\psi''(\pi/2)(\delta+\varphi_c)^2/2}}{-\ell\psi''(\pi/2)(\delta+\varphi_c)}.\nonumber
\end{align}
Thus the infinite integral
\begin{align}
&e^{-\ell\psi''(\pi/2)\varphi_c^2/2}\int_{\varphi_c}^{\infty}\cos(\ell\mcI[\psi(\pi/2-\theta)])e^{\ell\psi''(\pi/2)\theta^2/2}d\theta\nonumber
\end{align}
 gives the main term in \eqref{Domcpropexp1x}.
 
It remains to bound the second integral in \eqref{Domegacrewritetempze1}. Since $\mcR[\psi(\theta)]$ increases on $(0,\pi/2)$, there is a constant $C>0$ such that
\begin{align}
\abs{\int_0^{\theta_c-\delta}d\theta\frac{e^{\ell\psi(\theta)-\ell\mcR[\psi(\theta_c)]}}{|e^{2i\theta}+2c|}}
\leq Ce^{\ell\mcR[\psi(\theta_c-\delta)-\psi(\theta_c)]}.
\label{tempbound8sxzs}
\end{align}From \eqref{realtaylor5x1q} we can take $\delta$ so small that 
\begin{align}
\mcR[\psi(\theta_c-\delta)-\psi(\theta_c)]&<\frac{\psi''(\pi/2)}{2}[\delta\varphi_c+\delta^2/2]\leq \frac{\psi''(\pi/2)\delta^2}{4}.\nonumber
\end{align}
 Hence \eqref{tempbound8sxzs} is bounded above by 
\begin{align}
Ce^{-\ell|\psi''(\pi/2)|\delta^2/4}\nonumber.
\end{align}
\end{proof}
\label{PropDomegacexp3x}
\end{proposition}

We now have all of the ingredients to prove Proposition \ref{propeklcomegac}.
\begin{proof}[Proof of Proposition \ref{propeklcomegac}]
We divide the proof into five cases, one when $\tilde{\alpha}=-1$, the other four when $\tilde{\alpha}\in (-1,1)$ or $\tilde{\alpha}=1$ and $\ell+k$ is positive or negative. Each case consists of defining $R_9$ by the rearrangement of \eqref{tempd2wd} and using the stated substitutions and bounds.

First the case $\tilde{\alpha}=-1$ for which the assumptions imply $\ell+k\in\{-2,0,2\}$.
By \eqref{tempz3e4} we have
\begin{align}
\abs{\frac{\sin((\ell-k)F(\varphi_c))}{\ell-k}}\geq C(\eps)\min (\varphi_c,1/|\ell-k|).\label{lowerboundsinF2s}
\end{align}
By lemma \ref{galphaalpha} and proposition \ref{propEKL} and there are $C_1,C_2>0$ such that 
\begin{align}
E_{|k|,|\ell|}\leq C_1e^{-C_2|\ell|}.\label{etempws}
\end{align}
We define $R_9$ as the rearrangement of \eqref{tempd2wd}
\begin{align}
R_9&:= \frac{E_{k,\ell}-C_{\omega_c}(k,\ell)+(-1)^k |G(\omega_c)|^{\ell+k}\sin((\ell-k)F(\varphi_c))(\pi(1-a)^2(\ell-k)/\sqrt{1-2c})^{-1}}{-(-1)^k |G(\omega_c)|^{\ell+k}\sin((\ell-k)F(\varphi_c))(\pi(1-a)^2(\ell-k)/\sqrt{1-2c})^{-1}}\label{R54tmp3rf}
\end{align}
and use the formula for $C_{\omega_c}$ given by Proposition \ref{Comegasin} so that $R_9$ is equal to
\begin{align}
-\frac{E_{|k|,|\ell|}-(-1)^k |G(\omega_c)|^{\ell+k}R_{29}(\xi,\ell,k)\varphi_c(\pi(1-a)^2(\ell-k))^{-1}}{(-1)^k |G(\omega_c)|^{\ell+k}\sin((\ell-k)F(\varphi_c))(\pi(1-a)^2(\ell-k)/\sqrt{1-2c})^{-1}}.
\end{align}
Now we use the upper bounds \eqref{lowerboundsinF2s} and \eqref{etempws}, and the conditions $\varphi_c^{2-\gamma}(|\ell|+|k|)\geq 1, \ell+k\in\{-2,0,2\}$ to get
\begin{align}
|R_{9}|&\leq C_3\frac{e^{-C_2|\ell|}}{|G(\omega_c)|^{\ell+k}\min(\varphi_c,1/|\ell-k|)}+C_4\frac{\varphi_c}{|\ell-k|\min(\varphi_c,1/|\ell-k|)}\label{temp2eda}\\
&\leq C_5\max(\frac{1}{|\ell|},\varphi_c).\nonumber
\end{align}
For the case $\tilde{\alpha}\in (-1,1)$ and $\ell+k<-2$ we can set the two different expressions for $\tilde{C}_{\omega_c}$, \eqref{Cwcalphaminus11/2} and \eqref{comegacrewriteg5g3}, equal to one another so that bound \eqref{lowerboundsinF2s} carries over to
\begin{align}
\abs{\frac{|G(\omega_c)|^{\ell+k}e^{(\ell+k)c'\varphi_c^2}}{\pi(1-a)^2}\int_0^{\varphi_c}\cos((\ell-k)F(\theta))e^{(\ell+k)c'\theta^2}d\theta}&\geq C_1(\eps)|G(\omega_c)|^{\ell+k}\min(\frac{\varphi_c}{1-2c},\frac{1}{|\ell-k|})\label{tempboundges4}
\end{align}
for $|\ell|$ large enough. We also have that  $\ell\mcR[\psi(\theta_c)]= (\ell+k)\log|G(e^{i\theta_c})|\geq(\ell+k)\log|G(i)|> 0$ which implies $|G(\omega_c)|^{\ell+k}$ grows exponentially in $|\ell|$. 
Just as in \eqref{R54tmp3rf} we define $R_9$ to be the rearrangement of \eqref{tempd2wd} but now use the formula for $C_{\omega_c}$ provided by Proposition \ref{steepestdescComega} so that
\begin{align}
R_9=\frac{E_{|k|,|\ell|}-(-1)^k|G(\omega_c)|^{\ell+k}R_{19}(\xi,\ell,k)\varphi_c (|\ell+k|\pi(1-a)^2)^{-1}}{-(-1)^k|G(\omega_c)|^{\ell+k}(\pi(1-a)^2)^{-1}e^{-(\ell+k)c'\varphi_c^2}\int_0^{\varphi_c} \cos((\ell-k)F(\theta))e^{(\ell+k)c'\theta^2}d\theta}. \label{tempg3av}
\end{align}
Just as in \eqref{temp2eda}, we use \eqref{tempg3av} together with \eqref{etempws}, \eqref{tempboundges4} to get the bound 
\begin{align}
|R_9|&\leq C(\tilde{\alpha})\max(\frac{1}{|\ell|},\varphi_c).
\end{align}
For the case $\tilde{\alpha}=1$, $\ell+k<-2$ we again define $R_{9}$ as the rearrangement of \eqref{tempd2wd}. We then use Proposition \ref{Comegasin} followed by lemma \ref{templemmalpha13s} on $C_{\omega_c}$.  Then using the bounds \eqref{etempws}, \eqref{dawsonest1}, \eqref{dawsonest2} we obtain a $C>0$ such that
\begin{align}
|R_{9}|\leq C\begin{cases}
\frac{1}{|\ell+k|}+\frac{\varphi_c}{\sqrt{|\ell+k|}}, & \sqrt{|\ell+k|c'}\varphi_c\in(0,1]\\
\varphi_c^2, & \sqrt{|\ell+k|c'}\varphi_c\in[1,\infty).
\end{cases}
\end{align}
For the case $\tilde{\alpha}\in (-1,1)$, $\ell+k>2$ we again define $R_{9}$ as the rearrangement of \eqref{tempd2wd}. Then use Proposition \ref{PropDomegacexp3x} together with lemma \ref{templeextendsinexp1dx} on the function $E_{k,\ell}-C_{\omega_c}(k,\ell)$ appearing in the numerator of $R_{9}$. Then use the bound \eqref{lowerboundsinF2s} and $\varphi_c^{2-\gamma}(|\ell|+|k|)\geq 1$  to obtain
\begin{align}
|R_{9}|\leq C'(\tilde{\alpha})\max(\frac{1}{\sqrt{\ell}}+\varphi_c,\frac{\varphi_c^{1-\gamma}}{\sqrt{\ell}}+\frac{1}{\sqrt{\ell}}).
\end{align}
For the case $\tilde{\alpha}=1$, $\ell+k>2$ use lemmas \ref{templeextendsinexp1dx} and \ref{PropDomegacexp3x} and the bound \eqref{millest} to obtain
\begin{align}
|R_{9}|\leq C'(\varphi_c^2+\frac{\varphi_c}{\sqrt{\ell+k}}+\frac{1}{\ell+k}).
\end{align}
\end{proof}

\section{Uniform bound on $R_{\eps_1,\eps_2}$}\label{SecUniformbound}

We begin by proving Theorem \ref{ThmInvKast}. It is proved in \cite{C/J} that
for $n=4m$, $m\in\N_{>0}$, $(x_1,x_2)\in W_{\eps_1}$, $(y_1,y_2)\in B_{\eps_2}$ with $\eps_1,\eps_2\in \{0,1\}$ we have the formula
\begin{align}
K_{a,1}^{-1}((x_1,x_2),(y_1,y_2))=& \ \K_{1,1}^{-1}((x_1,x_2),(y_1,y_2))-B_{\eps_1,\eps_2}(a,(x_1,x_2),(y_1,y_2))\label{inverseKastelements}\\
&\qquad\qquad+B^*_{\eps_1,\eps_2}(a,(x_1,x_2),(y_1,y_2))\nonumber
\end{align}
where $B_{\eps_1,\eps_2}(a,x_1,x_2,y_1,y_2)$ is given by \eqref{Beps1eps2} below and
$B^*_{\eps_1,\eps_2}$ is the sum of three double contour integrals related to $B_{\eps_1,\eps_2}$ by symmetry. 

We recall $B_{\eps_1,\eps_2}$ as equation (3.11) in \cite{C/J},
\begin{align}
B_{\eps_1,\eps_2}(a,x_1,x_2,y_1,y_2)=\frac{i^{(x_1-x_2+y_1-y_2)/2}}{(2\pi i)^2}\int_{\Gamma_{r'}}\frac{dw_1}{w_1}\int_{\Gamma_{1/r'}}dw_2 \frac{V_{\eps_1,\eps_2}(w_1,w_2)}{w_2-w_1}\frac{H_{x_1+1,x_2}(w_1)}{H_{y_1,y_2+1}(w_2)}
\label{Beps1eps2}
\end{align}
where equation (2.11) in \cite{C/J} is
\begin{align}
H_{x,y}(w)&=\frac{w^{n/2}G(w)^{(n-x)/2}}{G(w^{-1})^{(n-y)/2}}
\label{Hxy}
\end{align}
for integers $0<x,y<n$ and $G$ is defined in \eqref{Gfunction}.
The expression for $V_{\eps_1,\eps_2}(w_1,w_2)$ is somewhat involved and we refer to equation (3.7) in \cite{C/J}. The precise form is of $V_{\eps_1,\eps_2}$ is not needed here. However we note that $V_{\eps_1,\eps_2}$ is analytic in $\big(\C\setminus (i(-\infty,-1/\sqrt{2c}]\cup i[-\sqrt{2c},\sqrt{2c}]\cup i[1/\sqrt{2c},\infty))\big )^2$. Thus the integrand of $B_{\eps_1,\eps_2}$ is analytic in the same set minus the collection of points $w_1=w_2$.  

Define the error term
\begin{align}
R_{\eps_1,\eps_2}(a,x_1,x_2,y_1,y_2)&=\frac{i^{(x_1-x_2+y_1-y_2)/2}}{(2\pi i)^2}\int_{\text{desc$_\xi$}}\frac{dw_1}{w_1}\int_{\text{asc$_{\xi}$}}dw_2 \frac{V_{\eps_1,\eps_2}(w_1,w_2)}{w_2-w_1}\frac{H_{x_1+1,x_2}(w_1)}{H_{y_1,y_2+1}(w_2)}
\label{Reps1eps2}
\end{align}
where desc$_\xi$ and asc$_\xi$ are the contours of steepest descent and ascent of the function $g_{\xi}$ passing through $\omega_c$, see \cite{C/J} or lemmas \ref{CJcurve1} and \ref{CJcurve2} below. We want to show that
\begin{align}
B_{\eps_1,\eps_2}(a,x_1,x_2,y_1,y_2)=C_{\omega_c}(x,y)+R_{\eps_1,\eps_2}(a,x_1,x_2,y_1,y_2).
\label{Beps1eps2Comegac}
\end{align}
The formula  \eqref{Beps1eps2Comegac} follows by a contour deformation in the contour integral formula for $B_{\eps_1,\eps_2}$ as given in \eqref{Beps1eps2}. For this section we use the coordinates in \eqref{changeofcoordinates}.
From equations \eqref{Hxy} and \eqref{gxi} we compute 
\begin{align}
\frac{H_{x_1+1,x_2}(w_1)}{H_{y_1,y_2+1}(w_2)}
=\Big (\frac{G(w_1^{-1})^{a_2}}{G(w_1)^{a_1}} \frac{G(w_2)^{b_1}}{G(w_2^{-1})^{b_2}}\Big )\exp{\{\frac{n}{2}(g_{\xi}(w_1)-g_\xi (w_2))\}}.
\label{Hxysub}
\end{align}
We have lemma 3.15 from \cite{C/J}  
\begin{lemma}
For $\xi=\xi_c$ there is a path of steepest descent for $g_\xi$ leaving $i$ at the angle $-\pi/6$ going to $0$ (via $\HH^+$) and a path of steepest ascent leaving at an angle $\pi/6$ and going to infinity (via $\HH^+$).\label{CJcurve1}
\end{lemma}
Note $\mathcal{R}[g_\xi(w)]$ is symmetric across the real and imaginary axes so this also gives the paths on the upper left quadrant and bottom half of the plane. We have lemma 3.20 from \cite{C/J} 
\begin{lemma}
For $\xi_c>\xi>-\frac{1}{2}\sqrt{1+2c}$, choose $\omega_c=e^{i\theta_c}$ with $\theta_c\in(0,\pi/2)$ such that $g'_\xi(\omega_c)=0$. There is a contour of steepest ascent leaving $\omega_c$ at an angle $\theta_c-\pi/4$ ending at infinity (via $\HH^+$) and a contour of steepest ascent leaving $\omega_c$ at an angle $\theta_c+3\pi/4$ ending at a cut (via $\HH^+$) and an ascent contour ending at $i\sqrt{2c}$ traveling via the cut $i[0,\sqrt{2c}]$.\\
 There is a contour of steepest descent leaving $\omega_c$ at an angle $\theta_c-3\pi/4$ ending at zero (via $\HH^+$) and a contour of steepest descent leaving $\omega_c$ at an angle $\theta_c+\pi/4$ ending at a cut (via $\HH^+$) and a descent contour ending at $i/\sqrt{2c}$ traveling via the cut $i[1/\sqrt{2c},\infty)$.\label{CJcurve2}
\end{lemma}

Now deform $\Gamma_{r'}$ to the path of steepest descent, $\text{desc}_{\xi}$, for $g_{\xi}$ passing through the critical points $\pm\omega_c$, $\pm\overline{\omega}_c$. This can be done due to the crude estimates for $\abs{w_2}=R$ large and $\abs{w_1}=1/R$ small,
\begin{align}
&\abs{\frac{1}{H_{y_1,y_2+1}(w_2)}}=\abs{\frac{G(w_1^{-1})^{a_2}}{G(w_1)^{a_1}}\exp(-ng_\xi(w_2)/2)}=\big(\frac{c}{2}\big)^{a_2}\frac{1}{R^{n(1+\xi)/2-a_2}}\abs{(1+O(1/R))},
\label{largeRgxi}\\
&\abs{H_{x_1+1,x_2}(w_1)}=\abs{\frac{G(w_2)^{b_1}}{G(w_2^{-1})^{b_2}}\exp(ng_{\xi}(w_1)/2)}=\big(\frac{c}{2}\big)^{b_1}\frac{1}{R^{n(1+\xi)/2-b_1}}\abs{(1+O(1/R))}.
\label{expanonH}
\end{align}
 which follow from \eqref{Gwlarge} and \eqref{Gwsmall}. Deform $\Gamma_{1/r'}$ to path of steepest ascent of $g_{\xi}$ passing through the points $\pm \omega_c$, $\pm \overline{\omega_c}$, label this contour asc$_\xi$. Since the $w_2$-contour passes over the $w_1$-contour we pick up a contribution 
 \begin{align}
\frac{i^{(x_2-x_1+y_1-y_2)/2}}{2\pi i}\int_{\Gamma_{\omega_c}}V_{\eps_1,\eps_2}(w,w)G(w)^{(y_1-x_1-1)/2}G(w^{-1})^{(x_2-y_2-1)/2}\frac{dw}{w}
\label{Comegacdef}
\end{align}
by the residue theorem.
Lemma (3.2) in \cite{C/J}  is
\begin{align}
V_{\eps_1,\eps_2}(w,w)=\frac{(-1)^{1+h(\eps_1,\eps_2)}a^{\eps_2}G(w^{-1})^{h(\eps_1,\eps_2)}+a^{1-\eps_2}G(w)G(w^{-1})^{1-h(\eps_1,\eps_2)}}{2(1+a^2)\sqrt{w^2+2c}\sqrt{1/w^2+2c}}.
\end{align} 
This together with the definitions in \eqref{klh1} and \eqref{klh2} give that the contribution \eqref{Comegacdef} is equal to
 $C_{\omega_c}(x,y)$ defined in \eqref{Comegacccc}. The rest of the integral is given by what we called the error term, \eqref{Reps1eps2}. This proves the formula \eqref{Beps1eps2Comegac} and hence Theorem \ref{ThmInvKast}.
\\

\begin{figure}[h]
\centering
\includegraphics[width = 0.6\textwidth]{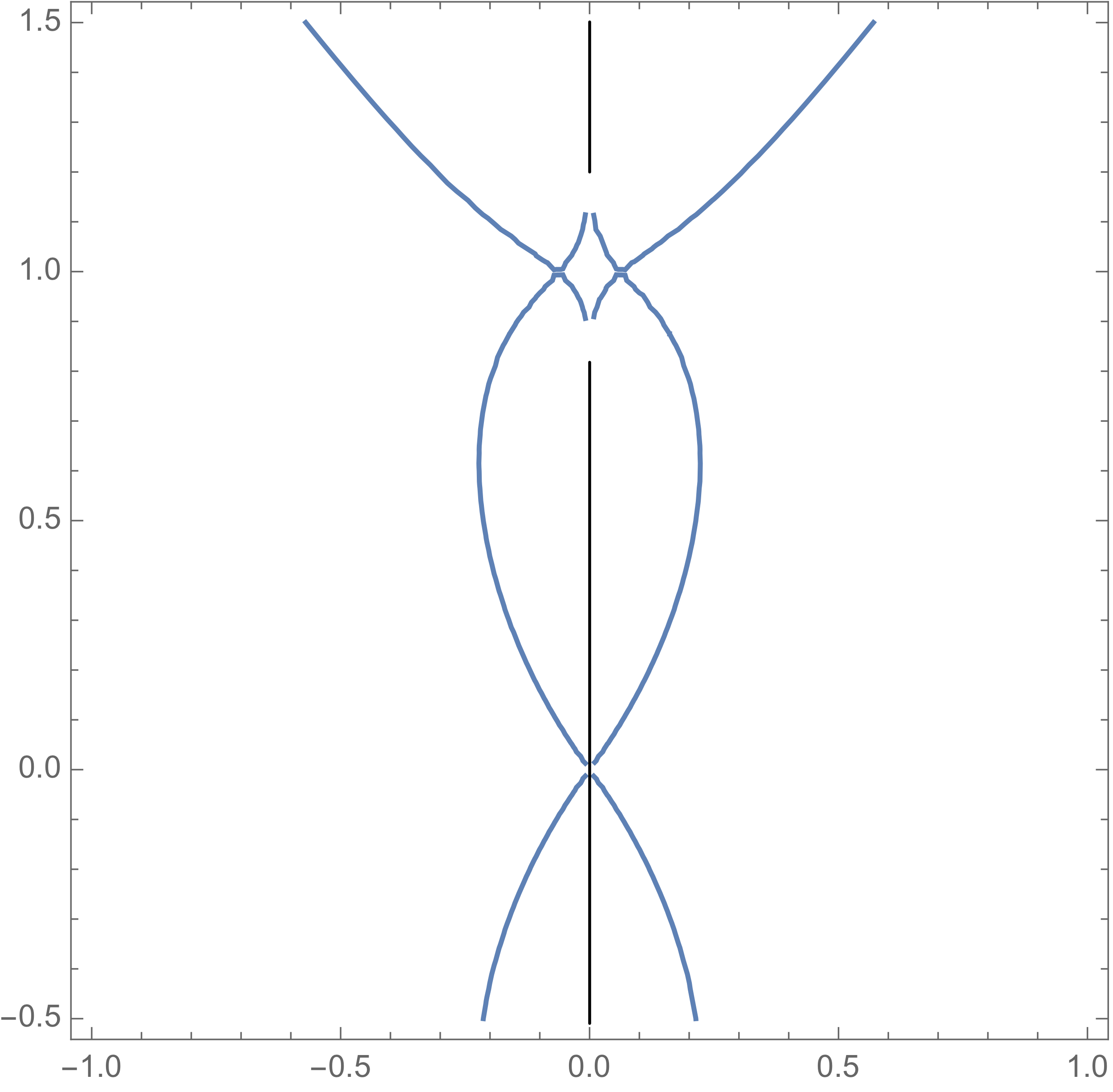}
\caption[Steepest ascent and descent contours]
	{A numerical plot of $\{w: \mathcal{I}[g_{\xi}(w)-g_{\xi}(\omega_c)]=0\}$, $\xi_c-\xi=0.01$ (blue) and the branch cut of of $g_\xi$ (black) where $a=1/2$.}
	\label{steepestdescentascentcontours}
\end{figure}

We now turn to the proof of Proposition \ref{uniformbound}. This involves an analysis of the steepest descent and ascent contours for the saddle point function $g_\xi$, across a region of the parameter space $\xi$ where two single critical points merge to form a double critical point.
In order to carry out this analysis we note that the ascent contour is infinite and the descent contour travels to the origin, where the integrand is not analytic. In light of this, we make the following deformations.
\\
Consider the upper right quadrant of $\C$. Take $R>2$ very large, by lemmas 3.15 and 3.20 in \cite{C/J}, the path of steepest ascent of $g_\xi$ in the upper right quadrant intersects the quarter circle centred at 0 that travels from $R$ to $iR$. Deform the section of asc$_\xi$ in the upper right quadrant that extends from the intersection to infinity to the section of the quarter circle starting at the intersection, travelling down the quarter circle and ending at the point $R\in \C$. Make the equivalent symmetric deformations in the other three quadrants and label this contour $\text{asc}^R_\xi$. Similarly, the path of steepest descent in the upper right quadrant intersects a small quarter circle centred at $0$ of radius $\eps$, take $\eps=1/R$ for simplicity. Deform the section of desc$_\xi$ starting at the intersection and going to zero to the contour starting from the intersection and travelling down along the small quarter circle to the point $1/R\in \C$. Make the equivalent symmetric deformations in the other quadrants and label this new contour desc$_{\xi,1/R}$. Now note from \eqref{expanonH} that we can take $R$ so large and fixed so that $\mcR[g_{\xi}(\omega_c)]> \mcR[g_\xi(w_1)]+1$ and $\mcR[-g_\xi(\omega_c)]>\mcR[-g_\xi(w_2)]+1$ for $|w_1|=1/R$, $|w_2|=R$.
\begin{figure}
\centering
\begin{tikzpicture}[scale=0.5]
\draw (-5,0) -- (5,0);
\draw (0,-5) -- (0,5);
\draw (0,0) .. controls (1.5,1.5) ..
 (0,2.2);
 \draw (0,0) .. controls (-1.5,1.5) .. (0,2.2);
 \draw (0,0) .. controls (-1.5,-1.5) .. (0,-2.2);
 \draw (0,0) .. controls (1.5,-1.5) .. (0,-2.2);
\draw (0,1.8) .. controls (1,2.5) and  (3,3) .. (4,4.5);
\draw (0,1.8) .. controls (-1,2.5) and (-3,3) .. (-4,4.5);
\draw (0,-1.8) .. controls (-1,-2.5) and (-3,-3) .. (-4,-4.5);
\draw (0,-1.8) .. controls (1,-2.5) and (3,-3) .. (4, -4.5);
\draw (1.1,1.5) -- (2,1.5) node[below]{desc$_{\xi}$};
\draw (2.3,3) -- (3.3,3) node[below]{asc$_{\xi}$};

\end{tikzpicture}
\qquad \qquad \qquad \qquad
\begin{tikzpicture}[scale=0.5]
\draw (-5.15,0) -- (5.15,0);
\draw (0,-5.15) -- (0,5.15);

\draw (0.14,0.14) .. controls (1.5,1.5) .. (0,2.2);
 \draw (0.2,0) arc (0:45:0.2);
 \draw (-0.14,0.14) .. controls (-1.5,1.5) .. (0,2.2);
 \draw (-0.14,0.14) arc (135:180:0.2);
 \draw (-0.14,-0.14) .. controls (-1.5,-1.5) .. (0,-2.2);
 \draw (-0.2,0) arc (180:225:0.2);
 \draw (0.14,-0.14) .. controls (1.5,-1.5) .. (0,-2.2);
 \draw (0.14,-0.14) arc (315:360:0.2);
 
\draw (0,1.8) .. controls (1,2.5) and  (3,3) .. (3.4,3.8);
\draw (5.1,0) arc (0:48:5.1);
\draw (0,1.8) .. controls (-1,2.5) and (-3,3) .. (-3.4,3.8);
\draw (-3.4,3.8) arc (132:180:5.1);
\draw (0,-1.8) .. controls (-1,-2.5) and (-3,-3) .. (-3.4,-3.8);
\draw (-5.1,0) arc (180:228:5.1);
\draw (0,-1.8) .. controls (1,-2.5) and (3,-3) .. (3.4, -3.8);
\draw (3.4,-3.8) arc (312:360:5.1);
\draw (1.1,1.5) -- (2.4,1.5) node[below]{$\widetilde{\text{desc}}_{\xi,1/R}$};
\draw (2.4,3) -- (3.3,3) node[below]{$\widetilde{\text{asc}}_{\xi}^R$};
\end{tikzpicture}
\caption{Global contour deformations} \label{fig:M1}
\end{figure}
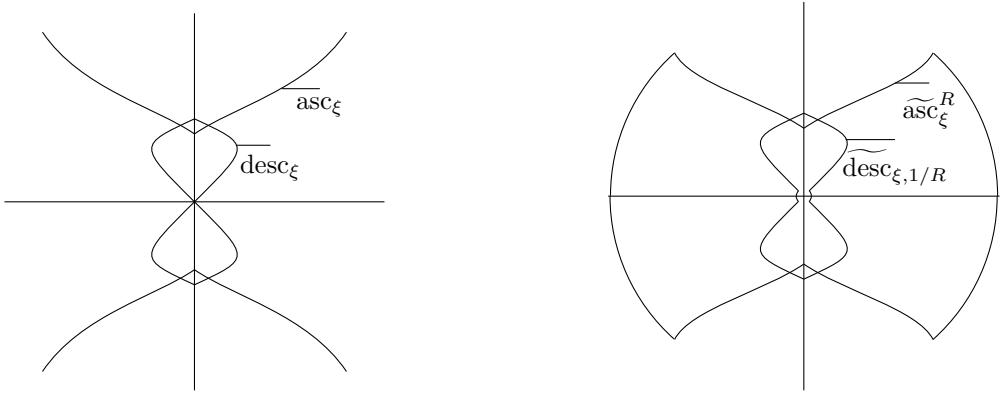

\begin{figure}
\centering
\begin{tikzpicture}[scale=2]
\draw (0,0) -- (0,0) node[left]{$i$};
\draw (0,-1.25) -- (0,1.25);
\draw (0.309,0) -- (1.162,0.84);
\draw (1.162,0.84) .. controls (1.3,1) .. (1.45, 1.2);
\draw (0,-1) -- (0.309,0);
\draw (0,1) -- (0.309,0);
\draw (0.309,0) -- (1.209,-0.70);
\draw (1.209,-0.7) .. controls (1.45,-0.9) .. (1.6,-1.1);
\draw (1.909,-0.65) -- (1.909,-0.65) node[below]{$\widetilde{\text{desc}}_{\xi,1/R}$};
\draw (1.5,0.99) -- (1.5,0.99) node[below]{$\widetilde{\text{asc}}_{\xi}^R$};
\end{tikzpicture}
\caption{Local contour deformations near $i$ in $\HH^+$.} \label{fig:M2}
\end{figure}
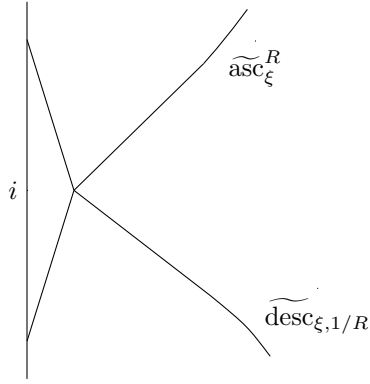
In order to deal with the contributions to the asymptotics coming from small neighbourhoods of the critical points we make the following local deformations. \\
Let $0\leq \xi_c-\xi\leq \delta^*_n$ where $\delta^*_n>0$ is a sequence going to zero. Consider the section of the contour desc$_{\xi,1/R}$ in $\HH^+$.  For $\delta_0>0$ sufficiently small and $n$ large enough, later we will prove that the section of desc$_{\xi,1/R}$ that leaves the small quarter circle centred at $0$ of radius $1/R$ and travelling up to $\omega_c'$ passes through the set $\{i+\delta_0e^{i\theta}:\theta\in[-\pi/4-a^*,-\pi/6+a^*]\}$ for some positive $a^*<\pi/24$. Label the contours intersection with this set as $w^*_\xi$. Next consider the connected component of desc$_{\xi,1/R}$ with endpoints $w^*_\xi$ and $\omega_c$, deform this component to a straight line connecting the two endpoints. Also, deform the component of desc$_{\xi,1/R}$ starting at $\omega_c$ and ending at the imaginary axis to the straight line starting from $\omega_c$ and leaving at an angle $3\pi/5$, extending to the imaginary axis. In the case $\omega_c=i$ this last section vanishes. Make the equivalent symmetric deformations in the other three quadrants of $\C$. Label this deformed contour $\widetilde{\text{desc}}_{\xi,1/R}$ and note it now depends on $\delta_0$ and $\xi$, it has also only been defined for $\delta_0$ sufficiently small and $n$ large enough. We call the equivalent deformations to $\text{asc}^R_\xi$ as $\widetilde{\text{asc}}_{\xi}^R$.
\\
We now have 
\begin{align}
R_{\eps_1,\eps_2}(a,x_1,x_2,y_1,y_2)&=\frac{i^{(x_1-x_2+y_1-y_2)/2}}{(2\pi i)^2}\int_{\widetilde{\text{desc}}_{\xi,1/R}}\frac{dw_1}{w_1}\int_{\widetilde{\text{asc}}^R_\xi}dw_2 \frac{V_{\eps_1,\eps_2}(w_1,w_2)}{w_2-w_1}\frac{H_{x_1+1,x_2}(w_1)}{H_{y_1,y_2+1}(w_2)}
\end{align}
So we have the bound
\begin{align}
\abs{R_{\eps_1,\eps_2}}&=\frac{1}{(2\pi)^2}\abs{\int_{\widetilde{\text{desc}}_{\xi,1/R}}\frac{dw_1}{w_1}\int_{\widetilde{\text{asc}}^R_\xi}dw_2 \Big (\frac{G(w_1^{-1})^{a_2}}{G(w_1)^{a_1}} \frac{G(w_2)^{b_1}}{G(w_2^{-1})^{b_2}}\Big ) \frac{V_{\eps_1,\eps_2}(w_1,w_2)}{w_2-w_1}\exp\{\frac{n}{2}(g_{\xi}(w_1)-g_\xi(w_2))\}}\\
&\leq\frac{M }{(2\pi)^2}\int_{\widetilde{\text{desc}}_{\xi,1/R}}\abs{dw_1}\int_{\widetilde{\text{asc}}^R_\xi}\abs{dw_2}|K(w_1,w_2)| \abs{\frac{\exp\{\frac{n}{2}(g_{\xi}(w_1)-g_\xi(w_2))\}}{w_2-w_1}}.
\label{startingpoint}
\end{align}
where
\begin{align}
M:=\sup_1\{\abs{V_{\eps_1,\eps_2}(w_1,w_2)}\}, &&K(w_1,w_2):=\frac{G(w_1^{-1})^{a_2}}{G(w_1)^{a_1}} \frac{G(w_2)^{b_1}}{G(w_2^{-1})^{b_2}}
\end{align}
and $\sup_1$ is the supremum over a compact subset $A$ of $( \overline{\B}(0,R)\setminus \B(0,1/R))^2$ such that $V_{\eps_1,\eps_2}$ is analytic on $A$ and $\widetilde{\text{desc}}_{\xi,1/R},\widetilde{\text{asc}}^R_\xi\subset A$ for every $\xi$. Note that both $\widetilde{\text{desc}}_{\xi,1/R},\widetilde{\text{asc}}^R_\xi\subset  \overline{\B}(0,R)\setminus \B(0,1/R)$ for all $0\leq \xi_c-\xi\leq \delta^*_n$.
Now we are in a position to prove our bound on $R_{\eps_1,\eps_2}$.

\begin{proof}[Proof of Proposition \ref{uniformbound}]
We first formulate a few key lemmas. The first of these, provides us with the fact that the contours $\widetilde{\text{desc}}_{\xi,1/R}$, $\widetilde{\text{asc}}_\xi^R$ are well-defined.
\begin{lemma}
Let $\delta^*_n>0$ be a sequence that goes to zero as $n\rarrow \infty$, for $\delta_0>0$ sufficiently small there is a natural number $N$ such that for all $n>N$ the section of desc$_{\xi}$ with endpoints $0$ and $\omega_c$ intersects the set $\{i+\delta_0 e^{i\varphi}:\varphi\in [-\pi/4,-\pi/12]\}$ for all $0\leq \xi_c-\xi\leq \delta^*_n$ and the section of $\text{asc}_\xi$ starting at $\omega_c$ and ending at infinity travelling via $\HH^+$ intersects the set $\{i+\delta_0 e^{i\varphi}:\varphi\in [\pi/12,\pi/4]\}$.\label{contoursexistlemma}
\end{lemma}
The next lemma gives control over the integral in a neighbourhood of the critical points.
\begin{lemma} Let $\delta^*_n>0$ be a sequence going to zero as $n\rarrow\infty$. 
If $\delta_0>0$ is sufficiently small then for $w_1\in \widetilde{\text{desc}}_{\xi,1/R}\cap\overline{\mathbb{B}}(i,\delta_0)$ and all $n$ large enough, there exists $C>0$ such that
\begin{align}
\mathcal{R}[g_{\xi}(w_1)-g_{\xi}(\omega_c)]\leq -C\max\{|w_1-\omega_c|^3,\sqrt{\xi_c-\xi}|w_1-\omega_c|^2\}
\label{localbound}
\end{align}
for all $0\leq \xi_c-\xi\leq \delta^*_n$.
Also, for $w_2\in \widetilde{\text{asc}}_{\xi}^R\cap\overline{\mathbb{B}}(i,\delta_0)$ and $n$ large enough, 
\begin{align}
\mathcal{R}[g_{\xi}(\omega_c)-g_{\xi}(w_2)]\leq -C\max\{|w_2-\omega_c|^3,\sqrt{\xi_c-\xi}|w_1-\omega_c|^2\}
\label{localbound}
\end{align}
for all $0\leq \xi_c-\xi\leq \delta^*_n$. 
\label{locallemma}
\end{lemma}
The next lemma gives control of the integral over sections of the descent path bounded away from the critical points, and follows from the above lemma.

\begin{lemma}
Let $\delta^*_n>0$ be a sequence that goes to zero as $n\rarrow \infty$, if $\delta_0>0$ is sufficiently small and $w^{*}_{\xi}$ is a point where desc$_{\xi}$ intersects $\partial \B(i,\delta_0)$ in $\HH^+$, $w_\xi^{**}$ be a point where $\text{asc}_\xi$ intersects $\partial \B(i,\delta_0)$ in $\HH^+$. There is a number $C(\delta_0)>0$ such that
\begin{align}
\mathcal{R}[g_{\xi}(w_{\xi}^{*})-g_{\xi}(\omega_c)]<-C(\delta_0)
\label{globalbound}
\end{align}
for all $0\leq\xi_c-\xi\leq\delta^*_n$, when $n$ is large. We also have
\begin{align}
\mathcal{R}[g_{\xi}(\omega_c)-g_{\xi}(w_{\xi}^{**})]<-C(\delta_0)
\end{align} 
for all $0\leq \xi_c-\xi\leq \delta^*_n$, when $n$ is large.
\label{globallemma}
\end{lemma}

Let $\delta_0>0$ be small enough and $n$ large enough such that the above lemmas hold. Define
\begin{align}
&\Gamma_{\text{d}}^\text{in}= \widetilde{\text{desc}}_{\xi,1/R}\cap\mathbb{B}(i,\delta_0)\cap \HH^+, && \Gamma_{\text{d}}^\text{out}= (\widetilde{\text{desc}}_{\xi,1/R}\setminus\mathbb{B}(i,\delta_0))\cap \HH^+, \\
&\Gamma_{\text{a}}^\text{in}= \widetilde{\text{asc}}_{\xi}^R\cap\mathbb{B}(i,\delta_0)\cap \HH^+, && \Gamma_{\text{a}}^\text{out}=( \widetilde{\text{asc}}_{\xi}^R\setminus\mathbb{B}(i,\delta_0))\cap \HH^+.\nonumber
\label{contourdivides}
\end{align}
It is enough to bound the integral in \eqref{startingpoint} restricted to $(\HH^+)^2$,
\begin{align}
\int_{\widetilde{\text{desc}}_{\xi,1/R}\cap \HH^+}\abs{dw_1}\int_{\widetilde{\text{asc}}^R_\xi\cap \HH^+}\abs{dw_2} \abs{K(w_1,w_2)}\abs{\frac{\exp\{\frac{n}{2}(g_{\xi}(w_1)-g_\xi(w_2))\}}{w_2-w_1}}.
\end{align}
In particular write this integral as a sum of integrals over products of sets of the form \eqref{contourdivides}. The bound
\begin{align}
&\int_{\Gamma_{\text{d}}^{\text{out}}}\abs{dw_1}\int_{\widetilde{\text{asc}}^R_\xi\cap \HH^+}\abs{dw_2}|K(w_1,w_2)| \abs{\frac{\exp\{\frac{n}{2}(g_{\xi}(w_1)-g_\xi(w_2))\}}{w_2-w_1}}\label{temp1232}\\&\leq \exp\{-\frac{n}{2}(C(\delta_0)-\mcR[g_\xi(\omega_c)]+C'r)\}\int_{\Gamma_{\text{d}}^{\text{out}}}\abs{dw_1}\int_{\widetilde{\text{asc}}^R_\xi\cap \HH^+}\abs{dw_2} \abs{\frac{1}{w_2-w_1}}
\nonumber
\end{align}
holds by lemma \ref{globallemma} and the trivial bound $\sup_{(\B(0,R)\setminus \B(0,1/R))^2}|K(w_1,w_2)|\leq e^{C'r}$ for some $C'>0$. Since $\mcR[g_\xi(e^{i\theta})]=0$ for any real $\theta$, $\mcR[g_\xi(\omega_c)]=0$. By an application of the Hayman-Wu theorem (see Garnett and Marshall \cite{GM}), the contours are finite length. To see this, recall a version of the Hayman-Wu theorem; If $\varphi:\mathbb{D}\rarrow \C$ is conformal on the unit disc and $L$ is a straight line in $\C$ then length$(\varphi^{-1}(L))\leq 4\pi$. One can take $\xi_c-\xi$ small enough that $\widetilde{\text{desc}}_{\xi,1/R}$ and $\widetilde{\text{asc}}^R_\xi$ do not intersect the branch cuts so we can find a compact set $A\subset \B(0,R+1)\setminus \B(i,\delta_0)$ such that $g_\xi$ is analytic on $A$ and $A$ contains both  $\widetilde{\text{desc}}_{\xi,1/R}$ and $\widetilde{\text{asc}}^R_\xi$ for all $\xi$. Cover $A$ with finitely many open balls that do not intersect the branch cuts. Apply the Hayman-Wu theorem to the straight line $\pm i\mcI[g_\xi(\omega_c)]+\R$ and $g_\xi$ restricted to each open ball. Since the number of balls is finite the contours are finite. The same theorem gives the same bound for each contour and so the contours lengths are bounded uniformly in $\xi,\xi'$. The integrand in \eqref{temp1232} is integrable since the two sections of the contours do not intersect. Hence the integral over $\Gamma_{\text{d}}^{\text{out}}\times\widetilde{\text{asc}}^R_\xi\cap \HH^+$ is $O(e^{-C(\delta_0)n/2+C'r})$. One obtains the same bound for $\Gamma_{\text{a}}^{\text{out}}\times\widetilde{\text{desc}}^R_\xi\cap \HH^+$ in an analogous fashion. The integral over sections of the contours contained in $(\B(i,\delta_0)\cap\HH^+)^2$ is
\begin{align}
\int_{\Gamma_{\text{d}}^{\text{in}}}\abs{dw_1}\int_{\Gamma_{\text{a}}^{\text{in}}}\abs{dw_2} |K(w_1,w_2)|\abs{\frac{\exp\{\frac{n}{2}(g_{\xi}(w_1)-g_\xi(w_2))\}}{w_2-w_1}}.
\end{align}
These contours are straight lines by definition
\begin{align}
\Gamma_{\text{d}}^{\text{in}}=
\{\omega_c+t'e^{i3\pi/5}\}\cup \{\omega_c+s'(w_\xi^*-\omega_c)\}=: \Gamma_{\text{d}}^{\text{in},1}\cup\Gamma_{\text{d}}^{\text{in},2},\\
\Gamma_{\text{a}}^{\text{in}}=\{\omega_c+te^{-i3\pi/5}\}\cup \{\omega_c+s(w_\xi^{**}-\omega_c)\} =: \Gamma_{\text{a}}^{\text{in},1}\cup \Gamma_{\text{a}}^{\text{in},2}.\nonumber
\end{align} 
It is straightforward to see that the lengths of these straight lines are uniformly bounded in $0\leq \xi_c-\xi\leq \delta^*_n$, indeed, their endpoints all lie in a bounded region of $\C$. Define $\gamma^{\text{a},1}(t)=\omega_c+te^{i\theta_{\text{a},1} }$, $\theta_{\text{a},1}=-3\pi/5$ for $t\in[0,\mcR[\omega_c]/\sin(\pi/10)]$, $\gamma^{\text{a},2}(t)=\omega_c+te^{i\theta_{\text{a},2}}$, $\theta_{\text{a},2}=\arg(w_\xi^{**}-\omega_c)$ for $t\in[0,|w_\xi^{**}-\omega_c|]$, $\gamma^{\text{d},1}(s)=\omega_c+se^{i\theta_{\text{d},1}}$, $\theta_{\text{d},1}=3\pi/5$ for $s\in[0,\mcR{\omega_c}/\sin(\pi/10)]$, and finally,  $\gamma^{\text{d},2}(s)=\omega_c+se^{i\theta_{\text{d},2}}$, $\theta_{\text{d},2}=\arg{(w_\xi^*-\omega_c)}$ for $s\in [0,|w_\xi^*-\omega_c|]$. These parametrise the straight lines in $\Gamma_{\text{d}}^{\text{in}}$ and $\Gamma_{\text{a}}^{\text{in}}$. For $i,j=1,2$ consider the bound
\begin{align}
&\label{templocalbound1}\int_{\Gamma_{\text{d}}^{\text{in},i}}\abs{dw_1}\int_{\Gamma_{\text{a}}^{\text{in},j}}\abs{dw_2}|K(w_1,w_2)| \abs{\frac{\exp\{\frac{n}{2}(g_{\xi}(w_1)-g_\xi(w_2))\}}{w_2-w_1}}\\
&\leq\exp\{-\frac{n}{2}(\mcR[g_{\xi}(\omega_c)-g_\xi(\omega_c)])\} \int_{\Gamma_{\text{d}}^{\text{in},i}}\abs{dw_1}\int_{\Gamma_{\text{a}}^{\text{in},j}}\abs{dw_2}|K(w_1,w_2)| \abs{\frac{\exp\{-C\frac{n}{2}(\abs{w_1-\omega_c}^3+\abs{w_2-\omega_c}^3)\}}{w_2-w_1}}\nonumber
\end{align}
which holds by lemma \ref{locallemma}. We can take $\delta_0$ small enough that for $j\in\{1,2\}$
\begin{align}
&|\log(G(w_j))-\log(G(\omega_c))|\leq C''|w_j-\omega_c|,\nonumber
\\&|\log(G(1/w_j))-\log(G(1/\omega_c))|\leq  C''|w_j-\omega_c|\nonumber
\end{align} for $w_j\in\B(i,\delta_0)\cap\HH^+$. From which we gain a $C'''>0$ such that
\begin{align}
|K(w_1,w_2)|\leq |G(\omega_c)|^{b_1-a_1}|G(1/\omega_c)|^{a_2-b_2}e^{C'''r[|w_1-\omega_c|+|w_2-\omega_c|]}.\nonumber
\end{align} Insert this bound into \eqref{templocalbound1}. Parametrising and recalling $\mcR[g_\xi(e^{ix})]=0$ for any real $x$, \eqref{templocalbound1} is bounded by the prefactor $|G(\omega_c)|^{b_1-a_1}|G(1/\omega_c)|^{a_2-b_2}$ multiplied by
\begin{align}
& \int_{\text{Dom}(\gamma^{\text{d},i})}\abs{ds}\int_{\text{Dom}(\gamma^{\text{a},j})}\abs{dt} \abs{\frac{\exp\{-Cn(s^3+t^3)/2+C'''r(s+t)\}}{te^{i\theta_{\text{a},j}}-se^{i\theta_{\text{d},i}}}}\\&=\frac{1}{n^{1/3}}\int_{n^{1/3}\text{Dom}(\gamma^{\text{d},i})}\abs{ds}\int_{n^{1/3}\text{Dom}(\gamma^{\text{a},j})}\abs{dt} \abs{\frac{\exp\{-C(s^3+t^3)/2+C'''rn^{-1/3}(s+t)\}}{te^{i\theta_{\text{a},j}}-se^{i\theta_{\text{d},i}})}}\label{temp13ea}\\
&\leq\frac{1}{n^{1/3}}\int_{0}^\infty ds\int_{0}^\infty dt \exp\{-C(s^3+t^3)/2+C'''(s+t)\}\frac{1}{\abs{te^{i\theta_{\text{a},j}}-se^{i\theta_{\text{d},i}}}}.
\label{temp2314}
\end{align}
when $r\leq n^{1/3}$. 
The double integral in \eqref{temp2314} is bounded above by some positive constant uniformly for $\xi,n$. To prove the $1/\sqrt{n\sqrt{\xi_c-\xi}}$ bound in \eqref{propstate}, we instead use $\mcR[g_\xi(w_i)-g_\xi(\omega_c)]\leq -C\sqrt{\xi_c-\xi}|w_i-\omega_c|^2$ in \eqref{templocalbound1} and then make the change of variables $s\rarrow s/\sqrt{n\sqrt{\xi_c-\xi}}, t\rarrow t/\sqrt{n\sqrt{\xi_c-\xi}}$ instead of $s\rarrow n^{-1/3}s, t\rarrow n^{-1/3}t$ in \eqref{temp13ea}.
\end{proof}

We now turn to the proof of lemmas \ref{contoursexistlemma} to \ref{globallemma}. Let $\delta_0<1-\sqrt{2c}$ so that for $-\frac{1}{2}\sqrt{1+2c}<\xi\leq\xi_c$, $g_\xi(w)$ is analytic in the annulus $\mathcal{L}(\delta_0)=\{z\in \C:1-\delta_0<|z|<1+\delta_0\}$. From Taylor's theorem,
\begin{align}
g_\xi(w)=g_\xi(\omega_c)+g''_\xi(\omega_c)(w-\omega_c)^2/2 + g'''_\xi(\omega_c)(w-\omega_c)^3/3!+R_{\omega_c}(w)(w-\omega_c)^4
\label{gomegactaylor}
\end{align}
for $\xi_c\geq \xi\geq \eps^*>-\frac{1}{2}\sqrt{1+2c}$. Where
\begin{align}
R_{\omega_c}(w)=\frac{1}{2\pi i}\int_{\abs{z-\omega_c}=\delta_1}\frac{g_\xi(z)}{(z-\omega_c)^4(z-w)}dz
\end{align}
for any $0<\delta_1<\delta_0$ and $R_{\omega_c}(w)(w-\omega_c)^4$ is $O(\abs{w-\omega_c}^4)$ uniformly in $\xi$.

We have the equations
\begin{align}
&wg'_\xi(w)=1+\xi\big(\frac{w}{\sqrt{w^2+2c}}+\frac{1}{w\sqrt{1/w^2+2c}}\big),
\label{firstderiv}\\
&wg'_\xi(w)+w^2g''_\xi(w)=2c\xi\big(\frac{w}{(w^2+2c)^{3/2}}-\frac{1}{w(1/w^2+2c)^{3/2}}\big)\label{secderiv},\\
&g'_\xi(w)+3wg''_\xi(w)+w^2g'''_\xi(w)=2c\xi\big(\frac{2(c-w^2)}{(w^2+2c)^{5/2}}+\frac{2}{w^2}\frac{c-1/w^2}{(1/w^2+2c)^{5/2}}\big).
\label{thirderiv}
\end{align}
which follow by taking derivatives. Before we prove lemmas \ref{contoursexistlemma}, \ref{locallemma},  \ref{globallemma} we relate the second and third derivatives of $g_\xi$ at $\omega_c$ to $\xi_c-\xi$.

\begin{lemma}
There are functions $R_{27}(\xi),R_{28}(\xi)$ bounded for some interval $0\leq \xi_c-\xi< c_2$ such that
\begin{align}
g''_\xi(\omega_c)&=-d_2i\sqrt{\xi_c-\xi}+R_{27}(\xi)(\xi_c-\xi)
\label{g''xiomegacthetac}
\end{align}
and
\begin{align}
g'''_\xi(\omega_c)&=-id_3+R_{28}(\xi)(\xi_c-\xi)^{1/2}
\end{align}
where $d_2=4\sqrt{\frac{c(1+c)}{(1-2c)^{5/2}}}$, $d_3=\frac{8c(1+c)\abs{\xi_c}}{(1-2c)^{5/2}}$.
\begin{proof}

From \eqref{secderiv} write
\begin{align}
g''_\xi(\omega_c)&=\omega_c^{-2}4ic\xi\mathcal{I}[\frac{\omega_c}{(\omega_c^2+2c)^{3/2}}]=4ic\xi e^{-2i\theta_c}\mathcal{I}[\frac{e^{i\theta_c}}{((e^{i\theta_c})^2+2c)^{3/2}}].
\end{align}
Define $F_4(\theta_c)=e^{i\theta_c}/{((e^{i\theta_c})^2+2c)^{3/2}}$. For each $\theta_c\in[0,\pi/2]$ there exists $\chi\in[\theta_c,\pi/2]$ such that
\begin{align}
\mathcal{I}[F_4(\theta_c)]=\mathcal{I}[F_4(\pi/2)]-\mathcal{I}[F_4'(\pi/2)]\varphi_c+R_{29}(\xi)\varphi_c^2
\end{align}
with $R_{29}(\xi)=\mathcal{I}[F_4''(\chi)]$. Since $\mathcal{I}[F_4(\pi/2)]=0$,
\begin{align}
g''_\xi(\omega_c)&=4ic\xi[-1+R_5(\xi)\varphi_c][-\mathcal{I}[F_4'(\pi/2)]\varphi_c+R_{29}(\xi)\varphi_c^2]
\end{align}
where $\varphi_cR_{30}(\xi)=e^{-2i\theta_c}+1$, $R_{30}(\xi_c)=2i$ and $|R_{30}(\xi)|\leq 2e^{\pi-2\theta_c}$.
\begin{align}
g''_\xi(\omega_c)&=-4ic\xi_c\mathcal{I}[-F_4'(\pi/2)\varphi_c]+R_{31}(\xi)\varphi_c^2-4ic(\xi_c-\xi)\mathcal{I}[F_4'(\pi/2)\varphi_c]
\label{temp232}
\end{align}
where $R_{31}(\xi)=4ic\xi[\mathcal{I}[R_{29}(\xi)](-1-R_{30}(\xi)\varphi_c)+R_{30}(\xi)F_4'(\pi/2)]$ is bounded.\\
A computation yields 
\begin{align}
F_4'[\theta_c]\Big\rvert_{\theta_c=\pi/2}=ie^{i\theta_c}\frac{2(c-e^{i2\theta_c})}{((e^{i\theta_c})^2+2c)^{5/2}}\Big\rvert_{\theta_c=\pi/2}=i\frac{2(1+c)}{(1-2c)^{5/2}}.
\end{align}
Inserting \eqref{thetacintermsofxi} into \eqref{temp232} we have
\begin{align}
g''_\xi(\omega_c)&=\frac{-8ic(1+c)\abs{\xi_c}}{(1-2c)^{5/2}}\sqrt{\frac{(1-2c)^{5/2}}{4c(1+c)\xi_c^2}}\sqrt{\xi_c-\xi}+R_{27}(\xi)(\xi_c-\xi)
\end{align}
for a bounded function $R_{27}$.

For the third derivative, from \eqref{thirderiv} write
\begin{align}
g_\xi'''(\omega_c)=4c\xi \omega_c^{-3}\mathcal{R}[\frac{2\omega_c(c-\omega_c^2)}{(\omega_c^2+2c)^{5/2})}]-3\omega_c^{-1}g_\xi''(\omega_c).
\end{align}
Define $L_2(\theta_c)=\mathcal{R}[\frac{2e^{i\theta_c}(c-e^{2i\theta_c})}{((e^{i\theta_c})^2+2c)^{5/2})}]$, for every $\theta_c\in [0,\pi/2]$ there is a $\chi\in[\theta_c,\pi/2]$ such that $R_{32}(\xi)=-L_2'(\chi)$ and $L_2(\theta_c)-L_2(\pi/2)=R_{32}(\xi)\varphi_c$. It is clear that $L_2(\pi/2)=\frac{2(1+c)}{(1-2c)^{5/2}}$. 
So now
\begin{align}
g'''_\xi(\omega_c)=4c\xi(i+R_{10}(\xi)\varphi_c)(\frac{2(1+c)}{(1-2c)^{5/2})}+R_{32}(\xi_c)\varphi_c)-3(i+R_{34}(\xi)\varphi_c)g''_\xi(\omega_c)
\end{align} 
where $|R_{33}(\xi)|\leq 3e^{3\pi/2-3\theta_c}$, $|R_{34}(\xi)|\leq e^{\pi/2-\theta_c}$. The lemma follows by the statement on $g''_\xi(\omega_c)$ and lemma \ref{thetacxi}.
\end{proof}
\label{gomegac''}
\end{lemma}

\begin{lemma}
Let $p\exp{(i\theta)}+\omega_c$ $(p>0,\theta\in (-\pi,\pi])$ be a point in the annulus $\mathcal{L}(\delta_0)$ satisfying 
\begin{align}
\mathcal{I}[g_\xi(\omega_c+pe^{i\theta})]=\mathcal{I}[g_\xi(\omega_c)], && \mathcal{R}[g_\xi(\omega_c+pe^{i\theta})]<\mathcal{R}[g_\xi(\omega_c)].
\label{ascdesccond1}
\end{align} 
There are constants $a_1,a_2>0$ sufficiently small such that the following holds. Set $\delta=\delta(\eps) = \min\{\eps a_1,a_2\}$. For every $\eps>0$, if $0\leq\xi_c-\xi\leq\delta^4$ and $\delta/2\leq p\leq\delta$ then
\begin{align}
\abs{\cos(3\theta)}<\eps, && \sin(3\theta)<\eps.
\end{align}
\\
If instead $\mathcal{R}[g_\xi(\omega_c+pe^{i\theta})]>\mathcal{R}[g_\xi(\omega_c)]$, then for every $\eps>0$ there is a $\delta>0$ sufficiently small such that  if $0\leq\xi_c-\xi\leq\delta^4$ and $\delta/2\leq p\leq\delta$ then
\begin{align}
\abs{\cos(3\theta)}<\eps, && \sin(3\theta)>\eps.
\end{align}
\begin{proof}
Let $0<\delta_1<1$ be small enough so that the remainders in equation \eqref{gomegactaylor} and lemma \ref{gomegac''} are bounded for all $0\leq\xi_c-\xi\leq \delta_1^{4}$ and $0\leq p\leq \delta_1$. Insert the Taylor expansion \eqref{gomegactaylor} into the first condition in \eqref{ascdesccond1},
\begin{align}
\mathcal{I}[g_\xi''(\omega_c)e^{2i\theta}+g'''_\xi(\omega_c)\frac{pe^{3i\theta}}{3}]=\mathcal{I}[-R_{\omega_c}(\omega_c+pe^{i\theta})e^{4i\theta}p^2].
\end{align}
Now apply lemma \ref{gomegac''} to the left hand side, a rearrangement gives
\begin{align}
\cos(3\theta)= \ &\frac{3}{d_3}\mathcal{I}\Big[R_{27}(\xi)e^{2i\theta}\frac{\xi_c-\xi}{p}+R_{28}(\xi)e^{3i\theta}\frac{\sqrt{\xi_c-\xi}}{3}\label{temp121332}\\&+R_{\omega_c}(\omega_c+pe^{i\theta})e^{4i\theta}p\Big]-\frac{3d_2\sqrt{\xi_c-\xi}}{d_3p}\cos(2\theta).
\nonumber
\end{align}
The absolute value of the right hand side of \eqref{temp121332} is bounded above by 
\begin{align}
\frac{3d_2\sqrt{\xi_c-\xi}}{d_3p}+\frac{3}{d_3}\abs{R_{\omega_c}(\omega_c+pe^{i\theta})}p+\abs{R_{27}(\xi)}\frac{3(\xi_c-\xi)}{d_3p}+\abs{R_{28}(\xi)}d_3\sqrt{\xi_c-\xi}.
\end{align}
This is bounded by some constant $C_1$ multiplied by $\delta$ when $\delta\leq\delta_1$. Hence we take $\delta=\min(\eps/C_1,\delta_1)$ and $a_1=1/C_1, a_2=\delta_1$. The remaining statements are proved by alterations to the above argument.
\end{proof}
\label{cos3thetalessthaneps}
\end{lemma}

For $s>0$ and a subset $S$ of $\R$ define $\overline{\B}(S,s)=\cup_{x\in S}\overline{\B}(x,s)$. We now prove a lemma from which Lemma \ref{contoursexistlemma} follows.

\begin{lemma}
For every $\eps>0$ sufficiently small there is a $\delta>0$ (with $\delta\rarrow 0$ when $\eps\rarrow 0$) such that desc$_\xi$ intersects $\{i+\eps e^{i\varphi}:\varphi\in \B(-\pi/6,\delta)\}$ for all $0\leq \xi_c-\xi\leq \eps^4$.
\begin{proof}
Let $p>0$, $\theta\in (-\pi,\pi]$ be such that $p\exp{(i\theta)}+\omega_c$ is in the component $\Gamma_\xi$ of desc$_\xi$ with endpoints $0$ and $\omega_c$. By Lemma
\ref{cos3thetalessthaneps} and the fact the $\cos(3\theta)$ is locally invertible near its roots; for every $\eps>0$ sufficiently small $\theta$ is close to one of the points $-5\pi/6,-\pi/6,\pi/2$ for all $0\leq \xi_c-\xi\leq \eps^4$ and $\eps/2\leq p \leq \eps$. Let us discount the cases where $\theta$ is close to $-5\pi/6$ or $\pi/2$. By lemma \ref{thetacxi} there is $C>0$ such that $|\omega_c-i|\leq C \eps^2$. Take $\eps$ so small that $|\omega_c-i|\leq \eps/8<\eps/2\leq p$. Because $\Gamma_\xi$ can not cross the imaginary axis and by basic geometry, $\theta$ is not close to $-5\pi/6$. Since $\Gamma_\xi$ leaves $\omega_c$ at an angle $\theta_c-3\pi/4$ and $\mcR[g_\xi(e^{ix})]=0$ for all real $x$, $\Gamma_\xi$ can not cross the unit circle. So $\theta$ is not close to $\pi/2$. \\
Set $i+\eps e^{i\varphi}=\omega_c+pe^{i\theta}$. We have the inequality $|p-\eps|\leq |pe^{i\theta}-\eps e^{i\varphi}|=|\omega_c-i|$. So there is a constant $C>0$ such that
$|\varphi+\pi/6|\leq |\theta+\pi/6|+C|\omega_c-i|^2\leq |\theta+\pi/6|+C\eps^2$.
\end{proof}
\label{closeto-pi/6}
\end{lemma}

\begin{lemma}  Let $\delta^*_n>0$ be a sequence going to zero as $n\rarrow\infty$. 
If $\delta_0>0$ is sufficiently small then for $w_1\in \widetilde{\text{desc}}_\xi\cap\overline{\mathbb{B}}(i,\delta_0)$ and all $n$ large enough, there exists $C>0$ such that
\begin{align}
\mathcal{R}[g_{\xi}(w_1)-g_{\xi}(\omega_c)]\leq -C\max\{|w_1-\omega_c|^3,\sqrt{\xi_c-\xi}|w_1-\omega_c|^2\}
\end{align}
for all $0\leq \xi_c-\xi\leq \delta^*_n$.
\begin{proof}
We give the proof for $\widetilde{\text{desc}}_\xi$, the case for $\widetilde{\text{asc}}_\xi$ is similar. Set $w_1-\omega_c=pe^{i\theta}$. From lemma \ref{gomegac''} and \eqref{gomegactaylor},
\begin{align}
&\mathcal{R}[g_{\xi}''(\omega_c)(w_1-\omega_c)^2/2!+g_{\xi}'''(\omega_c)(w_1-\omega_c)^3/3!+R_{\omega_c}(w_1)(w_1-\omega_c)^4]\label{taylordisc}\\
&=d_2\sqrt{\xi_c-\xi}\frac{p^2}{2}\sin(2\theta)+d_3\frac{p^3}{3!}\sin(3\theta)+\mathcal{R}[R_{27}(\xi)(\xi_c-\xi)p^2e^{2i\theta}\nonumber\\& \ \ +R_{28}(\xi)(\xi_c-\xi)^{1/2}p^3e^{3i\theta}+R_{\omega_c}(w_1)(w_1-\omega_c)^4].
\nonumber
\end{align}
For $\delta_0$ small enough and $n$ large there is a $C'>0$ such that $|R_{\omega_c}(w_1)|,|R_8(\xi)|,|R_{12}(\xi)|\leq C'$ using these inequalities in \eqref{taylordisc} we obtain
\begin{align}
\mathcal{R}[g_{\xi}(w_1)-g_{\xi}(\omega_c)]\leq d_2\sqrt{\xi_c-\xi}\frac{p^2}{2}(\sin(2\theta)+\frac{2C'}{d_2}\sqrt{\xi_c-\xi})+d_3\frac{p^3}{3!}(\sin(3\theta)+\frac{6C'}{d_3}(\sqrt{\xi_c-\xi}+p))
\end{align}
Next we recall a fact proved in the proof of lemma \ref{closeto-pi/6}, for any $\delta_0$ sufficiently small and $n$ large enough, if $w_\xi^*\in \Gamma_\xi\cap\partial B(i,\delta_0)$ then $\arg(w_\xi^*-\omega_c)\in \B(-\pi/6, \pi/24)$. Recall also, that $\widetilde{\text{desc}}_\xi\cap\B(i,\delta_0)$ consists of two straight lines: One leaving $\omega_c$ at angle $\theta=3\pi/5$, terminating at the imaginary axis and the other having endpoints $\omega_c$, $w_\xi^*$. Hence take $\delta_0$ so small that $6C'\delta_0/d_3<1/8$, and $n$ so large that $\max(2C'\sqrt{\xi_c-\xi}/d_2,6C'/d_3(\sqrt{\xi_c-\xi}+|\omega_c-i|))<1/8$ so
\begin{align}
\mathcal{R}[g_{\xi}(w_1)-g_{\xi}(\omega_c)]&\leq d_2\sqrt{\xi_c-\xi}\frac{p^2}{2}(\sin(2\theta)+1/8)+d_3\frac{p^3}{3!}(\sin(3\theta)+1/4)]\\
&\leq -3d_2\sqrt{\xi_c-\xi}\frac{p^2}{16}-d_3\frac{p^3}{4!}\nonumber
\end{align}
since both $\sin(2\theta),\sin(3\theta)<-1/2$ when $\theta=3\pi/5$ or $\theta\in [-5\pi/24,-3\pi/24]$. Note this proof also yields Lemma \ref{globallemma}.
\end{proof}
\end{lemma}
\section{Some technical lemmas}
Now we give the proof of the lemma relating $\omega_c$ and $\xi$.
\begin{proof}[Proof of lemma \ref{thetacxi}]
\label{proofxiomegac}
By \eqref{firstderiv}, $\xi=-1/(2\mcR[F_1(e^{i\theta_c})] )$ and $\xi_c=-1/(2\mcR[F_1(i)])$ where we define $F_1(w)=w/\sqrt{w^2+2c}$. Let $F_2(\theta_c)=1/\mcR[F_1(e^{i\theta_c})]$, $F_2$ is smooth on $[0,\pi/2]$. For any $\theta_c\in[0,\pi/2]$ there exists $\chi$ between $\pi/2$ and $\theta_c$ such that 
\begin{align}
F_2(\theta_c)-F_2(\pi/2)=-F_2'(\pi/2)\varphi_c+\frac{F_2''(\pi/2)}{2}\varphi_c^2-\frac{F_2'''(\pi/2)}{3!}\varphi_c^3+\frac{F_2''''(\chi)}{4!}\varphi_c^4.
\end{align}
Direct computation gives that $F_2'(\pi/2)=F_2'''(\pi/2)=0$ and that $F_2''(\pi/2)=-4c(1+c)/\sqrt{1-2c}$. \eqref{xiintermsofthetac} follows by defining $R_{35}(\varphi_c)=F_2''''(\chi)$ for the corresponding $\chi$ and the fact that $\xi_c-\xi=(F_2(\pi/2)-F_2(\theta_c))/2$.\\
Define $F_3(w)=F_1(w)+F_1(1/w)$. We are going to find a section of a small neighbourhood around $i$ in which to invert $F_3(w)$. Note $F_3$ is not one-to-one in any neighbourhood of $i$. It is straightforward to compute $F_3(i)=-1/\xi_c$, $F_3'(i)=0$ and $F_3''(i)=-8c(1+c)/(1-2c)^{5/2}$. One sees from the definition of $F_3$ that $F_3$ is analytic in the disc $\B(i,\delta_1)$ where $\delta_1=1-\sqrt{2c}$. Define $H$ such that $F_3(i)-F_3(w)=H(w)(w-i)^2$ so that $H$ is analytic in $\B(i,\delta_1)$ and $H(i)=-F_3''(i)/2$ is a positive real number. Non-zero analytic functions have isolated zeros so we can find $0<\delta_2<\delta_1$ so that $0\notin H(\overline{\B}(i,\delta_2))$ and define a branch of $\sqrt{H(w)}$ so that $\sqrt{H(w)}$ is analytic on $\overline{\B}(i,\delta_2)$. Now define the analytic function $h(w)=(w-i)\sqrt{H(w)}$ on $\overline{\B}(i,\delta_2)$, one sees $h$ has a non-vanishing derivative at $i$. Note $h(w)^2=F_3(i)-F_3(w)$. Let $c_1>0$ be chosen so that $\abs{h(w)}\geq 2c_1$ on $\abs{w-i}=\delta_2$. By Lagrange's Inverse Function theorem (see \cite{Gam}) we have
\begin{align}
h^{-1}(z)=\frac{1}{2\pi i}\int_{|\zeta-i|=\delta_2}\frac{\zeta h'(\zeta)}{h(\zeta)-z}d\zeta, \qquad |z|< c_1.
\end{align}
Now consider $-\frac{1}{2}\sqrt{1+2c}<\xi\leq \xi_c$ such that $c_1^2\geq\frac{1}{\xi}-\frac{1}{\xi_c}\geq0$ and label $w_{\pm}(\xi)=h^{-1}(\pm \sqrt{\frac{1}{\xi}-\frac{1}{\xi_c}})$. Take a sequence $\theta_c^{(n)}\rarrow \pi/2^-$ such that $e^{i\theta_c^{(n)}}\in \overline{B}(i,\delta_2)\cap \HH^+$ for all $n$, by continuity $\sqrt{H(e^{i\theta_c^{(n)}})}\rarrow \sqrt{H(i)}$ which is a positive real number, so $\arg(\sqrt{H(e^{i\theta_c^{(n)}})})\rarrow 0$. Since $[-\pi/4,0)\ni \arg(e^{i\theta_c^{(n)}}-i)\rarrow 0$ one can take $n$ so large that $\mcR[h(e^{i\theta_c^{(n)}})]>0$. This discounts the possibility $w_-=\omega_c$ and so we have $w_+=\omega_c$. Hence we have obtained
\begin{align}
\omega_c(\xi)=\frac{1}{2\pi i}\int_{|\zeta-i|=\delta_2}\frac{\zeta h'(\zeta)}{h(\zeta)}\frac{1}{1-\frac{\sqrt{\frac{1}{\xi}-\frac{1}{\xi_c}}}{h(\zeta)}}d\zeta.
\end{align}
For $|t|<1$, $\frac{1}{1-t}=1+t+\frac{t^2}{1-t}$,
\begin{align}
\omega_c(\xi)=\frac{1}{2\pi i}\int_{|\zeta-i|=\delta_2}&\frac{\zeta h'(\zeta)}{h(\zeta)}d\zeta+\sqrt{\frac{1}{\xi}-\frac{1}{\xi_c}}\frac{1}{2\pi i}\int_{|\zeta-i|=\delta_2}\frac{\zeta h'(\zeta)}{h(\zeta)^2}d\zeta\label{omegacintegral}\\&+\big(\frac{1}{\xi}-\frac{1}{\xi_c}\big)\frac{1}{2\pi i}\int_{|\zeta-i|=\delta_2}\frac{\zeta h'(\zeta)}{h(\zeta)^3}\frac{1}{1-\frac{\sqrt{\frac{1}{\xi}-\frac{1}{\xi_c}}}{h(\zeta)}}d\zeta.
\nonumber
\end{align}
Since $h$ on $ \overline{\B}(i,\delta_2)$ only has a simple zero at $\zeta=i$ the first term on the right hand side of \eqref{omegacintegral} is $i$ by Cauchy's residue theorem. In the second term of \eqref{omegacintegral} we have 
\begin{align}
\frac{1}{2\pi i}\int_{|\zeta-i|=\delta_2}\frac{\zeta h'(\zeta)}{h(\zeta)^2}d\zeta&=\frac{1}{2\pi i}\int_{|\zeta-i|=\delta_2}\zeta(\frac{ H'(\zeta)}{2\sqrt{H(\zeta)}^3(\zeta-i)}+\frac{1}{\sqrt{H(\zeta)}(\zeta-i)^2})d\zeta\nonumber\\
&=i\frac{H'(i)}{2\sqrt{H(i)}^3}+\lim_{\zeta\rarrow i}\frac{d}{d\zeta}\frac{\zeta}{\sqrt{H(\zeta)}}\nonumber\\
&=\frac{1}{\sqrt{H(i)}}\nonumber\\
&=\frac{1}{2}\sqrt{\frac{(1-2c)^{5/2}}{c(1+c)}}.\nonumber
\end{align}
By the mean value theorem for every $-\frac{1}{2}\sqrt{1+2c}\leq\xi\leq \xi_c$ there exists $\eta_\xi\in(0,(\xi-\xi_c)/\xi_c)$ such that $\sqrt{\frac{1}{\xi}-\frac{1}{\xi_c}}=\frac{\sqrt{\xi_c-\xi}}{|\xi_c|}(1+\frac{(\xi_c-\xi)/\xi_c}{2(1+\eta_\xi)^{3/2}})$. \eqref{omegacintermsofxi} follows by letting 
\begin{align}
R_2(\xi)=\frac{1}{\xi_c\xi 2\pi i}\int_{|\zeta-i|=\delta_2}\frac{\zeta h'(\zeta)}{h(\zeta)^3}\frac{1}{1-\frac{\sqrt{\frac{1}{\xi}-\frac{1}{\xi_c}}}{h(\zeta)}}d\zeta-\frac{1}{2}\sqrt{\frac{(\xi_c-\xi)(1-2c)^{5/2}}{2(1+\eta_\xi)^3\xi_c^2c(1+c)}}.
\end{align}
Since $e^{-i\varphi_c}=-i(\omega_c-i)+1$, we can take the principle branch $\log$ so
\begin{align}
\varphi_c&=i\log(1-i(\omega_c-i))=\omega_c-i-\frac{(\omega_c-i)^2}{2\pi}\int_{|z|=\delta_1}\frac{\log(1+z)}{w^2(w+i(\omega_c-i))}dz,
\label{temp242}
\end{align}
\eqref{thetacintermsofxi} follows from substituting \eqref{omegacintermsofxi} in to \eqref{temp242}.
\end{proof}

Now we give some lemmas from which we obtain the leading order terms in proposition \ref{propeklcomegac}. 
\begin{lemma}\label{templeextendsinexp1dx}
Let non-zero integers $\ell,k$ be such that $\ell+k>0$ is even and $\alpha=k/\ell$ lies in a compact subset of $[-1,1)$. There is bounded function $R_{35}(\xi,k,\ell)$ such that
\begin{align}
e^{\frac{(\ell+k)c}{(1-2c)^{3/2}}\varphi_c^2}&\int_{\varphi_c}^\infty \cos((\ell-k)\arg G(ie^{-i\theta}))e^{-\frac{(\ell+k)c}{(1-2c)^{3/2}}\theta^2}d\theta\label{templeextendsinexp}
\\&=-i^{\ell-k}\frac{\sin((\ell-k)F(\varphi_c))}{(\ell-k)/\sqrt{1-2c}}
+R_{35}(\xi,k,\ell)\Big(\frac{\sqrt{\ell+k}}{(\ell-k)^2}+\frac{\varphi_c(\ell+k)}{(\ell-k)^2}\Big).\nonumber
\end{align} 

If instead $\alpha=1+\kappa_{\ell}\in[-1,1]$ such that $\ell\kappa_{\ell}$ is bounded then there is a bounded function $R_{38}$ such that
\begin{align}
&e^{(\ell+k)c'\varphi_c^2}\int_{\varphi_c}^\infty\cos((\ell-k)\arg G(ie^{-i\theta}))e^{-(\ell+k)c'\theta^2}d\theta\label{temp6h66c}\\
&=i^{\ell-k}\frac{D_-(\sqrt{(\ell+k)c'}\varphi_c)}{\sqrt{(\ell+k)c'}}+\Big(\frac{1}{(\ell+k)^{3/2}}+\frac{\varphi_c}{\ell+k}\Big)R_{38}(\xi,k,\ell).\nonumber
\end{align}
\begin{proof}
Recall \eqref{psithetaalpha} and write the integral in \eqref{templeextendsinexp} as
\begin{align}
e^{\frac{(\ell+k)c}{(1-2c)^{3/2}}\varphi_c^2}\int_{\pi/2-\theta_c}^\infty \cos((\ell-k)\text{arg}G(ie^{-i\theta}))e^{-\frac{(\ell+k)c}{(1-2c)^{3/2}}\theta^2}d\theta\label{temp6z4b}\\=
e^{(\ell+k)c'\varphi_c^2}\int_{\varphi_c}^\infty \cos(\ell\mcI [\psi(\pi/2-\theta)])e^{-(\ell+k)c'\theta^2}d\theta.\nonumber
\end{align}
We use integration by parts to get \eqref{temp6z4b} equal to
\begin{align}
-\frac{f(\varphi_c)}{\ell-k}\sin(\ell\mcI[\psi(\pi/2-\varphi_c)])-\frac{e^{(\ell+k)c'\varphi_c^2}}{\ell-k}\int_{\varphi_c}^\infty \sin(\ell\mcI[\psi(\pi/2-\theta)])\frac{d}{d\theta}\Big[e^{-c'(\ell+k)\theta^2}f(\theta)\Big]d\theta
\end{align}
 where 
 \begin{align}
 f(\theta)=\Big (\frac{d}{d\theta}\text{arg}(G(ie^{-i\theta}))\Big)^{-1}.
  \end{align}
  Note: Want to show that if $G(e^{i\theta})$ is real and negative then $\theta=0$.
  \\
 Except for the discontinuities at $\theta=2\pi n, n\in \Z$ caused by $\arg$, one can check that $\text{arg}(G(e^{i\theta}))$ has negative derivative bounded above by some negative number for all $\theta$ and so $f$ is smooth on $\R\setminus \{2\pi n\}$. Indeed, from \eqref{psi'1} we can compute
\begin{align}
\frac{d}{d\theta}\text{arg}(G(e^{i\theta}))&=\mcI\Big[\frac{-ie^{i\theta}}{\sqrt{(e^{i\theta})^2+2c}}\Big]\\
&=\frac{-1}{\pi}\int_{-\sqrt{2c}}^{\sqrt{2c}}\mcR\Big [\frac{e^{i\theta}}{e^{i\theta}-is}\Big ]\frac{ds}{\sqrt{2c-s^2}}\nonumber\\
&=\frac{-2}{\pi}\int_0^{\sqrt{2c}}\frac{1+s^2\cos(2\theta)}{|e^{i\theta}+is|^2|e^{i\theta}-is|^2}\frac{ds}{\sqrt{2c-s^2}}\nonumber
\end{align}
where the integrand is bounded below by some positive number.
 Integrating by parts again we have \eqref{temp6z4b} equal to
 \begin{align}
- \frac{f(\varphi_c)}{\ell-k}\sin(\ell\mcI[\psi(\pi/2-\varphi_c)])+\left .\frac{e^{c'(\ell+k)\varphi_c^2}}{(\ell-k)^2}\right |^\infty_{\varphi_c} \cos(\ell\mcI[\psi(\pi/2-\theta)])f(\theta)\frac{d}{d\theta}\Big [e^{-c'(\ell+k)\theta^2}f(\theta)\Big ]\label{temp7zc3}\\
 +\frac{e^{c'(\ell+k)\varphi_c^2}}{(\ell-k)^2}\int_{\varphi_c}^\infty \cos(\ell\mcI[\psi(\pi/2-\theta)])\frac{d}{d\theta}\Big[ f(\theta) \frac{d}{d\theta}\Big [e^{-c'(\ell+k)\theta^2}f(\theta)\Big ]\Big]d\theta.\nonumber
 \end{align}
 From \eqref{tempImpsi12}
 \begin{align}
\ell\mcI[\psi(\pi/2-\theta)]=(\ell-k)\pi/2+\frac{\ell-k}{\sqrt{1-2c}}\theta-(\ell-k)R_{7}(-\theta)\theta^3.
 \end{align}
 Together with Taylors theorem applied to $f$, for each $\varphi_c>0$ we get a  $\chi\in (0,\varphi_c)$ such that
 \begin{align}
\frac{f(\varphi_c)}{\ell-k}\sin(\ell\mcI[\psi(\pi/2-\varphi_c)])=\frac{i^{\ell-k}\sin(\frac{\ell-k}{\sqrt{1-2c}}\varphi_c-(\ell-k)R_{7}(-\varphi_c)\varphi_c^3)}{\ell-k}\Big (-\sqrt{1-2c}+f'(\chi)\varphi_c\Big)\label{temp54gh}
 \end{align}
 where we used $f(0)=i(1-\alpha)/\psi'(\pi/2)$ and the values of \eqref{psievals1} and \eqref{psievals2} computed earlier.
By taking a derivative there is a $C>0$ such that
\begin{align}
\abs{f(\theta)\frac{d}{d\theta}\Big [e^{-c'(\ell+k)\theta^2}f(\theta)\Big ]}\leq Ce^{-c'(\ell+k)\theta^2}(1+(\ell+k)\theta)\label{temp32va}
\end{align}
since $f$ is smooth.
Similarly,
\begin{align}
\abs{\frac{d}{d\theta}\Big [ f(\theta)\frac{d}{d\theta}\Big [e^{-c'(\ell+k)\theta^2}f(\theta)\Big ]\Big ]}\leq Ce^{-c'(\ell+k)^2}((\ell+k)+(\ell+k)\theta+(\ell+k)^2\theta^2).\label{temp43vas}
\end{align}
 So by \eqref{temp32va} the second term in \eqref{temp7zc3} is bounded above by
\begin{align}
C\frac{1+(\ell+k)\varphi_c}{(\ell-k)^2}.\label{temp3gaf}
\end{align}
Then if we use \eqref{temp43vas} in the third term \eqref{temp7zc3} and integrate by parts, the third term in \eqref{temp7zc3} is bounded above by
\begin{align}
&\label{temp54fa}C\frac{e^{c'(\ell+k)\varphi_c^2}}{(\ell-k)^2}\Big ((\ell+k)^2\Big[\int_{\varphi_c}^\infty \frac{e^{-c'(\ell+k)\theta^2}}{2(\ell+k)}d\theta+\frac{\varphi_c e^{-c'(\ell+k)\varphi_c^2}
}{\ell+k}\Big ] \\&\quad+(\ell+k)\frac{e^{-c'(\ell+k)\varphi_c^2}}{2(\ell+k)}+(\ell+k)\int_{\varphi_c}^\infty e^{-c'(\ell+k)\theta^2}d\theta\Big)\nonumber\\
&\leq C'\Big(\frac{\sqrt{\ell+k}}{(\ell-k)^2}+\frac{1}{(\ell-k)^2}+\frac{\varphi_c(\ell+k)}{(\ell-k)^2}\Big )\nonumber
\end{align} 
where in the last line we used the fact that there is a $C'>0$ such that
\begin{align}
\int_z^\infty e^{-x^2}dx\leq C'e^{-z^2}
\end{align}
for $z>0$. 
Inserting \eqref{temp54gh} into \eqref{temp7zc3} the two bounds \eqref{temp3gaf} and \eqref{temp54fa} can be used to give \eqref{templeextendsinexp}.

Now we prove \eqref{temp6h66c}. We have 
\begin{align}
\cos(\ell\kappa_\ell\theta h(\theta))=1+\theta^2R_{39}(\theta,\kappa_\ell,\ell)\label{temp7dc4}
\end{align} 
 for a smooth $h$ on $\R$ and bounded $R_{39}$.
Recall \eqref{Ftheta}, we have that
\begin{align}
&e^{(\ell+k)c'\varphi_c^2}\abs{\int_{\varphi_c}^\infty\cos(\ell\kappa_\ell \theta F(\theta))e^{-(\ell+k)c'\theta^2}d\theta-\int_{\varphi_c}^\infty e^{-(\ell+k)c'\varphi_c^2}d\theta}\\ 
&\leq C_1e^{(\ell+k)c'\varphi_c^2}\int_{\varphi_c}^\infty \theta^2e^{-c'(\ell+k)\theta^2}d\theta\leq C_2\frac{D_-(\varphi_c\sqrt{(\ell+k)c'})}{(\ell+k)^{3/2}}+C_3\frac{\varphi_c}{\ell+k}\nonumber
\end{align}
where the last inequality follows since
\begin{align}
\int_z^\infty x^2e^{-x^2}dx=\frac{1}{2}e^{-z^2}(D_-(z)+z).
\end{align}
Now recall the bound $D_-(z)<1/(z+1)$ from \eqref{millest}.
\end{proof}
\end{lemma}

\begin{lemma}\label{templemmalpha13s}
Let $\ell,k$ be non-zero integers such that $\ell+k<0$ and $\alpha=k/\ell=1+\kappa_\ell$ for some $\kappa_\ell=O(1/|\ell|)$. There is a bounded function $R_{40}(\xi,\kappa_\ell,\ell)$ such that
\begin{align}
&e^{(\ell+k)c'\varphi_c^2}\int_0^{\varphi_c}\cos((l-k)F(\theta))e^{-(\ell+k)c'\theta^2}d\theta\\&=\frac{D_+(\varphi_c\sqrt{-(\ell+k)c'})}{\sqrt{-(\ell+k)c'}}+\frac{\varphi_c^2}{\sqrt{|\ell+k|}}\min(\varphi_c\sqrt{c'|\ell+k|},\frac{1}{\varphi_c\sqrt{c'|\ell+k|}})R_{40}(\xi,\kappa_\ell,\ell)
\end{align}
\begin{proof}
We have \eqref{temp7dc4} so
\begin{align}
\abs{\int_0^{\varphi_c}\cos(\ell\kappa_{\ell}F(\theta))e^{-(\ell+k)c'\theta^2}d\theta-\int_0^{\varphi_c}e^{-(\ell+k)c'\theta^2}d\theta}\leq C\frac{\varphi_c^2}{\sqrt{|\ell+k|}}\int_0^{\varphi_c\sqrt{c'|\ell+k|}}e^{\theta^2}d\theta
\end{align}
Now use the pair of inequalities
 \begin{align}
\int_0^z e^{x^2}dx\leq ze^{z^2+1}, && \int_0^z e^{x^2}dx\leq e^{z^2+1}/z\label{pairineqs}
\end{align}
for $z>0$. 
\end{proof}
\end{lemma}


\end{document}